\documentclass[11pt]{article}%
\usepackage{amsfonts}
\usepackage{amsmath,bbm,dsfont,mathrsfs,mathtools,appendix}
\usepackage{amssymb}
\usepackage{graphicx}
\usepackage{fullpage}%
\usepackage[colorlinks=true,linkcolor=blue,citecolor=red,plainpages=false,pdfpagelabels]%
{hyperref}%
\providecommand{\U}[1]{\protect\rule{.1in}{.1in}}

\newtheorem{theorem}{Theorem}

\newtheorem{definition}[theorem]{Definition}
\newtheorem{example}[theorem]{Example}

\newtheorem{lemma}[theorem]{Lemma}

\newtheorem{proposition}[theorem]{Proposition}
\newtheorem{remark}[theorem]{Remark}

\newenvironment{proof}[1][Proof]{\noindent\textbf{#1.} }{\ \rule{0.5em}{0.5em}}

\def\mV{\mathcal{V}}
\newcommand{\cF}{{\mathcal{F}}}
\newcommand{\cB}{{\mathcal{B}}}

\newcommand{\cH}{\mathcal{H}}

\newcommand{\N}{\mathbbm{N}}

\newcommand{\nn}{\nonumber}

\newcommand{\RA}{R}
\newcommand{\RC}{R^\star}

\newcommand{\cP}{\mathcal{P}}

\newcommand{\cS}{\mathcal{S}}
\newcommand{\cA}{\mathcal{A}}
\newcommand{\cE}{\mathcal{E}}
\newcommand{\D}{\mathcal{D}}
\newcommand{\cX}{{\cal X}}

\newcommand{\cM}{\mathcal{M}}
\newcommand{\cN}{\mathcal{N}}
\newcommand{\cY}{\mathcal{Y}}
\newcommand{\1}{\mathbbm{1}}
\newcommand{\cD}{\mathcal{D}}

\def\>{{\rangle}}
\def\<{{\langle}}
\newcommand{\be}{\begin{equation}}
\newcommand{\ee}{\end{equation}}
\newcommand{\bea}{\begin{eqnarray}}
\newcommand{\eea}{\end{eqnarray}}

\newcommand{\eps}{\varepsilon}

\def\mX{\mathcal{X}}

\newcommand{\ket}[1]{|#1\rangle} %ket
\newcommand{\bra}[1]{\langle#1|} %bra
\newcommand{\kb}[1]{|#1\rangle\!\langle#1|} %ketbra
\newcommand{\Tr}{\mathrm{Tr}}

\newcommand{\comment}[1]{}

\newcommand{\lr}{\rangle\!\langle}

\newcommand{\ra}{\rangle}

\newcommand{\cptp}{{\rm CPTP}}
\newcommand{\FO}{{\rm FO}}

\newcommand{\cds}{{\rm CDS}}
\newcommand{\rel}{\operatorname{rel}}
\newcommand{\err}{\operatorname{err}}

\def\mf{\mathfrak{F}}

\def\md{\mathfrak{D}}

\newcommand{\mR}{\mathcal{R}}

\def\mE{\mathcal{E}}

\def\mN{\mathcal{N}}

\def\mP{\mathcal{P}}

\def\mD{\mathcal{D}}

\def\mV{\mathcal{V}}

\newcommand{\eqdef}{\coloneqq}
\newcommand{\tr}{{\rm Tr}}

\numberwithin{equation}{section}
\numberwithin{theorem}{section}

\usepackage{color}
\definecolor{colorthree}{rgb}{0.01,0.51,0.93}

\allowdisplaybreaks

\begin{document}

\title{\textbf{Symmetric distinguishability as a quantum resource}}
{\author{\normalsize Robert Salzmann\thanks{University of Cambridge, Department of Applied
Mathematics and Theoretical Physics, Wilberforce Road, Cambridge CB3 0WA,
United Kingdom}
\and \normalsize Nilanjana Datta\footnotemark[1]
\and \normalsize Gilad Gour\thanks{Department of Mathematics and Statistics, University of
Calgary, Alberta, Canada T2N 1N4} \textsuperscript{,}\thanks{Institute for
Quantum Science and Technology, University of Calgary, Alberta, Canada T2N
1N4}
\and \normalsize  Xin Wang\thanks{Institute for Quantum Computing, Baidu Research, Beijing
100193, China}
\and\normalsize  Mark M.~Wilde\thanks{Hearne Institute for Theoretical Physics, Department of
Physics and Astronomy, and Center for Computation and Technology, Louisiana
State University, Baton Rouge, Louisiana 70803, USA}
\textsuperscript{ ,}\thanks{Stanford Institute for Theoretical Physics,
Stanford University, Stanford, California 94305, USA}}}
\maketitle

\begin{abstract}
We develop a resource theory of symmetric distinguishability, the fundamental objects of which are elementary quantum information sources, i.e., sources that
emit one of two possible quantum states with given prior probabilities. Such a source can be represented by a classical-quantum state of a composite system $XA$, corresponding to an ensemble of two
quantum states, with $X$ being classical and $A$ being quantum. We study the resource theory for two different classes of free operations: $(i)$ ${\rm{CPTP}}_A$, which consists of quantum channels 
acting only on $A$, and $(ii)$ conditional doubly stochastic (CDS) maps acting on $XA$. We introduce the notion of symmetric distinguishability of an elementary source 
and prove that it is
a monotone under both these classes of free operations. We study the tasks of distillation and dilution of symmetric distinguishability, both in the one-shot and asymptotic regimes. We prove that 
in the asymptotic regime, the optimal rate of  converting one elementary source to another is equal to the ratio of their quantum Chernoff divergences, under both these classes of free operations. This imparts a new operational interpretation to the quantum Chernoff divergence. We also obtain interesting operational interpretations of
the Thompson metric, in the context of the dilution of symmetric distinguishability. 

%journals to consider: IEEE Trans. Inf. Theory, New Journal of Physics (will be free in the resource theory of open access publications)

\end{abstract}

\tableofcontents

\section{Introduction}
\label{sec:intro}

Distinguishability plays a central role in all of modern science. The ability to distinguish one possibility from another allows for
making inferences from experimental data and making decisions based on
these inferences or developing new theories. Thus, it is essential to
understand distinguishability from a fundamental perspective. Furthermore,
distinguishability is a resource, in the sense that fewer trials of an
experiment are needed to arrive at conclusions when two different
possibilities are more distinguishable from one another.

In this paper, we adopt a resource-theoretic approach to distinguishability in
quantum mechanics that ultimately is helpful in and enriches our fundamental
understanding of distinguishability. We note here that, more generally, the
resource-theoretic approach to quantum information processing \cite{CG18} has
illuminated not only quantum information science but also other areas of
research in physics and mathematical statistics. Our work differs from prior
developments with a related motivation~\cite{Matsumoto10,Wang2019states,Wang2019channels}, in that
here we focus instead on what we call symmetric distinguishability, or
alternatively, the Bayesian approach to distinguishability. The outcome of our
efforts is a resource theory with a plethora of appealing features, including {\em{asymptotic reversibility}} 
with the optimal conversion rate being given by a ratio of Chernoff divergences.
We explain these concepts in
more detail in what follows. Our theory is similar in spirit to that proposed
previously in \cite{Morris09}, but there are some notable differences and our
conclusions are arguably  stronger than those presented in \cite{Morris09}. We refer to the resource theory
that we propose here as the
resource theory of symmetric distinguishability (RTSD).

\subsection{Overview of the resource theory of symmetric distinguishability}

The basic objects of this resource theory are elementary quantum information sources, which  emit one of 
two quantum states with certain prior probabilities (i.e., the source emits a state $\rho_0$ with probability $p$ or a state $\rho_1$ with probability $1-p$). Such a source 
can be represented by the following 
classical--quantum (c-q)\ state:
\begin{equation}
\label{cq-state}
\rho_{XA}\coloneqq p|0\rangle\!\langle0|\otimes\rho_{0}+\left(  1-p\right)
|1\rangle\!\langle 1|\otimes\rho_{1},
\end{equation}
where $p\in\left[  0,1\right]  $ is a prior probability and $\rho_{0}$ and
$\rho_{1}$ are quantum states\footnote{Throughout this paper we restrict attention for the most part to c-q states of the form in \eqref{cq-state}.}. Note that the classical system is equivalently specified by a random variable that we also denote as $X$. In analogy with the notation used in the resource theory of asymmetric distinguishability~\cite{Wang2019states,Wang2019channels},
it can be equivalently represented by a {\em{quantum box}} given by the triple $(p, \rho_0, \rho_1)$. The nomenclature ``box'' is used here to indicate that a quantum system is prepared in the state $\rho_0$ with probability $p$ and $\rho_1$ with probability $1-p$ and it is not known which is the case (thus, the system is analogous to an unopened box).

An important goal of the resource theory
is to transform 
a state of the above form to the following state 
\begin{equation}
\sigma_{XB}\coloneqq q|0\rangle\!\langle0|\otimes\sigma
_{0}+\left(  1-q\right)  |1\rangle\!\langle1|\otimes\sigma
_{1},\label{eq:target-state}%
\end{equation}
via a chosen set of free operations,
where $q\in\left[  0,1\right]  $ and $\sigma_{0}$ and
$\sigma_{1}$ are quantum states. Note that the target system
$B$\ need not be isomorphic to the initial system $A$. This 
corresponds to the following transformation
between boxes: $(p, \rho_0, \rho_1) \mapsto (q, \sigma_0, \sigma_1)$. Note that in what follows, we often
suppress the subscripts denoting the quantum systems, for notational simplicity.

Given such an elementary quantum source $\rho_{XA}\equiv (p, \rho_0, \rho_1)$, a natural way to study distinguishability is to consider the binary hypothesis testing task of discriminating between the states $\rho_0$ and $\rho_1$. There are two possible errors that can be incurred in the process, namely, the \emph{type~I error} (mistaking $\rho_0$ to be $\rho_1$) and the \emph{type~II error} (mistaking $\rho_1$ to be $\rho_0$). In the setting of asymmetric hypothesis testing, one minimizes the type~II error probability under the constraint that the type~I error probability is below a given threshold. In contrast, in symmetric hypothesis testing, the two error probabilities are considered on the same footing and weighted by the prior distribution. 
 The  latter (also known as Bayesian discrimination) is arguably the first problem
ever considered in the field of quantum information theory and solved in the single-copy case by Helstrom~\cite{Hel67,Hel69} and Holevo~\cite{Hol72}.
The operational quantity in this task is the minimum (average) error probability, which we formally define in \eqref{op-err} and denote as $p_{\operatorname{err}}(\rho_{XA})$.

Let $\rho_{XA}^{(n)}$ denote the c-q state corresponding to a source that emits the state $\rho_0^{\otimes n}$ (resp.~$\rho_1^{\otimes n}$) with probability $p$ (resp.~$(1-p)$). 
It is known that $p_{\operatorname{err}}(\rho_{XA}^{(n)})$ decays exponentially in $n$, with exponent given by the Chernoff divergence
(also known as the quantum Chernoff bound) of the states $\rho_0$ and~$\rho_1$ \cite{nussbaum2009chernoff,ACMBMAV07}. Just as the resource theory of asymmetric distinguishability (RTAD) provides a resource-theoretic
perspective to asymmetric hypothesis testing~\cite{Matsumoto10, Matsumoto11, Wang2019states}, our resource theory (RTSD) provides a resource-theoretic framework for symmetric hypothesis testing.
By the quantum Stein's lemma, the relevant operational quantity in asymmetric hypothesis testing is known to be characterized by the quantum relative 
entropy~\cite{HP91,ON00}, and in the RTAD, the optimal rate of transformation between quantum boxes was proved to be given by a 
ratio of quantum relative entropies~\cite{Wang2019states} (see also \cite{BST19} in this context). In analogy and given the result of \cite{ACMBMAV07}, it is natural to expect that, in the RTSD, the corresponding optimal asymptotic rate of transformation between quantum boxes
is given by a ratio of Chernoff divergences. It is pleasing to see that this is indeed the case. In contrast, in \cite{Morris09}, only one-shot transformations are considered and hence, in contrast to our work, asymptotic transformations are not studied.

We consider the RTSD for two different choices of free operations: $(i)$ local quantum
channels (i.e., linear, completely positive, trace-preserving maps) acting on the system $A$ alone,
and $(ii)$ the more general class of 
conditional doubly stochastic (CDS) maps.  We denote the former class of free operations by ${\rm CPTP}_A$. 
A CDS map acting on a c-q state $\rho_{XA}$ defined through \eqref{cq-state} consists
of quantum operations acting on the system $A$ and associated permutations of the letters $x \in {\mathcal{X}}$. A detailed justification behind the choice of CDS maps as free operations is given in Section~\ref{sec:axioms}, where the RTSD is introduced via an axiomatic approach.

In any quantum resource theory, there are several pertinent questions to
address. What are the conditions for the feasibility of transforming a source
state to a target state? If one cannot perform a transformation exactly, how
well can one do so approximately? What is an appropriate measure for
approximation when converting a source state to a target state? Is there a
\textquotedblleft golden unit\textquotedblright\ resource that one can go
through as an intermediate step when converting a source state to a target
state?\ At what rate can one convert repetitions (i.e., multiple copies) of a source state to
repetitions of a target state, either exactly or approximately? More
specifically, at what rate can one distill repetitions of a source state to
the golden unit resource, either exactly or approximately? Conversely, at what
rate can one dilute the golden unit resource to repetitions of a target state?
Is the resource theory asymptotically reversible? In this paper, we address all of these questions within the context of the RTSD. 

Given an elementary quantum source $\rho_{XA}$, a natural measure of symmetric distinguishability is given by the minimum error 
probability $p_{\operatorname{err}}(\rho_{XA})$ in the context of Bayesian state discrimination (mentioned above).
Then
\begin{equation}
{\rm{SD}}(\rho_{XA})\coloneqq -\log \left(2p_{\operatorname{err}}(\rho_{XA})\right)    
\label{eq:def-SD-measure}
\end{equation}
is a natural measure of the symmetric distinguishability (SD)
contained in $\rho_{XA}$.\footnote{Note, however, that this measure is not unique and it is possible to define other measures. In this paper, logarithms are taken to base $2$.} A justification of this choice arises from the consideration of {\em{free states}} and {\em{infinite-resource states}}. For a detailed discussion of the notion of symmetric distinguishability in a more general setting, see Section~\ref{sec:axioms}. A natural choice of a {\em{free state}} for this resource theory is a c-q state of the form~\eqref{cq-state}, for which $p=1/2$ and $\rho_0$ and $\rho_1$ are identical and hence indistinguishable. For such a state, $p_{\operatorname{err}}(\rho_{XA})=1/2$, which is achieved by random guessing,  and hence
${\rm{SD}}(\rho_{XA}) =0$. Note that the converse is true also (i.e., ${\rm{SD}}(\rho_{XA}) =0$ implies that $\rho_0$ and $\rho_1$ are identical and $p=1/2$). Hence our choice of SD respects the requirement that a state has zero SD if and only if it is free. On the other hand, a natural choice of an {\em{infinite-resource state}} is a c-q state of the form~\eqref{cq-state}, for which $\rho_0$ and $\rho_1$ have mutually orthogonal supports. This corresponds to an elementary quantum information source that emits perfectly distinguishable states. For such a state, $p_{\operatorname{err}}(\rho_{XA})=0$ and hence 
${\rm{SD}}(\rho_{XA})= +\infty$. Thus our choice of SD validates the identification of such states as infinite-resource states.

An infinite-resource state has the desirable property that it can be converted to
any other c-q state via CDS maps. Under ${\rm {CPTP}}_A$, it can be transformed to any other c-q state with the same prior. Moreover, any given c-q state cannot be transformed to an 
infinite-resource state unless it is itself an infinite-resource state.

We coin the word {\em{SD-bit}} to refer to the unit of symmetric distinguishability, the basic currency of this resource theory. In fact, for every positive 
real number $m \ge 0$, it is useful to identify a family of c-q states that have $m$ SD-bits. In Definition~\ref{golden-unit} of Section~\ref{sec:prereq}, we consider
a natural choice for such a family of c-q states. They are parametrized by $M \equiv 2^m$ %and its prior probability $q$,
and denoted as $\gamma_{XQ}^{(M)}$. 
Such a state, which we call an $M$-golden unit, has the following key property: $p_{\operatorname{err}}(\gamma_{XQ}^{(M)})=\frac{1}{2M}$ and hence ${\rm{SD}}(\gamma_{XQ}^{(M)})=\log M = m$.

Consideration of an $M$-golden unit also leads naturally
to a clear definition of the fundamental tasks of distillation and dilution in the RTSD as the conversions of a given c-q state to and from, respectively, an $M$-golden unit under
free operations\footnote{
\cite{Morris09} also considers transformations to and from a certain unit of resource. However, instead of the $M$-golden unit that we consider, there the unit of resource is a \emph{dbit}, which is a pair of orthogonal states and therefore corresponds to the case $M=\infty$ of our golden unit. %Hence, in this case the analysis is much simpler as any pair of states can be distilled from a dbit whereas a dbit can be diluted from a pair of states if and only if the corresponding states are mutually orthogonal themselves. 
The advantage of our choice ($M$-golden unit) is that it yields a more refined analysis also in the case of finite resources.}. The previously mentioned task of transformations between two arbitrary c-q states under free operations can be achieved by first distilling an $M$-golden unit (for the maximal possible $M$)
from the initial state and then diluting the distilled $M$-golden unit to the desired target state. This is discussed in detail in Sections~\ref{sec:distil} and \ref{sec:SD-dilution}.

One fundamental setting of interest is the one-shot setting, with the question being to determine the minimum error in converting an initial source to a target source.
We prove that this minimum error can be calculated by means of a semi-definite program. Thus,
we can efficiently calculate this error in time polynomial in the
dimensions of the $A$ and $A^{\prime}$ systems.

%We take the following state as the golden unit resource state for the resource
%theory of symmetric distinguishability:%
%\begin{equation}
%\gamma_{XQ}^{(M,q)}\coloneqq q|0\rangle\!\langle0|_{X}\otimes\pi_{M}+\left(  1-q\right)
%|1\rangle\!\langle1|_{X}\otimes\sigma^{(1)}\pi_{M}\sigma^{(1)},
%\end{equation}
%where $q\in\left(  0,1\right)  $, $\sigma_{x}$ is the Pauli flip operator, and%
%\begin{equation}
%\pi_{M}\coloneqq \left(  1-\frac{1}{M}\right)  |0\rangle\!\langle0|+\frac{1}{M}%
%|1\rangle\!\langle1|.
%\end{equation}
%The amount of distinguishability in this resource state is quantified in part
%by the parameter $M$. In the case that $M=2$, we have that $\pi_{M}=\sigma
%$_{x}\pi_{M}\sigma_{x}$, so that this is the case of a minimal golden unit resource state with no distinguishability. When $M=1$ or in the limit $M\rightarrow\infty$, the target output state is transformable by free
%operations to the perfectly correlated state $q|0\rangle\!\langle0|_{X}% \otimes|0\rangle\!\langle0|+\left(  1-q\right)  |1\rangle\!\langle1|_{X}% \otimes|1\rangle\!\langle1|$, which is an infinite-resource state containing an
%infinite amount of distinguishability. We note here that, due the symmetry of
%this state, it suffices to take $M\geq2$.

Moving on to the case of asymptotic transformations, we consider the following
fundamental conversion task via free operations:
\begin{align}\label{asymp-conv}
    \rho_{XA}^{(n)} \equiv (p, \rho_0^{\otimes n}, \rho_1^{\otimes n})\mapsto \sigma_{XB}^{(m)}\equiv (q, \sigma_0^{\otimes m}, \sigma_1^{\otimes m}),
\end{align}
where $p, q \in [0,1]$ and, for $i \in \{0,1\}$, $\rho_i$ and $\sigma_i$ are states of quantum systems $A$ and $B$, respectively.
The goal here is, for a fixed $n$, to make $m$ as large as possible, and to evaluate the optimal asymptotic rate $\frac{m}{n}$ of the transformation  in the limit as $n$ becomes
arbitrarily large. 
{We allow for approximations in the transformation and require the approximation error to vanish in the asymptotic limit ($n\to\infty$). This approximation error will be measured with respect to an error measure (defined for c-q states of the form considered in this paper) that we denote by the symbol~$D^\prime$.
The precise definition of $D^\prime$ and its mathematical properties are given in Section~\ref{sec:conv-dist}.

\subsection{Main results}

In this paper, we develop a consistent and systematic resource theory of
symmetric distinguishability that answers the most important questions
associated with a resource theory. In brief, the main contributions of this paper can be summarized as follows. In the following, all c-q states are assumed to be of the form~\eqref{cq-state}, and they hence represent elementary quantum information sources.
\begin{itemize}
    \item We define two new examples of {\em{generalized divergences}}, each of which satisfies the data-processing inequality (DPI). These are denoted as $\xi_{\min}$ and $\xi_{\max}$. \comment{; $\xi_{\min}$ and $\xi_{\max}$ satisfy the DPI under ${\rm{CPTP}}_A$ maps,} Moreover, we define a quantity, denoted as $\xi_{\max}^\star$, on c-q states of the form \eqref{cq-state}, and we prove that it satisfies monotonicity under CDS maps.
    
    \item All of the above quantities are of operational significance in the RTSD:
    
    \begin{itemize}
        \item The one-shot exact distillable-SD of $\rho_{XA}$ under ${\rm{CPTP}}_A$ maps is given by $\xi_{\min}(\rho_{XA})$ (Theorem~\ref{theo-distilCPTP});
        
        \item The one-shot exact SD-cost of $\rho_{XA}$ under ${\rm{CPTP}}_A$ maps is given by $\xi_{\max}(\rho_{XA})$ (Theorem~\ref{theo-diluteCPTP});
        
        \item  The one-shot exact distillable-SD of $\rho_{XA}$ under ${\rm{CDS}}$ maps is given by its symmetric distinguishability, ${\rm{SD}}(\rho_{XA})$ (Theorem~\ref{theo-distilCDS});
        
        \item In addition, the one-shot exact SD-cost of $\rho_{XA}$ under CDS maps is given by $\xi_{\max}^\star(\rho_{XA})$  (Theorem~\ref{theo-diluteCDS}).
    \end{itemize}
    
    \item $\xi_{\max}(\rho_{XA})$ and $\xi_{\max}^\star(\rho_{XA})$ are both defined in terms of the {\em{Thompson metric}} of the state ${\rho_{XA}}$ (see Theorems~\ref{theo-diluteCPTP} and \ref{theo-diluteCDS}), thus providing operational interpretations of the latter in the context of the RTSD\footnote{Even though the Thompson metric has been widely studied in the mathematics literature, to the best of our knowledge, this is the first time an operational meaning has been given to it.}.
    
    \item The optimal asymptotic rate of exact and approximate SD-distillation for a state $\rho_{XA}\equiv (p, \rho_0, \rho_1)$, under both ${\rm{CPTP}}_A$ 
    and CDS maps, is equal to its {\em{quantum Chernoff divergence}}, $\xi(\rho_0,\rho_1)$ (Theorem~\ref{theo-asymp} and \ref{theo-asympAprDistil}), where
    \begin{equation}
\xi(\rho_{0},\rho_{1})\coloneqq \sup_{s\in\left[  0,1\right]  }\left(  -\log\operatorname{Tr}[\rho_{0}^{s}\rho_{1}^{1-s}]\right)  .
\label{eq:chernoff-div-defn}
\end{equation}

\item The optimal asymptotic rate of exact SD-dilution for a state $\rho_{XA}$ is equal to its Thompson metric (see Theorem~\ref{theo-exact-asympDil}). This provides another clear operational interpretation for the latter. 

\item The optimal asymptotic rate of approximate SD-dilution for a state $\rho_{XA}$ is equal to its quantum Chernoff divergence (see Theorem~\ref{theo-asympDil}).
    \item The optimal asymptotic rate of transforming one c-q state to another, under both ${\rm{CPTP}}_A$ and ${\rm{CDS}}$ maps, is equal to the ratio of their quantum Chernoff divergences (see Theorem~\ref{thm:transformationrate}). This result constitutes a novel operational interpretation of the Chernoff
divergence beyond that reported in earlier work on symmetric quantum
hypothesis testing \cite{nussbaum2009chernoff,ACMBMAV07}. It also demonstrates
that the resource theory of symmetric distinguishability (RTSD) is {\em{asymptotically
reversible.}}

\end{itemize}

In the following sections, we develop all
of the above claims in more detail. In particular, in Section~\ref{sec:axioms}, we introduce a general resource theory of symmetric distinguishability, for arbitrary quantum information sources (given by an ensemble $\{p_x, \rho_x\}_{x \in \cX}$ of quantum states), via an axiomatic approach. The resource theory studied in the rest of the paper is a special case of the above, namely, the one for elementary quantum information sources (i.e.,~corresponding to the choice $|\cX| = 2$). Certain necessary ingredients of the RTSD are introduced in this section and in Section~\ref{sec:ingredients}. These include the notion of golden units, which facilitates a study of distillation and dilution of symmetric distinguishability (SD). SD-distillation and SD-dilution are studied in Sections~\ref{sec:distil} and \ref{sec:SD-dilution}, respectively, both in the one-shot and asymptotic regimes. In Section~\ref{sec:examples}, we elucidate the salient features of the RTSD for certain examples of elementary quantum information sources.
 The interesting task of converting one elementary quantum information source to another via free operations is studied in Section~\ref{sec:asymp}. We conclude the main part of the paper with a summary and some open questions for future research. Various relevant quantities of the RTSD can be formulated as semi-definite programs (SDPs). These are stated in Sections~\ref{sec:ingredients} and~\ref{sec:distil}, but some of their proofs appear in the appendices.

\section{An axiomatic approach to the resource theory of symmetric distinguishability }

\label{sec:axioms}

In this section, we introduce an axiomatic approach to a resource theory of symmetric distinguishability, from which the particular resource theory that we study in this paper
arises as a natural special case, corresponding to the choice $|X| \equiv |\mX| = 2$ in what follows.

Consider an ensemble $\{p_x,\rho_x\}_{x \in \mX}$ of  quantum states. Such an ensemble can be described by a c-q state
\be
\label{eq:cqgen}
\rho_{XA}=\sum_{x\in\mX}p_x|x\lr x|\otimes\rho_x\;.
\ee
There are many functions that can be used to quantify the distinguishability of the ensemble of states above. Perhaps the function that is most operationally motivated is the {\em{guessing probability}}:
\begin{align}
\label{eq:pguess}
    p_{\rm{guess}}(X|A) &\coloneqq  \max_{\{\Lambda_x\}_{x \in \cX}}  \sum_{x \in \cX} p_x \Tr (\Lambda_x \rho_x),
\end{align}
where the maximum is over every possible POVM $\{\Lambda_x\}_{x \in \cX}$.

We are interested in a notion of ensemble distinguishability that takes into account the prior distribution $\{p_x\}_{x \in \mX}$, while distinguishing between the states in the ensemble, as opposed to state distinguishability, which is concerned only with the distinguishability of the states in the set $\{\rho_x\}_{x \in \mX}$. 

Motivated by the guessing probability and symmetric hypothesis testing, we identify an ensemble as free in the resource theory of symmetric distinguishability if the guessing probability takes on its minimum value; i.e., if the guessing probability is equal to $\frac{1}{|\mathcal{X}|}$, which is the same value attained by a random guessing strategy. Note that the guessing probability is equal to $\frac{1}{|\mathcal{X}|}$ if and only if all the states of the ensemble $\{p_x, \rho_x\}_{x\in \cX}$ are identical and the prior distribution is uniform.
For the sake of completeness, we include a proof of this fact at the end of this section (see Lemma~\ref{lem:pguess} below).

Thus, for all such ensembles, the symmetric distinguishability  is equal to zero.  The corresponding c-q state $\rho_{XA} = \pi_X \otimes \omega_A$, where $\pi_X \coloneqq  {I_X}/{|\cX|}$ is the completely
mixed state, is then a `free' state of the resource theory of symmetric distinguishability because it has zero SD. This leads us to identify the set of free states in the resource theory of symmetric distinguishability as follows:
\be
\mathfrak{F}(XA)\eqdef\left\{\pi_X\otimes\omega_A : \omega_A\in\mathcal{D}(A)\right\}\;.
\ee
Note that, in the above, we use the notation $\mathcal{D}(A)$ to denote the set of quantum states (i.e., density matrices) of the quantum system $A$.

To make the notion of symmetric distinguishability (SD) precise, we introduce an axiomatic approach. Here, we  define a preorder relation $\prec$ on the set of all c-q states, which satisfies Axioms~I-V below (see Definition~\ref{def-PreorderSD}). We say a c-q state $\rho_{XA}$ has less symmetric distinguishability than $\sigma_{X'A'}$ if  $\rho_{XA}\prec\sigma_{X'A'}$. In this approach, SD is a property of a composite physical system shared between two parties, say, Xiao and Alice, with Xiao possessing classical systems (denoted by $X$, $X'$, etc.) and Alice possessing quantum systems denoted by $A$, $A'$, etc.
The word ``symmetric" refers to the fact that the distinguishability is symmetric with respect to the ordering in the ensemble; i.e., for every permutation $\pi$, the ensemble $\{p_x,\rho_x\}_{x}$ has the same SD as the ensemble $\{p_{\pi(x)},\rho_{\pi(x)}\}_x$ because the latter is just a relabeling of the former. We write this equivalence as the following relation:
\begin{flalign}
\label{ax1}
\qquad\qquad \textbf{Axiom I} & \quad\rho_{XA}\sim \mP_{X\to X}\!\left(\rho_{XA}\right)\quad\quad\forall\;\mP-\text{permutation channel}. &&
\end{flalign}
Similarly, every isometric channel $\mV\in\cptp(A\to A')$ that acts on the states $\{\rho_x\}_x$ does not change their ``overlap" (i.e., their Hilbert--Schmidt inner product), and this leads to the next axiom:
\begin{flalign}
\label{ax2}
\qquad\qquad \textbf{Axiom II} & \quad\rho_{XA}\sim \mV_{A\to A'}\!\left(\rho_{XA}\right)\quad\quad\forall\;\mV-\text{isometric channel}. &&
\end{flalign}

\smallskip

\noindent
The fact that states with zero SD cannot add SD leads to the next axiom:
for every classical system~$X'$ on Xiao's side and every state $\omega_{A'}$ on Alice's side,
\begin{flalign}
\label{ax3}
\qquad\qquad \textbf{Axiom III} & \quad\rho_{XA}\sim \rho_{XA}\otimes\left(\pi_{X'}\otimes\omega_{A'}\right) . &&
\end{flalign}
The next axiom concerns two c-q states $\rho_{XAA'},\sigma_{XAA'}\in {\mathcal{D}}(XAA')$ that have the form
\be
\rho_{XAA'}=\sum_{y\in \cY}q_y\rho_{XA}^y\otimes |y\lr y|_{A'}\quad\text{and}\quad\sigma_{XAA'}=\sum_{y\in \cY}q_y\sigma_{XA}^y\otimes |y\lr y|_{A'}\;,
\ee
with $\cY$  a finite alphabet, $q_y \in [0,1]$ the elements of a probability distribution, and $\{\ket{y}\}_{y \in \cY}$  an orthonormal basis. \comment{The interpretation of for example $\rho_{XAA^\prime}$ is that with probability $q_y$ the system is in an ensemble corresponding to the c-q state $\rho^y_{XA}$. Since the labels $y$ in such states are distinguishable by Alice, we assume that} Since the label $y$ is distinguishable in such states, we assume that 
\begin{flalign}
\label{ax4}
\qquad\qquad \textbf{Axiom IV} & \quad 
 \rho_{XA}^y\sim\sigma_{XA}^y\quad\forall\;y \in \cY\quad\Rightarrow\quad \rho_{XAA'}\sim\sigma_{XAA'} .&&
\end{flalign}
The motivation behind Axiom~IV is the following: Since the label $y$ can be perfectly inferred by a measurement of system $A^\prime$, the symmetric distinguishability of the states $\rho_{XAA^\prime}$ and $\sigma_{XAA^\prime}$ should be fully determined by the symmetric distinguishability of the individual states $\rho_{XA}^y$ and $\sigma_{XA}^y$. Hence, if $\rho_{XA}^y\sim\sigma_{XA}^y$ for all $y$, then their mixtures $\rho_{XAA^\prime}$ and $\sigma_{XAA^\prime}$ should also be equivalent.

\medskip

\noindent The last axiom is the following natural assumption: The SD of a c-q state does not increase by discarding subsystems; i.e.,
\begin{flalign}
\label{ax5}
\qquad\qquad \textbf{Axiom V} & \quad \Tr_{X'A'}(\rho_{XX'AA'}) \prec \rho_{XX'AA'.} &&
\end{flalign}

\smallskip

\noindent
The above axioms now lead to the formal definition of the preorder of SD:

\begin{definition}[Preorder of SD]
\label{def-PreorderSD}
Let $\md(XA)$ denote the set of c-q states on a classical system~$X$ and a quantum system~$A$, and let
$$\md_{cq}\coloneqq\cup_{X,A}\md(XA)$$
be the union over all finite-dimensional systems $X$ and $A$.
The preorder of SD is the smallest preorder relation on $\md_{cq}$ that satisfies Axioms I-V.\footnote{Note that an arbitrary preorder relation $\prec$ on $\md_{cq}$ can be interpreted as a set $R\subset \md_{cq}\times\md_{cq}$, such that $\rho_{XA}\prec \sigma_{X'A'}$ if and only if $(\rho_{XA},\sigma_{X'A'})\in R$. The smallest preorder relation satisfying Axioms I-V is then given by the intersection $\cap_{R\in \mathcal{I}}R$, with the index set $\mathcal{I}$ defined as $\mathcal{I}=\{R\,|\,R\text{ preorder relation on $\md_{cq}$ satisfying Axioms I-V}\}$. As the intersection of preorder relations is again a preorder relation, we see that $\cap_{R\in \mathcal{I}}R$ gives a well-defined preorder on~$\md_{cq}$.}
\end{definition}

In Appendix~\ref{sec:AppendAxiom}, we discuss a number of consequences of Axioms I-V and the preorder of SD. There, we also define general resource measures to quantify symmetric distinguishability, and we provide several examples.

We now use this preorder to define the set of free operations.

\begin{definition}[Free operations of the RTSD]
\label{def-FO}  A map $\cN\in\cptp(XA\to X'A')$ is said to be a free operation if 
\begin{align}
\cN(\rho_{XA})\prec\rho_{XA}
\end{align}
for every c-q state $\rho_{XA}$. We denote the set of such free operations by $\mf(XA\to X'A')$.
\end{definition}

The above definition of the free operations defines the resource theory of symmetric distinguishability (SD).
Specifically, we identify SD as a property of a composite physical system, shared between two parties, Xiao and Alice (with Xiao's systems being classical and Alice's being quantum), that can neither be generated nor increased by the set $\mathfrak{F}$ of free operations.

%\begin{remark}
%Note that the equivalence in~\eqref{ax4} implies that 
%\be\label{6}
%\sum_{y\in \cY}q_y\rho_{XA}^y\otimes |y\lr y|^{A'}\sim\sum_{y\in \cY}q_y(\mP_{X}^y\otimes\mU_{A}^y) %\left(\rho_{XA}^y\right)\otimes |y\lr y|^{A'}
%\ee
%where $\mP^y_X\in\cptp(X\to X)$ are permutation channels and $\mU^y_A\in\cptp(A\to A)$ are unitary channels.
%\end{remark}}

There is an important class of operations that play a central role in the resource theory of symmetric distinguishability. These are referred to as {\em{conditionally doubly stochastic (CDS) maps}} and were first introduced in~\cite{GGH+2018}. In fact, when the classical input and output systems of the resource theory are the same (i.e., $X=X'$), then the set of free operations, defined above, reduces to the set of CDS maps. This is stated in Lemma~\ref{FO-CDS} below. Before proceeding to the lemma, we  introduce the notion of CDS maps in the next section.

	\subsection{Conditional doubly stochastic (CDS) maps}
	
	Consider the following problem: Xiao picks a state $\rho_x\in\{\rho_1,...,\rho_{|\cX|}\}$ at random, with prior probability
	$p_x$, and then sends the state $\rho_x$ to Alice through a noiseless quantum channel. Alice knows the probability distribution $\{p_x\}_x$ from which Xiao sampled the state $\rho_x$, but she does not know the value of $x$. Therefore, the overall state can be represented as a classical-quantum state of the form:
	\be\label{cq}
	\rho_{XA}=\sum_{x \in \mathcal{X}} p_x |x\lr x|\otimes\rho_{x} ,
	\ee
	where $X$ is a classical register to which Alice does not have access and $A$ is a quantum register to which she does have access. What are the ways in which Alice can manipulate the state
	$\rho_{XA}$? 	There are two basic operations that she can perform: 
	\begin{enumerate}
	    \item Alice can perform a generalized measurement on her system $A$.
	    \item Alice can partially lose her knowledge of the distribution $p_x$, by performing random relabeling on the alphabet of the classical system $\cX \coloneqq  \{1,2, \ldots, |\cX|\}$. 
	\end{enumerate}
That is, Alice can perform a generalized measurement on the quantum system $A$, and based on the outcome, say $j$, apply a random relabeling map $\mD^{(j)}$ on the classical system $X$. Hence, the most general operation that Alice can perform is a CPTP map 
 $\mN\in\cptp(XA\to XA') $ of the form
	\be\label{mainform}
	\mN=\sum_j \mD^{j}_{X\to X}\otimes\mE^j_{A\to A'} ,
	\ee
	where each $\mathcal{D}^{j}$ is a classical doubly stochastic channel and each $\mE^j$ is a  completely positive (CP) map such that 
	$\sum_j\mE^j$ is CPTP. Alternatively, since every doubly stochastic matrix can be expressed as a convex combination of permutation matrices, there exist conditional probabilities $t_{z|j}$ such that
	\be
	\mD^{j}=\sum_{z=1}^{|\mathcal{X}|!}t_{z|j}\mathcal{P}^{z}, \qquad t_{z|j}\geq 0 \ \forall z,j, \qquad\sum_{z=1}^{|\mathcal{X}|!}t_{z|j}=1 \ \forall j\;.
	\ee
	where $\mathcal{P}^z$ is the permutation channel defined by $\mathcal{P}^z(|x\lr x|)=|\pi_{z}(x)\lr\pi_{z}(x)|$, with $\pi_z$ being one  of the $|\mathcal{X}|!$ permutations. We therefore conclude that
	\be
	\mN=\sum_{z=1}^{|\mathcal{X}|!} \mathcal{P}^{z}_{X\to X}\otimes\tilde{\mE}^z_{A\to A'},
	\ee
	where $\tilde{\mE}^z\equiv\sum_{j}t_{z|j}\mE^j$.
Hence, we can assume, without loss of generality, that in~\eqref{mainform} the map $\mD^{z}=\mathcal{P}^{z}$, with $z=1,...,
|\mathcal{X}|!$, so that the maps that Alice can perform are given by
	\be\label{mainform2}
	\mN=\sum_{z=1}^{|\mathcal{X}|!} \mathcal{P}^z_{X\to X}\otimes\mE^z_{A\to A'}\;.
	\ee
We call such CPTP maps \emph{conditional doubly stochastic} (CDS) maps.

The following lemma states that, when the input and output classical systems are the same (i.e., $X=X'$),  the free operations for the resource theory of symmetric distinguishability (RTSD), introduced in Definition~\ref{def-FO}, are given by CDS maps, denoted by $\cds(XA\to XA')$.

\smallskip

\begin{lemma}\label{FO-CDS}
For a classical system $X$ and quantum systems $A$ and $A'$, the following set equivalence holds
\be
\cds(XA\to XA')=\mf(XA\to XA')\;.
\ee
\end{lemma}
\begin{proof}
We first prove that
\be
\cds(XA\to XA')\subseteq\mf(XA\to XA')\;.
\label{eq:cdsinfree}
\ee
Since partial trace and isometries acting on Alice's systems are free, it follows that every quantum instrument on Alice's side is free. Let
\be
\mE_{A\to A'A''}(\omega_A)\eqdef\sum_{y\in \cY}\mE^y_{A\to A'}(\omega_A)\otimes|y\lr y|_{A''}
\ee
be a quantum instrument on Alice's side, and let $\rho_{XA}\in {\mathfrak{D}}(XA)$. Then, the action of the quantum instrument on $\rho_{XA}$ yields the state
\be\label{223}
\sum_{y\in \cY}\mE^y_{A\to A'}(\rho_{XA})\otimes|y\lr y|_{A''} ,
\ee
 and from Axiom~I, we have, for all $y\in\cY$, the equivalence
\be
\frac{\mE^y_{A\to A'}(\rho_{XA})}{\tr\!\left[\mE^y_{A\to A'}(\rho_{XA})\right]}\sim\frac{\cP_{X\to X}^y\otimes\mE^y_{A\to A'}(\rho_{XA})}{\tr\!\left[\mE^y_{A\to A'}(\rho_{XA})\right]},
\ee
where each $\cP_{X\to X}^y$ is a permutation channel.
Combining this with Axiom~IV, we conclude that the c-q state in~\eqref{223} can further be transformed to
\be
\sum_{y\in \cY}(\mP^y_{X\to X}\otimes\mE^y_{A\to A'})(\rho_{XA})\otimes|y\lr y|_{A''}\;.
\ee
Finally, tracing system $A''$ yields the overall transformation
\be
\rho_{XA}\mapsto\sum_{y\in \cY}(\mP^y_{X\to X}\otimes\mE^y_{A\to A'})(\rho_{XA})\;,
\ee
which is the general form of a CDS map.
This completes the proof of \eqref{eq:cdsinfree}.

Conversely, note that only Axiom~III is not covered by CDS maps. Therefore, the most general transformation $\mN_{XA\to XA'}\in\mf(XA\to XA')$ has the form
\be
\mN_{XA\to XA'}\left(\rho_{XA}\right) = \sum_{y\in \cY}(\tr_{X'}\circ\mP^y_{XX'}\otimes \mE^y_{A\to A'}) \left(\rho_{XA}\otimes\pi_{X'}\right) ,
\ee
where $\pi_{X'}$ is the maximally mixed state, $\{\mE_y\}_y$ is a quantum instrument, and $\mP_{XX'}^y\in\cptp(XX'\to XX')$ are joint permutation channels.
However, observe that 
\be
\mD_{X\to X}^y(\omega_X)\eqdef (\tr_{X'}\circ\mP^y_{XX'}) \left(\omega_X\otimes\pi_{X'}\right)\quad\forall\omega_X\in {\mathfrak{D}}(X)\;,
\ee
is a classical doubly stochastic channel and therefore can be expressed as a convex combination of permutation channels. We therefore conclude that $\mN_{XA\to XA'}\in\cds(XA\to XA')$.
\end{proof}
\smallskip

 Lemma~\ref{FO-CDS} above demonstrates that if the dimension of the classical system is fixed, then the free operations in the resource theory of SD are CDS maps. However, we point out that the lemma above can also be used to characterize $\mathfrak{F}(XA\to X'A')$ where $|X|\neq |X'|$. In particular, observe that
\be
\rho_{XA}\xrightarrow{\mathfrak{F}}\sigma_{X'A'}\quad\iff\quad\pi_{X'}\otimes\rho_{XA}\xrightarrow{\cds}\pi_{X}\otimes\sigma_{X'A'}
\ee
 for all $\rho_{XA},\sigma_{X'A'}\in\mathfrak{D}_{cq}$
because the maximally mixed states $\pi_X$ and $\pi_{X'}$ are free. In the general case, when $|X|\neq |X'|$, the set $\mathfrak{F}(AX\to A'X')$ can be viewed as a special subset of conditional thermal operations~\cite{NG17}, corresponding to a thermodynamical system with a completely degenerate Hamiltonian. We therefore call it the set of \emph{conditional noisy operations}.

In the rest of the paper, we focus on the case in which the input classical system $X$ has the same dimension as the output classical system $X'$, and furthermore, we constrain both of them to have dimension equal to two.  Thus, according to Lemma~\ref{FO-CDS}, the set of free operations reduces to CDS maps in all of our discussions that follow.

Moreover, in addition to CDS maps, we will also consider the set
\begin{align}
\cptp_A \coloneqq \Big\{{\rm{id}}\otimes\cE \, \Big| \, \cE\, \operatorname{CPTP} \text{on system } A\Big\}
\end{align}
as a possible set of transformations. Note that $\cptp_A$ has the clear physical interpretation of applying a fixed quantum channel onto the quantum part of the c-q state without changing its classical probability distribution. Furthermore, it is clear from the definition of CDS maps that we have the inclusion
\begin{align}
    \cptp_A \subset \cds.
\end{align}

The following lemma shows that the minimum error probability $p_{\err}(\rho_{XA})$ can only be increased by application of a CDS map.
\begin{lemma}[Monotonicity of minimum error probability under CDS maps]
    \label{lem:MonoErrProbCDS}
    
    Let $\cN\in\cds(XA\to XA')$ and $\rho_{XA}$ be a c-q state. Then
    \begin{align}
    p_{\err}(\cN(\rho_{XA})) \ge p_{\err}(\rho_{XA}),
    \end{align}
    where $p_{\operatorname{err}}(\rho_{XA})\coloneqq 1 - p_{\rm{guess}}(X|A)$, with $p_{\rm{guess}}(X|A)$ defined in \eqref{eq:pguess}.
\end{lemma}
\begin{proof}
 The guessing probability $p_{\rm{guess}}(X|A) $ can be written as follows \cite{KRS08}:
\begin{align*}
     p_{\rm{guess}}(X|A) &= 2^{\inf_{\omega_A} D_{\max} (\rho_{XA}\| I_X \otimes \omega_A)},
\end{align*}
where $D_{\max}$ denotes the max-relative entropy, which for a state $\rho$ and a 
    positive semi-definite operator $\sigma$ is defined as follows \cite{D09}: 
    \begin{align}
        D_{\max}(\rho\|\sigma) \coloneqq  \inf \{ \lambda \,:\, \rho \le 2^\lambda \sigma\}.
    \end{align}
Thus, we conclude that
\begin{align*}
p_{\err}(\cN(\rho_{XA})) &= 1- 2^{\inf_{\sigma_{A^\prime}} D_{\max} (\cN(\rho_{XA})\| I_X \otimes \sigma_{A^\prime})} \ge 1 - 2^{\inf_{\omega_A} D_{\max} (\cN(\rho_{XA})\| \cN(I_X \otimes \omega_A))} \\&\ge1- 2^{\inf_{\omega_A} D_{\max} (\rho_{XA}\| I_X \otimes \omega_A)} = p_{\err}(\rho_{XA}).
\end{align*}
The first inequality follows from the fact that for each quantum state $\omega_{A}$ on system $A$, the image under the CDS map $\cN$ can be written as $\cN(I_X\otimes\omega_A) =I_X\otimes\sigma_{A^\prime}$ with $\sigma_{A^\prime}$ a state on system~$A^\prime$ (which is evident from the definition of CDS maps). The second inequality follows from the data-processing inequality for the max-relative entropy~\cite{D09}.
\end{proof}

\bigskip
 The above lemma immediately leads to a natural choice of a measure of SD for the particular case of the RTSD that we study in this paper, namely, one in which the dimension
 of the classical system $X$ is fixed to $|X|=2$. This measure was mentioned in \eqref{eq:def-SD-measure}, and we recall its definition here:
 \begin{definition}\label{def:SD}
For a fixed dimension $|X|=2$, we define 
\be
{\rm SD}(\rho_{XA})\eqdef -\log\big(2p_{\operatorname{err}}(\rho_{XA})\big)\;.
\ee
This function is equal to zero on free states and behaves monotonically under CDS maps, as is evident from Lemma~\ref{lem:MonoErrProbCDS} above.

\end{definition}

\begin{lemma}\label{lem:pguess} 
For an ensemble $\{p_x, \rho_A^x\}_{x \in \cX}$, the following lower bound holds for the guessing probability:
\begin{equation}
 p_{\rm{guess}}(X|A) \geq  \frac{1}{|\cX|}   ,
 \label{eq:lower-bnd-guess-prob}
\end{equation}
with $p_{\rm{guess}}(X|A)$ defined in~\eqref{eq:pguess}.
This lower bound is saturated (i.e., $p_{\rm{guess}}(X|A)= \frac{1}{|\cX|}$) if and only if
\begin{align}\label{eq:A}
\exists \sigma_A \in \cD(A) \,\, \forall\, x \in \cX, \quad \rho_x = \sigma_A \,\,\,  , \quad {\hbox{and}} \quad p_x=\frac{1}{|\cX|}.
\end{align}
\end{lemma}

\begin{proof}
    It is known that~\cite{KRS08} 
    \begin{align}
     p_{\rm{guess}}(X|A) &= 2^{\inf_{\omega_A} D_{\max} (\rho_{XA}\| I_X \otimes \omega_A)}
     =  \frac{1}{|\cX|} 2^{\inf_{\omega_A} D_{\max} (\rho_{XA}\| \pi_X \otimes \omega_A)},
     \label{eq:KRS-result}
    \end{align}
    where the infimum is over every state $\omega_A$.
    In the above, $\rho_{XA}$ is the c-q state corresponding to the ensemble and is defined in~\eqref{eq:cqgen}, and $\pi_X = I_X/|\cX|$ is the completely mixed state.
    The lower bound in \eqref{eq:lower-bnd-guess-prob} is then a direct consequence of \eqref{eq:KRS-result} and the fact that $D_{\max}(\rho\|\sigma) \geq 0$ for states $\rho$ and~$\sigma$.
    %Since the guessing probability has a minimum value of $1/|\cX|$, corresponding to a random guessing strategy, we obtain
    %\begin{align}
     %   \frac{1}{|\cX|} \le p_{\rm{guess}}(X|A)  &= \frac{1}{|\cX|} 2^{\inf_{\omega_A}D_{\max} (\rho_{XA}\| \pi_X \otimes \omega_A)}
      %  \end{align}
      As a consequence of \eqref{eq:KRS-result}, we also conclude that
\begin{align}
\label{eq:equivalencepgues}
    p_{\rm{guess}}(X|A)= \frac{1}{|\cX|} \, & \iff   \,\inf_{\omega_A} D_{\max} (\rho_{XA}\| \pi_X \otimes \omega_A) =0.  
\end{align}
Now note that the infimum on the right-hand side of \eqref{eq:equivalencepgues} is actually a minimum (due to the finiteness of  $\inf_{\omega_A}D_{\max} (\rho_{XA}\| \pi_X \otimes \omega_A)$ and continuity of $2^{\lambda}\sigma$ in $\lambda$). Therefore,
\begin{align}
p_{\rm{guess}}(X|A)= \frac{1}{|\cX|} \,  \iff \exists \,\sigma_A\text{ s.t. } \rho_{XA} = \pi_X\otimes\sigma_A,
\end{align}
where we also employed the property of max-relative entropy that, for states $\rho$ and $\sigma$, $D_{\max}(\rho \| \sigma) = 0 $ if and only if $\rho = \sigma$. 
\end{proof}

\begin{remark}
By rearranging \eqref{eq:KRS-result}, we find that
\begin{equation}
    \log (|\cX|p_{\rm{guess}}(X|A)) = \inf_{\omega_A} D_{\max} (\rho_{XA}\| \pi_X \otimes \omega_A).
\end{equation}
This quantity is a measure of symmetric distinguishability alternative to ${\rm{SD}}(\rho_{XA})$, as defined in~\eqref{eq:def-SD-measure}. Indeed, this measure of symmetric distinguishability has the appealing feature that it is the max-relative entropy from the state $\rho_{XA}$ of interest to the set of free states. As is common in quantum resource theories \cite{CG18}, one could then define an infinite number of SD measures based on the generalized divergence of the state of interest with the set of free states.
\end{remark}

\section{Some key ingredients of the resource theory of symmetric distinguishability}

\label{sec:ingredients}

\subsection{Infinite-resource states and golden units}

\label{sec:prereq}

In the following sections, we study the information-theoretic tasks of distillation and dilution within the framework of the resource theory of symmetric distinguishability (SD), with respect to the golden unit of Definition~\ref{golden-unit} below. We consider two different choices of free operations: $(i)$ CPTP maps on the quantum system $A$ and the identity channel on the classical system, which we denote as ${\rm{CPTP}}_A$, and $(ii)$ CDS maps, which were introduced in the previous section. The reason for considering both of these is that they lead to novel and interesting results. Also, in some cases, a proof for one choice of free operations follows as a simple corollary from that for the other choice. We refer to these tasks as SD-distillation and SD-dilution, respectively. We study the exact and approximate one-shot cases, as well as the asymptotic case for both of these tasks.

In this section, we prove certain key results that serve as prerequisites for the above study; they involve {\em{infinite-resource states}} and {\em{golden units}}---notions that were 
introduced in Section~\ref{sec:intro}. We recall their definitions before stating the relevant results. 
As mentioned in the Introduction, the basic objects of the RTSD are elementary quantum information sources represented by 
\begin{equation}
\rho_{XA}\coloneqq p|0\rangle\!\langle0|\otimes\rho_{0}+\left(  1-p\right)
|1\rangle\!\langle 1|\otimes\rho_{1},
\label{eq:basic-cq-state}
\end{equation}
where $p\in\left[  0,1\right]  $ is a prior probability and $\rho_{0}$ and
$\rho_{1}$ are quantum states.

The main operational quantity associated with such a state, in the context of symmetric hypothesis testing, is the minimum error probability of Bayesian state discrimination of the states
$\rho_0$ and $\rho_1$:
\begin{align}
\label{op-err}
p_{\operatorname{err}}(\rho_{XA})= \min_{0\le\Lambda\le\1}\Big(p\Tr\!\left(\Lambda\rho_0\right) + (1-p) \Tr\!\left((\1 -\Lambda)\rho_1\right)\Big).
\end{align}
\begin{remark}
  Since the c-q state $\rho_{XA}$ is also represented by the quantum box $(p, \rho_0, \rho_1)$ (as mentioned in the Introduction), we sometimes use the notation
 $p_{\operatorname{err}}(p, \rho_0, \rho_1)$ instead of $p_{\operatorname{err}}(\rho_{XA})$ in the following. 
\end{remark}
The following well-known theorem gives an explicit expression for this minimum error probability \cite{Hel67,Hel69,Hol72}.
\begin{theorem}[Helstrom--Holevo Theorem]
\label{thm:HH-state-disc}
For a c-q state $\rho_{XA}$ of the form in \eqref{eq:basic-cq-state}, the following equality holds
\begin{align}
p_{\operatorname{err}}(\rho_{XA}) = \frac{1}{2}\Big(1- \left\|p\rho_0- (1-p)\rho_1\right\|_1\Big).
\end{align}
\end{theorem}

Recall from Section~\ref{sec:intro} that a c-q state $\rho_{XA}$ is said to be an \emph{infinite-resource state} if $p_{\operatorname{err}}(\rho_{XA})=0,$ which is equivalent to the quantum states $\rho_0$ and $\rho_1$ having mutually orthogonal support. Hence, the symmetric distinguishability given by ${\rm{SD}}(\rho_{XA})=-\log\left(2 p_{\operatorname{err}}(\rho_{XA})\right)$ is infinite in this case.

The following lemma shows that we can transform any infinite-resource state to any other c-q state of the form in \eqref{eq:basic-cq-state} via CDS maps.
\begin{lemma}
    \label{lem:InfToInf}
    Let $\omega_{XA}$ be an infinite-resource state, and let $\sigma_{XB}$ be a general c-q state. Then there exists a CDS map $\cN: XA \to XB$ such that
    \begin{align}
    \cN(\omega_{XA}) = \sigma_{XB}.
    \end{align}
\end{lemma}
\begin{proof}
We write the c-q states explicitly as
\begin{align*}
\omega_{XA} &= p\kb{0}\otimes\omega_0 + (1-p)\kb{1}\otimes\omega_1 ,
\\\sigma_{XB} &= q\kb{0}\otimes\sigma_0 + (1-q)\kb{1}\otimes\sigma_1,
\end{align*}
for some $p,q\in [0,1]$ and quantum states $\omega_0$, $\omega_1$, $\sigma_0$, and $\sigma_1$.
As $\omega_{XA}$ is an infinite-resource state, and hence $\omega_0$ and $\omega_1$ have mutually orthogonal supports, we can pick a POVM $\{\Lambda,\1-\Lambda\}$ such that $\Tr(\Lambda\omega_0) =\Tr((\1-\Lambda)\omega_1) = 1$ and consequently $\Tr(\Lambda\omega_1) =\Tr((\1-\Lambda)\omega_0) =0.$
Consider a pair $(\cE_0, \cE_1)$ of quantum operations, i.e., completely positive, trace non-increasing linear maps that sum to a CPTP map, defined as follows: 
\begin{align*}
\cE_0(\cdot) = q\Tr(\Lambda\cdot)\sigma_0 + (1-q)\Tr((\1-\Lambda)\cdot)\sigma_1,\quad\quad\cE_1(\cdot) = q\Tr((\1-\Lambda)\cdot)\sigma_0 + (1-q)\Tr( \Lambda\cdot)\sigma_1,
\end{align*}
and consider the corresponding CDS map
\begin{align*}
\cN = {\rm{id}}_X \otimes \cE_0 + \cF_X\otimes \cE_1,
\end{align*}
where $\cF_X$ denotes the flip channel on the classical system $X$. We then immediately get
\begin{align*}
\cN(\omega_{XA}) = \sigma_{XB},
\end{align*}
which concludes the proof.
\end{proof}
\smallskip

As mentioned in the Introduction, it is useful to consider a particular class of c-q states that lead  naturally
to a clear definition of the fundamental tasks of distillation and dilution in the RTSD. These states are parametrized by $M \in [1, \infty]$ and $q \in (0,1)$, and for $M$ large enough have SD equal to $\log M$. 
We refer to such a state as an $(M, q)$-golden unit. It is defined as follows:
	\begin{definition}[Golden unit]\label{golden-unit}
	We choose the following class of classical-quantum (c-q) states of a composite system $XQ$, where $Q$ is a qubit. Each state is labelled by a parameter $M\in [1, \infty]$ and a probability $q \in (0,1)$ and is defined as follows:
	\begin{align}
	\gamma_{XQ}^{(M,q)} \coloneqq q\kb{0}_X\otimes\pi_M + (1-q)\kb{1}\otimes\sigma^{(1)}\pi_M\sigma^{(1)},
	\end{align}
	where
	$$\pi_M \coloneqq  \left(1-\frac{1}{2M}\right)\kb{0} + \frac{1}{2M}\kb{1}$$
	is a state of a qubit $Q$ and $\sigma^{(1)}$ denotes the Pauli-$x$ matrix. We call the state $\gamma_{XQ}^{(M,q)}$ an $(M, q)$-golden unit. Note that for $M=\infty$, we have $\pi_\infty = \kb{0}$ and hence the golden unit reduces to an infinite-resource state.
	\end{definition}

The goodness of this choice of the golden unit lies in the fact that its SD has a useful scaling property, as stated in the following lemma.
\begin{lemma}
\label{lem:goldenunitscaling}
    For all  $M\in [1, \infty)$ such that $2M \ge \max\{1/q,\,1/(1-q)\}$ and $q \in (0,1)$,
	\begin{align}
	p_{\operatorname{err}}(\gamma_{XQ}^{(M,q)}) = \frac{1}{2M},
	\end{align}
	and hence its symmetric distinguishability is given by ${\rm{SD}}(\gamma_{XQ}^{(M,q)})= \log M$.
%\smallskip
%
%\noindent
For $M=\infty$, and hence $p_{\operatorname{err}}(\gamma_{XQ}^{(M,q)})=0$, the SD of $\gamma_{XQ}^{(M,q)}$ is infinite. 
\end{lemma}
\begin{proof}
    In the case $M=\infty$, we trivially have $p_{\operatorname{err}}(\gamma_{XQ}^{(M,q)})=0$ as $\pi_\infty=\kb{0}$ and $\sigma^{(1)}\pi_\infty\sigma^{(1)}=\kb{1}$ are orthogonal.
    For the case $M\in[1,\infty)$, we use the Helstrom--Holevo Theorem to conclude that
	\begin{align}
	\nn  p_{\operatorname{err}}(\gamma_{XQ}^{(M,q)}) & = \frac{1}{2}\left(1- \left\|q\,\pi_M - (1-q)\sigma^{(1)}\pi_M\sigma^{(1)}\right\|_1\right)\\
	& = \frac{1}{2}\left(1- \left\|\left(q-\frac{1}{2M}\right)\kb{0} -\left(1-q - \frac{1}{2M}\right)\kb{1}\right\|_1\right)\nn\\
	& = \frac{1}{2}\left(1- \left|q-\frac{1}{2M}\right|-\left|1-q-\frac{1}{2M}\right|\right).
	\label{eq:simplify-p-err-golden-unit}
	\end{align} 
	Hence, for $2M \ge \max\{1/q,\,1/(1-q)\}$, we get
	\begin{align}
	p_{\operatorname{err}}(\gamma_{XQ}^{(M,q)}) = \frac{1}{2M},
	\end{align}
	and hence ${\rm{SD}}(\gamma_{XQ}^{(M,q)}) =  -\log\left(2 p_{\operatorname{err}}(\gamma_{XQ}^{(M,q)})\right)= \log M$. 
\end{proof}

Furthermore, we note that for $q_1,q_2\in[0,1]$ such that the corresponding distribution $\vec{q}_1\coloneqq(q_1,1-q_1)$ is majorised by $\vec{q}_2\coloneqq(q_2,1-q_2)$ (i.e., $\vec{q}_1\prec\vec{q}_2$), the golden unit $\gamma^{(M,q_1)}_{XQ}$ is dominated by $\gamma^{(M,q_2)}_{XQ}$ in the preorder of SD (see Definition~\ref{def-PreorderSD}). This is the statement of the following lemma.
\begin{lemma}
\label{lem:GoldenUnitPreorder}
    Let $M\in[1,\infty]$ and $q_1,q_2\in[0,1]$ be such that the corresponding distribution vectors satisfy $\vec{q}_1\prec\vec{q}_2$ in majorisation order.
    Then 
    \begin{align}
    \gamma_{XQ}^{(M,q_1)} \prec \gamma_{XQ}^{(M,q_2)},
    \end{align}
    and hence $\gamma_{XQ}^{(M,q_2)}$ can be transformed to $\gamma_{XQ}^{(M,q_1)}$ via a CDS map.
\end{lemma}
\begin{proof}
    As $\vec{q}_1\prec\vec{q}_2$, there exists a $\lambda\in[0,1]$ such that $\lambda q_2+(1-\lambda)(1-q_2) = q_1$. Now consider the CDS map
    \begin{align*}
    \cN = \lambda{\rm{id}}_X\otimes{\rm{id}}_Q + (1-\lambda)\cF_X\otimes\cF_Q,
    \end{align*}
    where $\cF_X$ and $\cF_Q$ denotes the flip channel on systems $X$ and $Q$, respectively. This directly gives
    \begin{align*}
    \cN\left(\gamma^{(M,q_2)}_{XQ}\right) &= \kb{0}\otimes\Big(\lambda q_2\,\pi_M+ (1-\lambda)(1-q_2)\,\cF_Q(\sigma^{(1)}\pi_M\sigma^{(1)})\Big)\\&\quad+\kb{1}\otimes\Big((1-\lambda) q_2\,\cF_Q(\pi_M)+ \lambda(1-q_2)\,\sigma^{(1)}\pi_M\sigma^{(1)})\Big)\\
    &=q_1 \kb{0}\otimes\pi_M + (1-q_1)\kb{1}\otimes\sigma^{(1)}\pi_M\sigma^{(1)} = \gamma^{(M,q_1)}_{XQ},
    \end{align*}
    which finishes the proof.
\end{proof}

When we consider CDS maps as free operations, it suffices to focus on the case $q=1/2$. We denote the corresponding golden unit simply as $\gamma_{XQ}^{(M)}$, and call it the $M$-{\em{golden unit}}. For future reference, we write it out explicitly:
\begin{align}\label{golden1/2}
	\gamma_{XQ}^{(M)} \equiv \gamma^{(M, 1/2)}_{XQ} = \frac{1}{2}\kb{0}\otimes\pi_M + \frac{1}{2}\kb{1}\otimes\sigma^{(1)}\pi_M\sigma^{(1)},
	\end{align}
	where
	$$
	\pi_M = \left(1-\frac{1}{2M}\right)\kb{0} + \frac{1}{2M}\kb{1}.
	$$
  \begin{remark}
Note that the golden unit $\gamma_{XQ}^{(M)}$ is equivalent under ${\rm{CDS}}$ maps to $\pi_M\otimes\omega$, for an arbitrary quantum state $\omega.$ To see that, consider first the CDS map
\begin{equation}
\cN = {\rm{id}}_X\otimes\bra{0}\cdot\ket{0}\omega + \cF_X\otimes\bra{1}\cdot\ket{1}\omega.
\end{equation}
and note that $\cN(\gamma_{XQ}^{(M)}) = \pi_M\otimes\omega$. To see the other direction, consider the ${\rm{CDS}}$ map
\begin{equation}
\cM = \frac{1}{2}\Big({\rm{id}}_X\otimes\Tr(\cdot)\kb{0} + \cF_X\otimes\Tr(\cdot)\kb{1}\Big) ,
\end{equation}
which gives 
$\cM(\pi_M\otimes\omega) = \frac{1}{2}\Big(\pi_M\otimes\kb{0} + \sigma^{(1)}\pi_M\sigma^{(1)}\otimes\kb{1}\Big)=\gamma_{XQ}^{(M)}$.
\end{remark} 

\subsection{A suitable error measure for approximate transformation tasks}
\label{sec:conv-dist}

In this section, we introduce the notion of a {\em{minimum conversion error}} for the transformation of one c-q state to another (both of the form defined in~\eqref{eq:basic-cq-state}) via free operations. For the 
transformation $\rho_{XA}\mapsto\sigma_{XB}$, we denote this quantity as $d'_{\operatorname{FO}}(\rho_{XA}\mapsto\sigma_{XB})$. The latter is defined in terms of a {\em{scaled trace distance}}, which we denote as
$D^\prime(\rho_{XA}, \sigma_{XB})$.
%divergence $D^\prime(\cdot\|\cdot)$ for a pair of c-q states of the form defined in~\eqref{eq:basic-cq-state}, and a corresponding minimum conversion error.
We also discuss some of the properties of the above quantities. These quantities are then used to define the \emph{one-shot approximate distillable-SD} and the \emph{one-shot approximate SD-cost} in the following sections. 
\begin{definition}
\label{def:err-divergence}
For general c-q states $\rho_{XA}$ and $\sigma_{XA}$, define the {\em{scaled trace distance}} 
  \begin{align}
  \label{eq:ErrorDivergence}
  D^\prime(\rho_{XA}, \sigma_{XA}) =\begin{cases} \frac{\frac{1}{2}\left\|\rho_{XA} - \sigma_{XA}\right\|_1}{p_{\err}(\sigma_{XA})},&\text{if }p_{\err}(\sigma_{XA}) >0,\\0,&\text{if }p_{\err}(\sigma_{XA}) =0\text{ and }\rho_{XA}=\sigma_{XA},\\\infty,&\text{if }p_{\err}(\sigma_{XA}) =0\text{ and }\rho_{XA}\neq\sigma_{XA}.\end{cases}
  \end{align}
 \end{definition}
 
 \begin{remark}
 Note that the scaling factor in the definition of the scaled trace distance $D^\prime(\rho_{XA}, \sigma_{XA})$ depends on the state
 $\sigma_{XA}$ in the second slot.
 \end{remark}

  \begin{definition}
  \label{def:conversiondistance}
   For a set of free operations denoted by $\operatorname{FO}$, we define the minimum conversion error corresponding to the scaled trace distance $D'(\cdot, \cdot)$ as follows:
   \begin{align}
  d'_{\operatorname{FO}}(\rho_{XA}\mapsto\sigma_{XB}) = \min_{\cA\in\operatorname{FO}}D'(\cA(\rho_{XA}), \sigma_{XB}),
  \end{align}
  with $\rho_{XA}$ and $\sigma_{XB}$ being general c-q states on the classical system $X$ and the quantum systems $A$ and $B$, respectively.
  \end{definition}
  
  \begin{proposition}
      Let $\rho_{XA}\equiv(p,\rho_0,\rho_1)$ and $\sigma_{XB}\equiv(q,\sigma_0,\sigma_1)$ be c-q states.
      The minimum conversion error $d'_{\operatorname{FO}}(\rho_{XA}\mapsto\sigma_{XB})$ for $\FO \in\{\cds,\cptp_A\}$ can be written as 
      \begin{align}
      d'_{\operatorname{FO}}(\rho_{XA}\mapsto\sigma_{XB}) & = \frac{\min_{\mathcal{M}\in\operatorname{FO}}\left\Vert
\mathcal{M}(\rho_{XA})-\sigma_{XB}\right\Vert _{1}}%
{\min_{\mathcal{N}\in\operatorname{FO}}\left\Vert \mathcal{N}(\sigma_{XB})-\gamma^{(\infty,q)}_{XQ}\right\Vert _{1}},
      \end{align}
      where $\gamma^{(\infty,q)}_{XQ} = q \kb{0}_X \otimes \kb{0}_Q + (1-q) \kb{1}_X \otimes \kb{1}_Q$ is an infinite resource state.
  \end{proposition}
  
  \begin{proof}
      By inspecting the definition of $d'_{\operatorname{FO}}(\rho_{XA}\mapsto\sigma_{XB})$, we see that it remains to prove the following:
      \begin{equation}
          p_{\err}(\sigma_{XA}) = 
          \min_{\mathcal{N}\in\operatorname{FO}}\left\Vert \mathcal{N}(\sigma_{XB})-\gamma^{(\infty,q)}_{XQ}\right\Vert _{1},
          \label{eq:err-prob-TD-to-inf-resource}
      \end{equation}
      where $\operatorname{FO}$ is either CPTP$_A$ or CDS. We leave the proof of the equality above to Appendix~\ref{app:proof-of-err-ratio-interpretation} (see Lemmas~\ref{thm:CDS-dist-max-res-err-prob} and \ref{thm:CPTP-A-dist-max-res-err-prob} therein).
  \end{proof}
	
	\begin{remark}
	The reason for considering the scaled trace distance $D'(\cdot, \cdot)$ instead of the usual trace distance 
	as an error measure for transformations in the RTSD is that using the latter would allow for the unreasonable possibility of a finite-resource state being arbitrarily close to an infinite-resource state.
	In particular, as any infinite-resource state can be transformed to any other c-q state via CDS maps (see Lemma~\ref{lem:InfToInf}), this would imply that, for any finite allowed error measured in trace norm, the transformation
	\begin{align}
	\rho^{(n)}_{XA}\equiv \left(p,\rho_0^{\otimes n},\rho_1^{\otimes n}\right) \mapsto \sigma^{(m)}_{XB}\equiv \left(q,\sigma_0^{\otimes m},\sigma_1^{\otimes m}\right)
	\end{align}
	would be possible at an infinite rate (as long as $\rho_0\neq\rho_1$).
	To see this, for all $n\in\N$, pick  a POVM $\{\Lambda_n,\1-\Lambda_n\}$ on the composite system of $n$ copies of system $A$ such that both type~I and II error probabilities corresponding to the source $(p,\rho_0^{\otimes n},\rho_1^{\otimes n})$ vanish asymptotically, i.e.
	\begin{align}
	\lim_{n\to\infty}\Tr\!\left((\1-\Lambda_n)\rho_0^{\otimes n}\right) = \lim_{n\to\infty}\Tr\!\left(\Lambda_n\rho_1^{\otimes n}\right) =0.
	\end{align}
	This is possible because $\rho_0\neq\rho_1$.\comment{because the minimum error probability $p_{\err}(p,\rho_0^{\otimes n},\rho^{\otimes n}_{1})$ converges to zero for $n\to\infty$ at an exponential rate equal to the Chernoff divergence $\xi(\rho_0,\rho_1)= -\log(\min_{0\le s\le1}\Tr(\rho_0^s\rho_1^{1-s}))>0$ \cite{ACMBMAV07,nussbaum2009chernoff}.} Hence, considering the infinite-resource state 
	$$\omega_{XQ} = p\kb{0}\otimes\kb{0} + (1-p)\kb{1}\otimes\kb{1}, $$ 
	with the quantum system being a qubit $Q$, and the measure-prepare channel $\cE_n(\cdot) = \Tr\!\left(\Lambda_n\cdot\right)\kb{0} + \Tr\!\left((\1-\Lambda_n)\cdot\right)\kb{1}$, we see that
	\begin{align*}
	\lim_{n\to\infty}\frac{1}{2}\left\|({\rm{id}}\otimes\cE_n)(\rho^{(n)}_{XA}) - \omega_{XQ}\right\|_1 =\lim_{n\to\infty}\Big( p\Tr\!\left((\1-\Lambda_n)\rho_0^{\otimes n}\right)+(1-p)\Tr\!\left(\Lambda_n\rho_1^{\otimes n})\right)\Big) =0.
	\end{align*}
	Therefore, as the infinite-resource state $\omega_{XQ}$ can be transformed to any other c-q state without error, we see that for every $\eps>0$, we can pick $n\in\N$ large enough such that for every c-q state $\sigma_{XB}\equiv(q,\sigma_0,\sigma_1)$ and $m\in\N$ there exists a CDS map $\cN$ such that
	\begin{align}
	\frac{1}{2}\left\|\cN(\rho_{XA}^{(n)})-\sigma_{XB}^{(m)}\right\|_1 \le \eps.
	\end{align} 
	Hence, the transformation $\rho^{(n)}_{XA}\equiv \left(p,\rho_0^{\otimes n},\rho_1^{\otimes n}\right) \mapsto \sigma^{(m)}_{XB}\equiv \left(q,\sigma_0^{\otimes m},\sigma_1^{\otimes m}\right)$ is possible at an infinite rate and with arbitrarily small error measured in trace distance. For $q=p$, this transformation can also be performed at an infinite rate by just using $\cptp_A$ operations.
	
    Note that $D'(\cdot, \cdot)$ is exactly defined in such a way that for every infinite-resource state $\omega_{XA}$ and $\rho_{XA}\neq\omega_{XA}$, we have $D'(\rho_{XA}, \omega_{XA})=\infty$. Therefore, the problem discussed above, in which we obtain unreasonable infinite rates in the transformations in the RTSD, does not occur when we choose $D'(\cdot, \cdot)$ as the error measure. Moreover, in the following, we see that $D'(\cdot, \cdot)$ has many desirable properties that lead to reasonable asymptotic rates in SD distillation, SD dilution, and the transformation of general elementary quantum sources.
	\end{remark} 
	
  The scaled trace distance $D'$ satisfies the data-processing inequality under CDS maps:
  \begin{lemma}[DPI for $D'$ under CDS maps]
  \label{lem:DPI-D-prime}
  Let $\rho_{XA},\sigma_{XA}$ be two c-q states and $\cN$ a CDS map. Then
  \begin{align}
  \label{eq:DPID'}
   D'(\cN(\rho_{XA}), \cN(\sigma_{XA})) \le D'(\rho_{XA}, \sigma_{XA}). 
  \end{align}
  \end{lemma}
\begin{proof}
 The statement directly follows from the data-processing inequality for the trace distance under general CPTP maps, i.e.,
 \begin{align}\frac{1}{2}\Big\|\cN(\rho_{XA}) -\cN(\sigma_{XA})\Big\|_1 \le\frac{1}{2}\Big\|\rho_{XA} - \sigma_{XA}\Big\|_1
 \end{align} 
 and the monotonicity of the minimum error probability under CDS maps proven in Lemma~\ref{lem:MonoErrProbCDS}, i.e.
 \begin{align}
 p_{\err}(\cN(\sigma_{XA})) \ge  p_{\err}(\sigma_{XA}).
 \end{align}
 This concludes the proof.
 \end{proof}
 
 The following lemma now establishes a bound relating the minimum error probabilities of two c-q states $\rho_{XA}$ and $\sigma_{XA}$, involving a multiplicative term related to their scaled trace distance $D^\prime$. This lemma is the key ingredient for proving all converses in approximate asymptotic SD-distillation, SD-dilution, and the transformation of general elementary quantum sources.

 \begin{lemma}
  	\label{lem:PerrDivBound}
  	Let $\rho_{XA}$ and $\sigma_{XA}$ be c-q states such that $D'(\rho_{XA}, \sigma_{XA})$ is finite. Then
  	\begin{align}
  	p_{\err}(\rho_{XA}) \le \Big(D'(\rho_{XA}, \sigma_{XA})+1\Big)p_{\err}(\sigma_{XA}).
  	\end{align}
  \end{lemma}
  
  \begin{proof}
  	First, it is helpful to note that by writing the c-q states explicitly as
  	\begin{align}
  	&\rho_{XA} = p\kb{0}\otimes\rho_0 + (1-p)\kb{1}\otimes\rho_1,\label{rxa}\\
  	&\sigma_{XA} = q\kb{0}\otimes\sigma_0 + (1-q)\kb{1}\otimes\sigma_1,
  	\end{align}
  	they can be block-diagonalised in the same basis and hence we can write
  	\begin{align}
  	\frac{1}{2}\left\|\rho_{XA} - \sigma_{XA}\right\|_1 = \frac{1}{2}\Big(\left\|p\rho_0 - q\sigma_0\right\|_1 +\left\|(1-p)\rho_1 - (1-q)\sigma_1\right\|_1\Big).
  	\end{align}
  	We can hence bound the minimum error probability of $\rho_{XA}$ as
  	\begin{align}
  	\nn p_{\err}(\rho_{XA}) &= \min_{0\le\Lambda\le \1}\Big(p\Tr\!\left(\Lambda\rho_0\right) + (1-p)\Tr\!\left((\1-\Lambda)\rho_1\right)\Big)\\
  	& = \nn\min_{0\le\Lambda\le \1}\Big(q\Tr\!\left(\Lambda\sigma_0\right) + (1-q)\Tr\!\left((\1-\Lambda)\sigma_1\right)\\&\nn\quad\quad+\Tr\!\left(\Lambda (p\rho_0-q\sigma_0)\right) + \Tr\!\left((\1-\Lambda) ((1-p)\rho_1-(1-q)\sigma_1)\right)
  	\Big)\\
  	& = \nn\min_{0\le\Lambda\le \1}\Big(q\Tr\!\left(\Lambda\sigma_0\right) + (1-q)\Tr\!\left((\1-\Lambda)\sigma_1\right) \\ 
  	  &\nn\quad\quad+\Tr\!\left((\kb{0} \otimes \Lambda + \kb{1}\otimes (\1-\Lambda) ) (\rho_{XA}-\sigma_{XA})\right)
  	\Big)\\
  	& \nn \le \min_{0\le\Lambda\le \1}\Big(q\Tr\!\left(\Lambda\sigma_0\right) + (1-q)\Tr\!\left((\1-\Lambda)\sigma_1\right)\Big)+
  	\frac{1}{2}\left\|\rho_{XA} - \sigma_{XA}\right\|_1 \\
  	%& \nn \le \min_{0\le\Lambda\le \1}\Big(q\Tr\!\left(\Lambda\sigma_0\right) + (1-q)\Tr\!\left((\1-\Lambda)\sigma_1\right)\Big)\\&\quad\quad+\frac{1}{2}\Big(\left\|p\rho_0 - q\sigma_0\right\|_1 + \left\|(1-p)\rho_1-(1-q)\sigma_1\right\|_1\Big)\nn \\
  	&=p_{\err}(\sigma_{XA}) + \frac{1}{2}\left\|\rho_{XA} - \sigma_{XA}\right\|_1 \nn
  	\\&= \Big(D'(\rho_{XA}, \sigma_{XA})+1\Big)p_{\err}(\sigma_{XA}).
  	\end{align}	
  	This concludes the proof.
  \end{proof}
 
 \comment{ \begin{remark}
  Recall that $D'(\cdot, \cdot)$ is a scaled trace distance, where the scaling factor is given by the inverse of the minimum error probability of the target state. This means that for a particular target state $\sigma_{XA}$, a ball of radius $\eps\ge 0$ with respect to $D'$, centered at this state, which is given by
  \begin{align}
\cB'_\eps(\sigma_{XA}) = \left\{\rho_{XA} \Big| D'(\rho_{XA}, \sigma_{XA}) \le p_{\err}(\sigma_{XA})\right\},
  \end{align} 
  is equal to the corresponding ball with respect to trace distance $\frac{1}{2}\|\cdot -\cdot\|_1$ with radius $p_{\err}(\sigma_{XA})\eps$ denoted by $\cB^1_{p_{\err}(\sigma_{XA})\eps}(\sigma_{XA})$, i.e.
  \begin{align}
    \cB'_\eps(\sigma_{XA}) = \cB^1_{p_{\err}(\sigma_{XA})\eps}(\sigma_{XA}).
  \end{align}

 Hence, we see that the higher the resource of the  target states $\sigma_{XA}$,  i.e.~the smaller $p_{\err}(\sigma_{XA})$, the more discrete does the topology around $\sigma_{XA}$ become. In particular, for an infinite resource target state, i.e.~a state $\sigma_{XA}$ for which $p_{\err}(\sigma_{XA})=0$, there is no other state which is at a finite $D'$ distance from it.

  Moreover, Lemma~\ref{lem:PerrDivBound} establishes that the amount of resource of the states within this $D'$ ball with a small radius $\eps$  cannot be significantly lower than the amount of resource of the target state $\sigma_{XA}$ itself.
  
  %the only state which has finite distance to $\sigma_{XA}$ is $\sigma_{XA}$ itself.
  
  \end{remark}
  }
  
\subsubsection{Semi-definite program for the scaled trace distance $D'$}

We now prove that the scaled trace distance $D^\prime(\cdot, \cdot)$ can be calculated by means of a semi-definite program (SDP). SDPs can be computed efficiently by numerical solvers~\cite{Vandenberghe1996}. As semi-definite programming is a powerful theoretical and numerical tool for quantum information theory, with a plethora of applications, we expect that the following SDP characterizations of $D^\prime(\cdot, \cdot)$ may be useful for a further understanding of this quantity.

\begin{proposition}
\label{prop:SDP-div-err}
For general c-q states $\rho_{XA}$ and $\sigma_{XA}$ with $p_{\err}(\sigma_{XA}) >0$, the scaled trace distance, $D^\prime(\rho_{XA}, \sigma_{XA})$, in Definition~\ref{def:err-divergence} is given by the following semi-definite program:
\begin{equation}
\label{eq:SDP-primal-div-err}
\begin{split}
\max \ & t \\
\text{s.t.} \ & -I_{XA}\le L_{XA}\le I_{XA},\\
& -tI_A\le P_A\le t I_A,\\
&t-\tr P_A(q\sigma_{0}-(1-q)\sigma_{1})=\tr L_{XA}(\rho_{XA}-\sigma_{XA}).
\end{split}
\end{equation}
The dual SDP is as follows:
\begin{equation}
\label{eq:SDP-dual-div-err}
\begin{split}
\min \ & \tr(B_{XA}+C_{XA}) \\
\text{s.t.} \ & B_{XA}, C_{XA}, D_A, E_A \ge 0, s\in \mathbb{R},\\
& B_{XA}-C_{XA} = s(\rho_{XA}-\sigma_{XA}),\\
& D_A-E_A = s(q\sigma_0-(1-q)\sigma_1),\\
& \tr (D_A+E_A)\le s-1.
\end{split}
\end{equation}
If $p_{\err}(\sigma_{XA}) >0$, then strong duality holds, so that the optimal value in \eqref{eq:SDP-primal-div-err} is equal to the optimal value in \eqref{eq:SDP-dual-div-err}.
\end{proposition}
The proof of Proposion~\ref{prop:SDP-div-err} can be found in Appendix~\ref{app:dual-derivation}.

\subsubsection{Semi-definite program for the minimum conversion error}

The one-shot transformation task from a source $\rho_{XA}\coloneqq p|0\rangle\!
\langle0|\otimes\rho_{0}+\left(  1-p\right)  |1\rangle\!\langle1|\otimes\rho
_{1}$\ to a target $\sigma_{XA^{\prime}}\coloneqq q|0\rangle\!\langle0|\otimes\sigma
_{0}+\left(  1-q\right)  |1\rangle\!\langle1|\otimes\sigma_{1}$ using $\cds$ as the set of free operations can be phrased
as the following optimization task:
\begin{equation}
d'_{CDS}(\rho_{XA}\mapsto\sigma_{XA'}) =\min_{\mathcal{N}_{XA\rightarrow XA^{\prime}}\in\text{CDS}}  D^{\prime
}(\mathcal{N}_{XA\rightarrow XA^{\prime}}(\rho_{XA}), \sigma_{XA^{\prime}
})  .\label{eq:opt-conv-error}
\end{equation}
We now prove that the minimum conversion error in \eqref{eq:opt-conv-error} can be calculated by means of a semi-definite program.

\begin{proposition}
\label{prop:SDP-conv-err}
The minimum conversion error in \eqref{eq:opt-conv-error} can be evaluated by
the following semi-definite program:%
\begin{equation}
\min_{\substack{B_{XA^{\prime}},C_{XA^{\prime}},D_{A^{\prime}},E_{A^{\prime}%
}\geq0,\\\Omega_{AA^{\prime}}^{0},\Omega_{AA^{\prime}}^{1}\geq0,s\geq
1}}\left\{
\begin{array}
[c]{c}%
\operatorname{Tr}[B_{XA^{\prime}}+C_{XA^{\prime}}]:\\
B_{XA^{\prime}}-C_{XA^{\prime}}=\tau_{XA^{\prime}}-s\sigma_{XA^{\prime}},\\
D_{A^{\prime}}-E_{A^{\prime}}=s(q\sigma_{0}-\left(  1-q\right)  \sigma_{1}),\\
\operatorname{Tr}[D_{A^{\prime}}+E_{A^{\prime}}]\leq s-1,\\
\operatorname{Tr}_{A^{\prime}}[\Omega_{AA^{\prime}}^{0}+\Omega_{AA^{\prime}%
}^{1}]=sI_{A},\\
\langle0|_{X}\tau_{XA^{\prime}}|0\rangle_{X}=\operatorname{Tr}_{A}[(p\rho
_{A}^{0})^{T}\Omega_{AA^{\prime}}^{0}]+\operatorname{Tr}_{A}[(\left(
1-p\right)  \rho_{A}^{1})^{T}\Omega_{AA^{\prime}}^{1}],\\
\langle1|_{X}\tau_{XA^{\prime}}|1\rangle_{X}=\operatorname{Tr}_{A}[(p\rho
_{A}^{0})^{T}\Omega_{AA^{\prime}}^{1}]+\operatorname{Tr}_{A}[(\left(
1-p\right)  \rho_{A}^{1})^{T}\Omega_{AA^{\prime}}^{0}],\\
\langle1|_{X}\tau_{XA^{\prime}}|0\rangle_{X}=\langle0|_{X}\tau_{XA^{\prime}%
}|1\rangle_{X}=0
\end{array}
\right\}.
\end{equation}

\end{proposition}
The proof of Proposion~\ref{prop:SDP-conv-err} can be found in Appendix~\ref{app:dual-derivation}.

\section{SD-distillation}

\label{sec:distil}

In this section, we study the fundamental task of distillation of symmetric distinguishability (SD-distillation), both in the one-shot and asymptotic settings.

	\subsection{One-shot exact SD-distillation}
	
	One-shot exact SD-distillation of a given c-q state 
	\begin{equation}\label{rho}
	\rho_{XA} \coloneqq p|0\rangle\!\langle 0|\otimes\rho_0+\left(  1-p\right)
	|1\rangle\!\langle 1|\otimes\rho_1,
	\end{equation}
	with $p\in [0,1]$ and $\rho_0, \rho_1$  states of a quantum system $A$, is the task of converting a single copy of it to an $M$-golden unit
	via free operations. The maximal value of $\log M$ for which this conversion is possible is equal to the \emph{one-shot exact distillable-SD} for the chosen set of free operations. This is defined formally as follows:
	
\begin{definition}
\label{def:exact-distillable-SD}
For a set of free operations  denoted by $FO$ and $q\in[0,1]$,
the one-shot exact distillable-SD of the c-q state $\rho_{XA}$ defined in~\eqref{rho} is given by
	\begin{align}
	\label{eq:DistillMistill}
	\xi_d^{\FO,q}(\rho_{XA}) \coloneqq  \log\Big(\sup\left\{M \, \Big| \, \cA\left(\rho_{XA}\right) = \gamma_{XQ}^{(M,q)},\,\cA \in  {\hbox{free operations (\FO)}}\, \right\}\Big).
	\end{align}
For the choice 
	 \begin{align}
	 \FO \equiv \Big\{{\rm{id}}\otimes\cE \, \Big| \, \cE\, \operatorname{CPTP} \text{on system A}\Big\}\equiv {\rm{CPTP}}_A ,
	 \end{align}
	 the only sensible choice in \eqref{eq:DistillMistill} is $q=p$, as free operations of the form $\rm{id}\otimes\cE$ cannot change the prior in the c-q state. In that case, the above quantity is called the {\em{one-shot exact distillable-SD under ${\rm{CPTP}}_A$ maps}} and we simply write
  	\begin{align}
  	  \xi_{d}(\rho_{XA}) \equiv \xi^{{\rm{CPTP}}_A,p}_{d}(\rho_{XA}).
  	\end{align}
  	Whereas for the choice ${\FO}\equiv {\hbox{CDS}}$ and $q=1/2$, the above quantity is called the {\em{one-shot exact distillable-SD under CDS maps}} and we use the notation
  	\begin{align}
  	    \xi^{\star}_d(\rho_{XA})\equiv\xi^{{\rm{CDS}},1/2}_{d}(\rho_{XA}).
  	\end{align}
  	Explicitly, for a c-q state $\rho_{XA}$ given by \eqref{rho}
  	we then have

\begin{align}
\label{eq:CPTP-A-Distill}
\xi_d(\rho_{XA})   = \log\Big(\sup\left\{M\Big| ({\rm{id}} \otimes \cE)\left(\rho_{XA}\right) = \gamma_{XQ}^{(M,p)},\,\cE \in \operatorname{CPTP}\, \right\}\Big)
\end{align}
and
\begin{align}
\label{eq:CDSDistill}
\xi_d^\star(\rho_{XA})  = \log\Big(\sup\left\{M\Big| \cN\left(\rho_{XA}\right) = \gamma_{XQ}^{(M)},\,\cN \in {\rm{CDS}}\, \right\}\Big),
\end{align}
where $\gamma_{XQ}^{(M)} \equiv \gamma_{XQ}^{(M, 1/2)}$ as stated previously.
\end{definition}

\begin{remark}
\label{rem:CPTPSDDistl}
  Note that in the case of prior $p\in(0,1)$ and the free operations being ${\rm{CPTP}}_A$, the distillable-SD is, by definition, independent of $p$. In fact, in that case it can be equivalently written as
  \begin{align}
   \xi_d(\rho_{XA})\equiv\xi_d(\rho_0,\rho_1) = \log\Big(\sup\left\{M\Big| \cE(\rho_0) = \pi_M\text{ and } \cE(\rho_1) = \sigma^{(1)}\pi_M\sigma^{(1)},\,\cE \in \operatorname{CPTP}\, \right\}\Big).  
  \end{align}
  For $p\in\{0,1\}$ one easily sees that
  \begin{align}
     \xi_d(\rho_{XA}) = \infty.
  \end{align}
Therefore, we restrict to the non-singular case $p\in(0,1)$ in the following Theorem~\ref{theo-distilCPTP}.
\end{remark}
\bigskip
\noindent	We state the following theorems now.

\begin{theorem}
\label{theo-distilCPTP}
The one-shot exact distillable-SD under ${\rm{CPTP}}_A$ maps of a c-q state $\rho_{XA}$, defined through~\eqref{rho}, with $p\in(0,1)$ is given by
    \begin{align}
        \xi_d(\rho_{XA})=  \xi_{\min}(\rho_{XA}) ,
    \end{align} 
where \begin{align}
	\label{eq:mindiv}
	\xi_{\min}(\rho_{XA}) \equiv \xi_{\min}(\rho_0,\rho_1)= -\log Q_{\min}(\rho_0,\rho_1) ,
	\end{align}
	and $Q_{\min}(\rho_0,\rho_1)$ is given by the following SDP:
	\begin{align}
	\label{eq:Qmin}
	Q_{\min}(\rho_0,\rho_1) \coloneqq 2\min\left\{\Tr\!\left(\Lambda\rho_0\right)\Big| \Tr\!\left(\Lambda(\rho_0+\rho_1)\right)= 1,\,\,0\le\Lambda\le\1\right\}.
	\end{align}
\end{theorem}	

\begin{remark}
      An alternative way of writing $Q_{\min}(\rho_0,\rho_1)$ is as follows:
      \begin{align}
      \label{eq:QminConstraint}
	\nn Q_{\min}(\rho_0,\rho_1) &= 2\min\left\{\Tr\!\left(\Lambda\rho_0\right)\Big| \Tr(\Lambda \rho_0)= \Tr((\1 - \Lambda)\rho_1) ,\,\,0\le\Lambda\le\1\right\}\\ &=2\min\left\{\Tr\!\left(\Lambda\rho_0\right)\Big| \Tr(\Lambda \rho_0)\ge \Tr((\1 - \Lambda)\rho_1) ,\,\,0\le\Lambda\le\1\right\}.
	\end{align}
	\comment{where the second equality follows from the argument around \eqref{eq:QminEquality} below.}
	This clarifies that the optimization is over all POVMs $\{\Lambda, \1 - \Lambda\}$ such that the Type~I error probability $\Tr(\Lambda \rho_0)$ is (greater than or) equal to the Type~II error probability $\Tr((\1 - \Lambda)\rho_1)$.
	
	To see the second equality in \eqref{eq:QminConstraint}, note first that the last line is trivially smaller than or equal to the right-hand side of the first line, since we are minimising over a larger set. To arrive at the other inequality, let $\tilde\Lambda_{\min}$ be a minimiser of the last line of \eqref{eq:QminConstraint} and $c\coloneqq \Tr(\tilde\Lambda_{\min}(\rho_0+\rho_1)) \ge 1$. Let $\Lambda_{\min} = \tilde\Lambda_{\min}/c$. Clearly $0\le \Lambda_{\min} \le \tilde \Lambda_{\min} \le \1$ and $\Tr(\Lambda_{\min}(\rho_0+\rho_1)) = 1$. Moreover,
	$$Q_{\min}(\rho_0,\rho_1)\le2\Tr\!\left(\Lambda_{\min}\rho_0\right) = \frac{2\Tr\!\left(\tilde\Lambda_{\min}\rho_0\right)}{c} \le 2\Tr\!\left(\tilde\Lambda_{\min}\rho_0\right) ,$$ from which we conclude the equality in \eqref{eq:QminConstraint}.
\end{remark}

\begin{theorem}\label{theo-distilCDS}
	The one-shot exact distillable-SD under CDS maps of a c-q state $\rho_{XA}$, defined through~\eqref{rho}, is given by 
	\begin{align}
	\label{eq:DistillCDSExact}
	\xi_d^\star(\rho_{XA}) = -\log(2p_{\operatorname{err}}(\rho_{XA}))={\rm{SD}}(\rho_{XA}).
	\end{align}
\end{theorem}

\begin{remark}
As a consequence of Theorem~\ref{theo-distilCPTP}, it follows that the one-shot exact distillable-SD under ${\rm{CPTP}}_A$ maps can be calculated by means of a semi-definite program, due to the form of $Q_{\min}(\rho_0,\rho_1)$ in \eqref{eq:Qmin}. As a consequence of Theorem~\ref{theo-distilCDS}, it follows that the one-shot exact distillable-SD under CDS maps can  be calculated by means of a semi-definite program, due to the expression for $p_{\operatorname{err}}(\rho_{XA})$ in \eqref{op-err}.
\end{remark}

Proofs of the above theorems are given in Sections~\ref{proof:distilCPTP} and~\ref{proof:distilCDS}, respectively. The quantities $\xi_{\min}$ and $Q_{\min}$ appearing in Theorem~\ref{theo-distilCPTP} have several useful and interesting properties, which are given in Section~\ref{sec:props}.

\medskip

\subsection{Optimal asymptotic rate of exact SD-distillation}

Consider the c-q state
$$
\rho_{XA}^{(n)}\coloneqq  p \kb{0} \otimes \rho_0^{\otimes n} + (1-p) \kb{1} \otimes \rho_1^{\otimes n},
$$
and let $\xi_d(\rho_{XB}^{(n)})$ and $\xi_d^\star(\rho_{XB}^{(n)})$ denote its one-shot exact distillable-SD under CPTP$_A$ maps and CDS maps, respectively. Then the {\em{optimal asymptotic rates of exact SD-distillation}} under CPTP$_A$ maps and CDS maps are defined by the following two quantities, respectively:
\begin{align}
\label{eq:Asbla}
\liminf_{n\to\infty}\frac{\xi_{d}(\rho_{XA}^{(n)})}{n}\quad ; \quad \liminf_{n\to\infty}\frac{\xi_{d}^\star(\rho_{XA}^{(n)})}{n}.
\end{align}

The next theorem asserts that both limits in \eqref{eq:Asbla} actually exist and are equal to the well-known quantum Chernoff divergence~\cite{ACMBMAV07,nussbaum2009chernoff}:
\begin{theorem}[Optimal asymptotic rate of exact SD-distillation]
\label{theo-asymp}
For $p\in(0,1)$, the optimal asymptotic rates of exact SD-distillation under ${\rm{CPTP}}_A$ and CDS maps are given by the following expression:
\begin{align}
	\label{eq_AsympResultDist}
	\lim_{n\to\infty}\frac{\xi_{d}(\rho_{XA}^{(n)})}{n}
	 =\lim_{n\to\infty}\frac{\xi_{d}^\star(\rho_{XA}^{(n)})}{n} 
		=\lim_{n\to\infty}\frac{\xi_{\min}(\rho_0^{\otimes n},\rho_1^{\otimes n})}{n} = \xi(\rho_0,\rho_1),
	\end{align}
	where $\xi(\rho_0,\rho_1) \coloneqq -\log\min_{0\le s\le 1}\Tr(\rho_0^s\rho_1^{1-s})$ denotes the quantum Chernoff divergence.
\end{theorem}
Here, the restriction to $p\in(0,1)$ is sensible as for $p\in\{0,1\}$ we directly get $\xi_{d}(\rho_{XA}^{(n)})=\xi^{\star}_{d}(\rho_{XA}^{(n)})=\infty$ for all $n\in\N.$

A proof of the above theorem is given in Section~\ref{proof-asymp}.

	\subsection{Properties of $Q_{{\min}}$ and $\xi_{{\min}}$}\label{sec:props}

	In this section, we establish some basic properties of the distinguishability measures $Q_{{\min}}$ and $\xi_{{\min}}$.
	\comment{ First of all, note that we can have a slightly changed constraint in the optimisation problem defining $Q_{\min}$ without changing the minimum, i.e.,
	\begin{align}
	\label{eq:QminEquality}
	\nn Q_{\min}(\rho_0,\rho_1) &\coloneqq  2\min\left\{\Tr\!\left(\Lambda\rho_0\right)\Big| \Tr\!\left(\Lambda(\rho_0+\rho_1)\right)= 1,\,\,0\le\Lambda\le\1\right\} \\
	&=2\min\left\{\Tr\!\left(\Lambda\rho_0\right)\Big| \Tr\!\left(\Lambda(\rho_0+\rho_1)\right)\ge 1,\,\,0\le\Lambda\le\1\right\}.
	\end{align}
	To see this, note first that the last line is trivially smaller than or equal to the right-hand side of the first line, since we are minimising over a larger set. To arrive at the other inequality, let $\tilde\Lambda_{\min}$ be a minimiser of the last  line of \eqref{eq:QminEquality} and $c\coloneqq \Tr(\tilde\Lambda_{\min}(\rho_0+\rho_1)) \ge 1$. Let $\Lambda_{\min} = \tilde\Lambda_{\min}/c$. Clearly $0\le \Lambda_{\min} \le \tilde \Lambda_{\min} \le \1$ and $\Tr(\Lambda_{\min}(\rho_0+\rho_1)) = 1$. Moreover,
	$$Q_{\min}(\rho_0,\rho_1)\le2\Tr\!\left(\Lambda_{\min}\rho_0\right) = \frac{2\Tr\!\left(\tilde\Lambda_{\min}\rho_0\right)}{c} \le 2\Tr\!\left(\tilde\Lambda_{\min}\rho_0\right) ,$$ from which we conclude the equality \eqref{eq:QminEquality}.}
	
	\begin{lemma}\label{symm}
	The distinguishability measures 
	   $Q_{\min}(\rho_0,\rho_1)$ and  $\xi_{\min} (\rho_0,\rho_1)$ are symmetric in their arguments. 
	\end{lemma}  
	
	\begin{proof}
	  To see this, note that if $\Lambda_{\min}$ is a minimiser for $Q_{\min}(\rho_0,\rho_1)$, i.e., satisfying $2\Tr(\Lambda_{\min}\rho_0) = Q_{\min}(\rho_0, \rho_1)$,  $0\le \Lambda_{\min}\le \1 $, and $\Tr(\Lambda_{\min}(\rho_0+\rho_1))= 1$, then $\tilde\Lambda_{\min} = \1 - \Lambda_{\min}$ also satisfies $0\le\tilde{\Lambda}_{\min}\le\1$ and $\Tr(\tilde\Lambda_{\min}(\rho_0+\rho_1))= 1$. Moreover, $2\Tr(\tilde\Lambda_{\min}\rho_1) = 2\Tr(\Lambda_{\min}\rho_0) = Q_{\min}(\rho_0,\rho_1)$. From that we see $$Q_{\min}(\rho_1,\rho_0)\le Q_{\min}(\rho_0,\rho_1).$$ The reversed inequality can be obtained by symmetry, which yields\begin{equation}
	Q_{\min}(\rho_0,\rho_1) = Q_{\min}(\rho_1,\rho_0).
	\end{equation}  
	By a straightforward consequence of the definition of $ \xi_{\min}(\rho_0,\rho_1)$ in terms of  $Q_{\min}(\rho_0,\rho_1)$, it follows that $ \xi_{\min}(\rho_0,\rho_1)$ is symmetric in its arguments.
	\end{proof}

\bigskip 
	The quantity $\xi_{\min}$ also satisfies a data-processing inequality (DPI) under CPTP maps, which is the statement of the following lemma:
	
	\begin{lemma}[DPI for $\xi_{\min}$ under CPTP maps]\label{lem:DPI}
		Let $\cE$ be a quantum channel. Then
		\begin{align}
		\label{eq:DPI}
		\xi_{\min}(\cE(\rho_0),\cE(\rho_1)) \leq \xi_{\min}(\rho_0,\rho_1).
		\end{align}
	\end{lemma}
	\begin{proof}
	Consider
	\begin{align}
	Q_{\min}(\cE(\rho_0),\cE(\rho_1)) &=  2\min\left\{\Tr\!\left(\Lambda\cE(\rho_0)\right)\Big| \nn\Tr\!\left(\Lambda\cE(\rho_0+\rho_1)\right)= 1,\,\,0\le\Lambda\le\1\right\} \\&
	 =2\min\left\{\Tr\!\left(\cE^*(\Lambda)\rho_0\right)\Big| \nn\Tr\!\left(\cE^*(\Lambda)(\rho_0+\rho_1)\right)= 1,\,\,0\le\Lambda\le\1\right\} \\
	& \nn\ge2\min\left\{\Tr\!\left(\Lambda\rho_0\right)\Big| \Tr\!\left(\Lambda(\rho_0+\rho_1)\right)= 1,\,\,0\le\Lambda\le\1 \right\}\\
	& =Q_{\min}(\rho_0, \rho_1).
	\end{align}
 In the last inequality,  we have used  that $0\le \cE^*(\Lambda)\le \1$ for each $0\le\Lambda\le \1$, implying that we are effectively minimising over a smaller set.
 Hence, directly from the definition of $\xi_{\min}$ in \eqref{eq:mindiv}, we conclude the data-processing inequality \eqref{eq:DPI}.
	\end{proof}
	
	\bigskip
The next lemma gives upper and lower bounds on $Q_{\min}$, which turn out to be the key ingredients for proving the asymptotic result in \eqref{eq_AsympResultDist}.
\begin{lemma}
	\label{lem:BoundsQ}
	For all $q\in[0,1]$, we have
	\begin{align}
	\label{eq:UpperLoweQBound}
	p_{\operatorname{err}}(q,\rho_0,\rho_1) \le \frac{1}{2} Q_{\min}(\rho_0,\rho_1) \le 2p_{\operatorname{err}}(1/2,\rho_0,\rho_1) .
	\end{align}
\end{lemma}
\begin{proof}
Let $\Lambda_{\min}$ be a minimiser of $Q_{\min}(\rho_0,\rho_1)$. As $\Tr(\Lambda_{\min}(\rho_0+\rho_1))= 1$, we have $\Tr(\Lambda_{\min}\rho_0) = \Tr((\1-\Lambda_{\min})\rho_1)$. Hence, for each $q\in[0,1]$
\begin{align}Q_{\min}(\rho_0,\rho_1)/2 &= \Tr(\Lambda_{\min}\rho_0) = q \Tr(\Lambda_{\min}\rho_0) +(1-q)\Tr((\1-\Lambda_{\min})\rho_1)\nn\\&\ge \min_{0\le\Lambda\le\1}\Big( q \Tr(\Lambda\rho_0) +(1-q)\Tr((\1-\Lambda)\rho_1)\Big) = p_{\operatorname{err}}(q,\rho_0,\rho_1).
\end{align}	
For the remaining inequality in \eqref{eq:UpperLoweQBound}, we use the specific choice of $$\Lambda = \left(\rho_0+\rho_1\right)^{-1/2}\rho_1\left(\rho_0+\rho_1\right)^{-1/2},$$
with $\left(\rho_0+\rho_1\right)^{-1/2}$ being defined via the pseudo inverse. Note that $\left\{\Lambda,\1-\Lambda\right\}$ forms the so-called \emph{pretty good measurement} \cite{Belavkin75,Belavkin75a,HolPGM78,Hausladen93bach,hausladen94}. 
As $$\left(\rho_0+\rho_1\right)^{-1}\left(\rho_0+\rho_1\right) = \Pi_{\rho_0+\rho_1},$$ with $\Pi_{\rho_0+\rho_1}$ being the projector onto the support of the operator $\rho_0+\rho_1$, we 
see that
 \begin{align}
 \label{eq:Achievbar}\Tr\!\left(\Lambda(\rho_0+\rho_1)\right) = \Tr(\rho_1\Pi_{\rho_0+\rho_1}) =1.
 \end{align} 
 In the above, we have used that the null spaces satisfy $$N(\rho_0+\rho_1) = N(\rho_0)\cap N(\rho_1),$$
 which follows by the positive semi-definiteness of $\rho_0$ and $\rho_1$. Moreover, by \eqref{eq:Achievbar} we already get $\Tr((\1-\Lambda)\rho_1) = \Tr(\Lambda\rho_0)$. 
 Hence, we see 
 \begin{align}
 Q_{\min}/2 \le \Tr(\Lambda\rho_0) = \frac{1}{2}\left(\Tr\!\left(\Lambda\rho_0\right)+ \Tr\!\left((\1 - \Lambda)\rho_1\right)\right) .
 \end{align}
Using the well known fact that the error probability of the pretty good measurement is upper bounded by twice the minimum error probability \cite{BK02} (also compare \cite{HaWi06,CaTaMaABa07}), i.e.
 \begin{align}
 \label{eq:PGMerrorbound}
 p^{\operatorname{pgm}}_{\operatorname{err}}(1/2,\rho_0,\rho_1) =\frac{1}{2}\left(\Tr\!\left(\Lambda\rho_0\right)+ \Tr\!\left(\1 - \Lambda)\rho_1\right)\right) \le 2p_{\operatorname{err}}(1/2,\rho_0,\rho_1),
 \end{align}
this completes the proof.

Note that in order to see \eqref{eq:PGMerrorbound},  \cite[Eq.~(13)]{BK02} gives the following lower bound on the guessing probability of the pretty good measurement:
\begin{align}
p_{\operatorname{guess}}(1/2,\rho_0,\rho_1)^2 \le p^{\operatorname{pgm}}_{\operatorname{guess}} (1/2,\rho_0,\rho_1),
\end{align}
where $p_{\operatorname{guess}}(1/2,\rho_0,\rho_1)$ is the optimal guessing probability in the discrimination task. This bound immediately implies \eqref{eq:PGMerrorbound} since $(1-x)^2 \ge 1-2x$ for every real number $x$.\comment{
\begin{align*}
1- 2p_{\err}(1/2,\rho_0,\rho_1) &\le 1- 2p_{\err}(1/2,\rho_0,\rho_1) + p_{\err}(1/2,\rho_0,\rho_1)^2 = (1-p_{\err}(1/2,\rho_0,\rho_1))^2 \\&\le 1-p_{\err}^{\operatorname{pgm}}(1/2,\rho_0,\rho_1).
\end{align*}}
\end{proof}

\begin{lemma}
For states $\rho_0$ and $\rho_1$, the following bound holds
 \begin{align}
\label{eq:lowerboundQCB}
\xi_{\min}(\rho_0,\rho_1) \ge \xi(\rho_0,\rho_1) -  1,
\end{align}   
where $\xi_{\min}(\rho_0,\rho_1)$ is defined in \eqref{eq:mindiv} and $\xi(\rho_0,\rho_1)$ in \eqref{eq:chernoff-div-defn}.
\end{lemma}
\begin{proof}
  From the Helstrom--Holevo Theorem and Lemma~\ref{lem:BoundsQ}, we conclude the upper bound 
\begin{align*}
Q_{\min}(\rho_0,\rho_1) \le 2- \|\rho_0-\rho_1\|_1.
\end{align*}
Using  Theorem~1 in \cite{ACMBMAV07}	(also compare Eq.~(6) therein) we see 
\begin{align*}
Q_{\min}(\rho_0,\rho_1) \le 2- \|\rho_0-\rho_1\|_1\le 2Q(\rho_0,\rho_1),
\end{align*}
with $Q(\rho_0,\rho_1) = \min_{0\le s\le1}\Tr(\rho_0^s\rho_1^{1-s})$.
This directly yields the desired lower bound on $\xi_{\min}$.
\end{proof}

%\medskip

%\noindent
	\begin{remark}
	  Using Lemma~\ref{symm}, $\xi_{\min}$ can be written as
	  \begin{align}
	   \xi_{\min}(\rho_0,\rho_1) &\coloneqq  -\log\!\left( \min\left\{\Tr\!\left(\Lambda\rho_1\right)\Big| \Tr\!\left(\Lambda(\rho_0+\rho_1)\right)= 1,\,\,0\le\Lambda\le\1\right\}\right)-1.
	   \end{align}
	  The use of the subscript in $\xi_{\min}$ is motivated by the similarity of the above expression (modulo the additive constant) with $D_{\min}$:
	\begin{align}
	D_{\min}(\rho_0\|\rho_1) &= -\log\min\left\{\Tr\!\left(\Lambda\rho_1\right)\Big| \Tr\!\left(\Lambda\rho_0\right)= 1,\,\,0\le\Lambda\le\1\right\}.
	\end{align}	  
	The notation $\xi_{\min}$ is further motivated by analogy with the resource theory of asymmetric distinguishability~\cite{Wang2019states}, where the quantity analogous to $\xi_{\min}$ is the min-relative entropy $D_{\min}$ \cite{D09}.
	\end{remark}
	%\section{Proofs of our results on one shot (exact)- and asymptotic rates of SD distillation}\label{sec:proofsThm1-2}
	%In this section we derive the expressions for the one-shot exact distillable-SD under ${\rm{CPTP}}_A$  maps and under CDS maps, as given in Theorem~\ref{theo-distilCPTP} and Theorem~\ref{theo-distilCDS}, respectively.
	
	\subsection{Proof of Theorem~\ref{theo-distilCPTP} --- One-shot exact distillable-SD under  ${\rm{CPTP}}_A$ maps}\label{proof:distilCPTP}
\begin{proof}
    We first prove the achievability, i.e., the lower bound $\xi_d(\rho_{XA}) \ge \xi_{\min}(\rho_{XA}) \equiv \xi_{\min}(\rho_0,\rho_1)$.\\
    
Let $\Lambda_{\min}$ be a minimiser of the optimisation problem corresponding to $Q_{\min}$ given in \eqref{eq:Qmin}, i.e., $2\Tr\!\left(\Lambda_{\min}\rho_0\right) = Q_{\min}(\rho_0,\rho_1)$. Let $\cE$ be the following measure-and-prepare channel:
\begin{align}
\cE(\rho) = \Tr((\1-\Lambda_{\min})\rho)\kb{0} + \Tr(\Lambda_{\min}\rho)\kb{1}.
\end{align}
Hence, using $\Tr(\Lambda_{\min}\rho_0) = \Tr((\1-\Lambda_{\min})\rho_1)$ we get
\begin{align}
(\1\otimes\cE)\left(\rho_{XA}\right) = p\kb{0}\otimes\pi_M +(1-p) \kb{1}\otimes\sigma^{(1)}\pi_M \sigma^{(1)},
\end{align}
with $M = {1}/({2\Tr(\Lambda_{\min}\rho_0))}$.
This implies that 
\begin{align}
\xi_{\min}(\rho_0,\rho_1) \coloneqq  -\log Q_{\min}(\rho_0,\rho_1) = -\log(2\Tr(\Lambda_{\min}\rho_0))  \le \xi_{d}(\rho_{XA}).
\end{align}

To obtain the converse, i.e., the reverse inequality $\xi_d(\rho_0,\rho_1) \le \xi_{\min}(\rho_0,\rho_1)$, we first note that for all $M\ge 0$ we have $\pi_M + \sigma^{(1)}\pi_M\sigma^{(1)}=\1$ and hence by picking $\Lambda = \kb{1}$ 
\begin{align*}
	Q_{\min}(\pi_M,\sigma^{(1)}\pi_M\sigma^{(1)}) &= 2\min\left\{\Tr\!\left(\Lambda\pi_M\right)\Big| \Tr\!\left(\Lambda\right)= 1,\,\,0\le\Lambda\le\1\right\} \\&\le 2\langle 1|\pi_M| 1\rangle = \frac{1}{M},
\end{align*}
and therefore
\begin{align*}
\xi_{\min}(\pi_M,\sigma^{(1)}\pi_M\sigma^{(1)}) \ge \log M.
\end{align*}
Now considering $\cE$ to be an arbitrary CPTP map such that 
\begin{align*}
(\1\otimes\cE)\left(\rho_{XA}\right)= p\kb{0}\otimes\pi_{M} +(1-p) \kb{1}\otimes\sigma^{(1)}\pi_{M} \sigma^{(1)}, 
\end{align*}
for some $M\ge1$, we see by the data processing inequality in Lemma~\ref{lem:DPI} that
\begin{align*}
\xi_{\min}(\rho_0,\rho_1) \ge \xi_{\min}(\cE(\rho_0),\cE(\rho_1))= \xi_{\min}(\pi_M,\sigma^{(1)}\pi_M\sigma^{(1)}) \ge \log M .
\end{align*}
As $\cE$ is an arbitrary CPTP map satisfying the constraint in \eqref{eq:DistillMistill}, we get
$$\xi_{\min}{(\rho_0,\rho_1)} \ge \xi_d{(\rho_{XA})},$$ and hence $\xi_{\min}(\rho_{XA}) \equiv \xi_{\min}(\rho_0,\rho_1)=\xi_d(\rho_{XA})$.
\end{proof}

\subsection{Proof of Theorem~\ref{theo-distilCDS} --- One-shot exact distillable-SD under CDS maps}

\label{proof:distilCDS}

\begin{proof}
We start with the achievability part, i.e., the lower bound 
$$\xi^\star_{d}(\rho_{XA}) \ge-\log(2p_{\operatorname{err}}(\rho_{XA})).$$
Let $\Lambda_{\min}$ be the minimiser of
\begin{align}
p_{\operatorname{err}}(\rho_{XA})  = \min_{0\le\Lambda\le\1}\Big(p\Tr\!\left(\Lambda\rho_0 \right)+ (1-p)\Tr\!\left((\1-\Lambda)\rho_1\right)\Big).
\end{align}
Define the quantum operations, i.e., completely positive, trace non-increasing maps
\begin{align}
\cE_0(\rho) = \frac{1}{2}\Big(\Tr((\1-\Lambda_{\min})\rho)\kb{0} + \Tr(\Lambda_{\min}\rho)\kb{1}\Big), \nn\\
\cE_1(\rho) = \frac{1}{2}\Big(\Tr((\1-\Lambda_{\min})\rho)\kb{1} + \Tr(\Lambda_{\min}\rho)\kb{0}\Big), 
\end{align}
and note that $\cE_0$ and $\cE_1$ sum to a CPTP map. Hence, we can define the corresponding CDS map as
\begin{align}
\cN = {\rm id}_X\otimes\cE_0 + {\cF}_X\otimes\cE_1,
\end{align}
 where ${\rm id}_X$ and $\cF_X$ denote the identity and flip channel on the the classical system $X$, respectively.
Noting that
\begin{align}
p\,\cE_0(\rho_0) + (1-p)\cE_1(\rho_1) \nn&= \frac{1}{2}\Big( p \Tr((\1-\Lambda_{\min})\rho_0) + (1-p)\Tr(\Lambda_{\min}\rho_1)\Big)\kb{0} \\
&\qquad +\frac{1}{2}\Big( p \Tr(\Lambda_{\min}\rho_0) + (1-p)\Tr\!\left((\1-\Lambda_{\min})\rho_1\right)\Big)\kb{1} \nn\\ &= \frac{1}{2}\Big((1- p_{\operatorname{err}}(\rho_{XA}))\kb{0} + p_{\operatorname{err}}(\rho_{XA})\kb{1}\Big) \nn\\&=\frac{1}{2} \pi_{{1}/{2p_{\operatorname{err}}(\rho_{XA})}},
\end{align}
and by symmetry
\begin{align}
p\,\cE_1(\rho_0) + (1-p)\cE_0(\rho_1) &= \frac{1}{2}\Big((1- p_{\operatorname{err}})\kb{1} + p_{\operatorname{err}}\kb{0}\Big)\nn\\& = \frac{1}{2}\sigma^{(1)}\pi_{{1}/{2p_{\operatorname{err}}(\rho_{XA})}}
\sigma^{(1)}.
\end{align}
Hence,
\begin{align}
\cN(\rho_{XA}) \nn&=\nn \kb{0}\otimes\left(p\cE_0(\rho_0) + (1-p)\cE_1(\rho_1)\right) 
+ \kb{1}\otimes\left(p\cE_1(\rho_0) + (1-p)\cE_0(\rho_1)\right)\nn\\
&= \frac{1}{2}\Big(\kb{0}\otimes\pi_{1/2p_{\operatorname{err}}(\rho_{XA})} + \kb{1}\otimes\sigma^{(1)}\pi_{1/2p_{\operatorname{err}}(\rho_{XA})}\sigma^{(1)}\Big) \nn\\&=
	\gamma_{XQ}^{(1/2p_{\operatorname{err}}(\rho_{XA}))} ,
	%\sigma^{(1/2p_{\operatorname{err}}(\rho_{XA}))}_{XA},
\end{align}
which implies 
\begin{align}
\xi^\star_d(\rho_{XA}) \ge -\log(2p_{\operatorname{err}}(\rho_{XA})).
\end{align}

To obtain the upper bound in \eqref{eq:DistillCDSExact}, we use monotonicity of the minimum error probability under CDS maps. More precisely, let $M\geq 1$  satisfy the constraint in \eqref{eq:CDSDistill}; i.e., there exists a CDS map $\cN$ such that
\begin{align}
\cN(\rho_{XA}) = \gamma_{XQ}^{(M)}.
\end{align}
Using the monotonicity of the minimum error probability $p_{\operatorname{err}}$ under CDS maps, we obtain
\begin{align}
-\log\!\left(2p_{\operatorname{err}}(\rho_{XA})\right)  &\nn\ge -\log\big(2p_{\operatorname{err}}(\cN(\rho_{XA}))\big)  = -\log\!\left(2p_{\operatorname{err}}(\gamma_{XQ}^{(M)})\right) \\&= \log M  .
\end{align}
As $M$ is arbitrary under the constraints in \eqref{eq:CDSDistill}, we have shown that
\begin{align}
 -\log(2p_{\operatorname{err}}(\rho_{XA})) \ge\xi^\star_d(\rho_{XA}),
\end{align}
which finishes the proof.
\end{proof}

\subsection{Proof of Theorem~\ref{theo-asymp} --- Optimal asymptotic rate of exact SD-distillation}

\label{proof-asymp}

\begin{proof}
We first prove the result in the case of free operations being ${\rm{CPTP}}_A$ maps.
As a consequence of Theorem~\ref{theo-distilCPTP}, we have the equality
\begin{align}
\frac{\xi_d(\rho^{(n)}_{XA})}{n} = \frac{\xi_{\min}(\rho_0^{\otimes n},\rho_1^{\otimes n})}{n},
\end{align}
and so it suffices to prove the asymptopic result for $\xi_{\min}$.
Using \eqref{eq:lowerboundQCB} and the fact that the quantum Chernoff divergence $\xi$ is additive, we get
\begin{align*}
\liminf_{n\to\infty} \frac{\xi_{\min}(\rho_0^{\otimes n},\rho_1^{\otimes n})}{n} \ge \xi(\rho_0,\rho_1).
\end{align*}
To get the other inequality, we use the lower bound in \eqref{eq:UpperLoweQBound} to see that for every $q\in(0,1)$
\begin{align}
\frac{\xi_{\min}(\rho_0^{\otimes n},\rho_1^{\otimes n})}{n} \le \frac{-\log(2p_{\operatorname{err}}(q,\rho_0^{\otimes n},\rho_1^{\otimes n}))}{n}.
\end{align}
Hence, by using the main result in \cite{ACMBMAV07}, we conclude that
\begin{align}
\limsup_{n\to\infty}\frac{\xi_{\min}(\rho_0^{\otimes n},\rho_1^{\otimes n})}{n} \le \lim_{n\to\infty}\frac{-\log(p_{\operatorname{err}}(q,\rho_0^{\otimes n},\rho_1^{\otimes n})) -1}{n} = \xi(\rho_0,\rho_1).
\end{align}
In the case of free operations being CDS maps, the equality
\begin{align}
\lim_{n\to\infty}\frac{\xi^\star_{d}(\rho^{(n)}_{XA})}{n} = \lim_{n\to\infty}\frac{-\log(p_{\operatorname{err}}(\rho^{(n)}_{XA})) -1}{n} = \xi(\rho_0,\rho_1) 
\end{align}
directly follows from Theorem~\ref{theo-distilCDS}, together with the main result of \cite{ACMBMAV07}, which finishes the proof.
\end{proof}
  
  \subsection{Approximate SD-distillation}
  
  We now define the \emph{one-shot approximate distillable-SD} for a general c-q state
	 \begin{align}
	  \rho_{XA} = p\kb{0}\otimes\rho_0 + (1-p)\kb{1}\otimes\rho_1.
	 \end{align}
  \begin{definition}
  For $\eps\ge0$ and golden unit $\gamma_{XQ}^{(M,q)}$, the \emph{one-shot approximate distillable-SD} of the c-q state $\rho_{XA}$ is given by
  \begin{align}
  \label{eq:oneApproxDistill}\nn
\xi^{\FO,q,\eps}_{d}(\rho_{XA}) &\coloneqq \log\Big(\sup\left\{ M\Big| d^\prime_{\operatorname{FO}}( \rho_{XA}\mapsto \gamma_{XQ}^{(M,q)})\le \eps \right\}\Big) \\
&=\log\Big(\sup\left\{ M\Big| D^\prime\!\left(\cA\left(\rho_{XA}\right), \gamma_{XQ}^{(M,q)}\right)\le \eps,\,\,\cA\in {\rm{FO}}\right\}\Big),
  \end{align}
  where the minimum conversion error $d^\prime_{\operatorname{FO}}$ is defined in Definition~\ref{def:conversiondistance}.
  	For the choice  \begin{align}{\hbox{\FO}}\equiv \Big\{{\rm{id}} \otimes\cE\Big| \cE \, \operatorname{CPTP} \text{  on system A}\Big\}\equiv {\rm{CPTP}}_A,\end{align}
  	the only sensible choice in \eqref{eq:oneApproxDistill} is $q=p$, as free operations of the form $\rm{id}\otimes\cE$ cannot change the prior in the c-q state. In that case, we simply write
  	\begin{align}
  	  \xi^{\eps}_{d}(\rho_{XA}) \equiv \xi^{{\rm{CPTP}}_A,p,\eps}_{d}(\rho_{XA}).
  	\end{align}
  	Whereas for the choice ${\hbox{\FO}}\equiv {\rm{CDS}}$ and $q=1/2$, we use the notation
  	\begin{align}
  	    \xi^{\star,\eps}_d(\rho_{XA})\equiv\xi^{{\rm{CDS}},1/2,\eps}_{d}(\rho_{XA}).
  	\end{align}
    \end{definition}
    
    \subsubsection{One-shot approximate distillable-SD as a semi-definite program}
    
    In this section, we prove that the one-shot approximate distillable-SD\ under $\cptp_A$ can be
evaluated by means of a semi-definite program and comment on SDP formulations of the one-shot approximate distillable-SD\ under $\cds$. Under $\cptp_A$ this quantity is defined as follows: \comment{Consider that this quantity is
defined for CPTP$_{A}$ maps as follows:}%
\begin{equation}
\xi_{d}^{\varepsilon}(\rho_{XA}) \coloneqq \log   \sup_{\mathcal{A}%
\in\text{CPTP}_{A}}\left\{  M\ \middle|\ D^{\prime}(\mathcal{A}(\rho_{XA}%
),\gamma_{XQ}^{(M,p)})\leq\varepsilon\right\}
,\label{eq:start-point-SDP-distill-CPTP-A}%
\end{equation}
and for CDS maps as
\begin{equation}
\xi_{d}^{\star,\varepsilon}(\rho_{XA}) \coloneqq 
\log  \sup_{\mathcal{A}%
\in\text{CDS}}\left\{  M\ \middle |\ D^{\prime}(\mathcal{A}(\rho_{XA}),\gamma
_{XQ}^{(M,1/2)})\leq\varepsilon\right\}    .
\end{equation}

We begin with the following:

\begin{proposition}
For all $\varepsilon\geq0$ and $p\in(0,1)$ and for every elementary source
described by the c-q state $\rho_{XA}$, the one-shot approximate distillable-SD under $\cptp_A$ maps can be evaluated by the following semi-definite program:
\begin{equation}
\xi_{d}^{\varepsilon}(\rho_{XA})=-\log\left[  2\inf_{\substack{0\le\Lambda\le\1,\\r\in\left[
0,1/2\right]  }}\Big\{
\begin{array}
[c]{c}%
r\Big|\,p\left\vert 1-r-\operatorname{Tr}[\Lambda\rho_{0}]\right\vert +\left(
1-p\right)  \left\vert r-\operatorname{Tr}[\Lambda\rho_{1}]\right\vert 
\leq\varepsilon\min\left\{  r,p,1-p\right\} 
\end{array}
\Big\}  \right]  .\label{eq:SDP-for-1-shot-distill-SD}%
\end{equation}

\end{proposition}

\begin{proof}
By definition the quantity $\xi_{d}^{\varepsilon}(\rho_{XA})$ is
equal to the negative logarithm of%
\begin{equation}
\inf_{\mathcal{A}\in\text{CPTP}_{A},r\in\left[  0,1\right]  }\left\{
r\ |\ D^{\prime}\left(\mathcal{A}(\rho_{XA}),\gamma_{XQ}^{(1/r,p)}\right)\leq
\varepsilon\right\}  .
\end{equation}
Let $\mathcal{A}\in$CPTP$_{A}$. Then by applying a completely dephasing
channel to the $Q$ system, the state $\gamma_{XQ}^{(1/r,p)}$ does not change,
whereas the local channel becomes a measurement channel of the following form:%
\begin{equation}
\mathcal{M}(\omega) \coloneqq \operatorname{Tr}[\Lambda\omega]|0\rangle\!\langle
0|_{Q}+\operatorname{Tr}[\left(  I-\Lambda\right)  \omega]|1\rangle\!\langle1|_Q.
\end{equation}
So we find that the optimal value is given by%
\begin{equation}
\inf_{\mathcal{M}}\left\{  r\ |\ D^{\prime}\left(\mathcal{M}(\rho_{XA}),\gamma
_{XQ}^{(1/r,p)}\right)\leq\varepsilon\right\}  ,
\end{equation}
as a consequence of the data-processing inequality for $D'$ given in Lemma~\ref{lem:DPI-D-prime} and with the optimization over every measurement channel $\mathcal{M}$. Now
consider that%
\begin{multline}
\mathcal{M}(\rho_{XA})=p|0\rangle\!\langle0|_{X}\otimes\operatorname{Tr}%
[\Lambda\rho_{0}]|0\rangle\!\langle0|_{Q}+p|0\rangle\!\langle0|_{X}\otimes
\operatorname{Tr}[\left(  I-\Lambda\right)  \rho_{0}]|1\rangle\!\langle1|_{Q}\\
+\left(  1-p\right)  |1\rangle\!\langle1|_{X}\otimes\operatorname{Tr}%
[\Lambda\rho_{1}]|0\rangle\!\langle0|_{Q}+\left(  1-p\right)  |1\rangle
\langle1|_{X}\otimes\operatorname{Tr}[\left(  I-\Lambda\right)  \rho
_{1}]|1\rangle\!\langle1|_{Q}%
\end{multline}
and%
\begin{equation}
D^{\prime}\left(\mathcal{M}(\rho_{XA}),\gamma_{XQ}^{(1/r,p)}\right)=\frac{\frac{1}%
{2}\left\Vert \mathcal{M}(\rho_{XA})-\gamma_{XQ}^{(1/r,p)}\right\Vert _{1}%
}{p_{\text{err}}(\gamma_{XQ}^{(1/r,p)})}.
\end{equation}
By applying \eqref{eq:simplify-p-err-golden-unit}, we know that%
\begin{align}
p_{\text{err}}(\gamma_{XQ}^{(1/r,p)})  & =\frac{1}{2}\left(  1-\left\vert
p-\frac{r}{2}\right\vert -\left\vert 1-p-\frac{r}{2}\right\vert \right)  \\
& =\min\left\{  \frac{r}{2},p,1-p\right\}  ,
\end{align}
where the last equality follows from a simplification that holds for
$p\in\left[  0,1\right]  $ and $r\in\left[  0,1\right]  $. Now consider that%
\begin{multline}
\gamma_{XQ}^{(1/r,p)}=p|0\rangle\!\langle0|_{X}\otimes\left(  1-\frac{r}%
{2}\right)  |0\rangle\!\langle0|_{Q}+p|0\rangle\!\langle0|_{X}\otimes\frac{r}%
{2}|1\rangle\!\langle1|_{Q}\\
+\left(  1-p\right)  |1\rangle\!\langle1|_{X}\otimes\frac{r}{2}|0\rangle\!
\langle0|_{Q}+\left(  1-p\right)  |1\rangle\!\langle1|_{X}\otimes\left(
1-\frac{r}{2}\right)  |1\rangle\!\langle1|_{Q}.
\end{multline}
So then we find that%
\begin{align}
  \left\Vert \mathcal{M}(\rho_{XA})-\gamma_{XQ}^{(1/r,p)}\right\Vert
_{1}
&  =p\left\vert 1-\frac{r}{2}-\operatorname{Tr}[\Lambda\rho_{0}]\right\vert
+p\left\vert \frac{r}{2}-\operatorname{Tr}[\left(  I-\Lambda\right)  \rho
_{0}]\right\vert \nonumber\\
&  \qquad+\left(  1-p\right)  \left\vert \frac{r}{2}-\operatorname{Tr}%
[\Lambda\rho_{1}]\right\vert +\left(  1-p\right)  \left\vert \left(
1-\frac{r}{2}\right)  -\operatorname{Tr}[\left(  I-\Lambda\right)  \rho
_{1}]\right\vert \\
&  =2p\left\vert 1-\frac{r}{2}-\operatorname{Tr}[\Lambda\rho_{0}]\right\vert
+2\left(  1-p\right)  \left\vert \frac{r}{2}-\operatorname{Tr}[\Lambda\rho
_{1}]\right\vert .
\end{align}
So the optimization problem is equivalent to the following:%
\begin{equation}
\inf_{\Lambda\geq0,r\in\left[  0,1\right]  }\left\{
\begin{array}
[c]{c}%
r:p\left\vert 1-\frac{r}{2}-\operatorname{Tr}[\Lambda\rho_{0}]\right\vert
+\left(  1-p\right)  \left\vert \frac{r}{2}-\operatorname{Tr}[\Lambda\rho
_{1}]\right\vert \\
\qquad\leq\varepsilon\min\left\{  \frac{r}{2},p,1-p\right\}  ,\ \Lambda\leq I
\end{array}
\right\}  .
\end{equation}
Let us make the substitution $r\rightarrow2r$, and the above becomes the
following:%
\begin{equation}
2\inf_{\Lambda\geq0,r\in\left[  0,1/2\right]  }\left\{
\begin{array}
[c]{c}%
r:p\left\vert 1-r-\operatorname{Tr}[\Lambda\rho_{0}]\right\vert +\left(
1-p\right)  \left\vert r-\operatorname{Tr}[\Lambda\rho_{1}]\right\vert \\
\qquad\leq\varepsilon\min\left\{  r,p,1-p\right\}  ,\ \Lambda\leq I
\end{array}
\right\}  .
\end{equation}
This concludes the proof.
\end{proof}

\begin{remark}
We can see that if $\varepsilon=0$ on the right-hand side of
\eqref{eq:SDP-for-1-shot-distill-SD}, this expression reduces to the $\xi_{\min}$ quantity
defined in \eqref{eq:mindiv}.
\end{remark}

\comment{\begin{remark} An alternate SDP formulation for
the one-shot approximate distillable-SD under $\cptp_A$, stemming from an SDP formulation for the minimum conversion error, is also possible. Similarly, the one-shot approximate distillable-SD under CDS maps can also be formulated as an SDP. We do not include details of these SDPs here, but note that they are available in the arXiv posting of our paper \cite{RTSDarxiv}.
\end{remark}}

We now provide an alternate characterization of the semi-definite program for
the one-shot approximate distillable-SD in Proposition~\ref{prop:alt-SDP-one-shot-distillable-SD} below. The proof is given
in Appendix~\ref{app:proof-one-shot-distillable-SD-SDP} and starts from the SDP for the minimum conversion error in \eqref{prop:SDP-div-err}.

\begin{proposition}
\label{prop:alt-SDP-one-shot-distillable-SD}
The approximate one-shot distillable-SD\ under CPTP$_{A}$ maps can be
calculated by means of the following semi-definite program:%
\begin{equation}
\xi_{d}^{\varepsilon}(\rho_{XA})=-\log \inf_{\substack{c^{0,0}%
,c^{0,1},c^{1,0},c^{1,1}\geq0,\\e^{0},e^{1}\geq0,\Lambda\geq0,r\in\left[
0,1\right]  }}\left\{
\begin{array}
[c]{c}
r:\\
2\sum_{i,j\in\left\{  0,1\right\}  }c^{i,j}\leq\varepsilon\left(  1-p-\left(
e^{0}+e^{1}\right)  \right)  ,\\
c^{0,0}\geq-p\left(  \operatorname{Tr}[\Lambda\rho_{0}]-\left(  1-\frac{r}%
{2}\right)  \right)  ,\\
c^{0,1}\geq p\left(  \operatorname{Tr}[\Lambda\rho_{0}]-\left(  1-\frac{r}%
{2}\right)  \right)  ,\\
c^{1,0}\geq-\left(  1-p\right)  \left(  \operatorname{Tr}[\Lambda\rho
_{1}]-\frac{r}{2}\right)  ,\\
c^{1,1}\geq\left(  1-p\right)  \left(  \operatorname{Tr}[\Lambda\rho
_{1}]-\frac{r}{2}\right)  ,\\
e^{0}\geq\frac{r}{2}-p,\quad e^{1}\geq\left(  1-p\right)  -\frac{r}{2}%
,\quad\Lambda\leq I
\end{array}
\right\}  .
\end{equation}

\end{proposition}

We find that the one-shot approximate distillable-SD under CDS maps can be evaluated by a semi-definite program as well, and the proof of Proposition~\ref{prop:one-shot-distillable-SD-CDS-SDP} below is given in Appendix~\ref{app:approx-one-shot-distill-CDS}:

\begin{proposition}
\label{prop:one-shot-distillable-SD-CDS-SDP}
The approximate one-shot distillable-SD\ under CDS maps can be calculated by
means of the following semi-definite program:%
\begin{equation}
\xi_{d}^{\star,\varepsilon}(\rho_{XA})=
-\log \inf_{\substack{c^{i,j},\Lambda^{i,j}\geq0,\\r\in\left[
0,1\right]  }}\left\{
\begin{array}
[c]{c}%
r:\\
2\sum_{i,j\in\left\{  0,1\right\}  }c^{i,j}\leq\varepsilon r,\\
p\operatorname{Tr}[\Lambda^{0,0}\rho_{0}]+\left(  1-p\right)
\operatorname{Tr}[\Lambda^{1,0}\rho_{1}]-\frac{1}{2}\left(  1-\frac{r}%
{2}\right)  +c^{0,0}\geq0,\\
p\operatorname{Tr}[\Lambda^{0,1}\rho_{0}]+\left(  1-p\right)
\operatorname{Tr}[\Lambda^{1,1}\rho_{1}]-\frac{r}{4}+c^{0,1}\geq0,\\
p\operatorname{Tr}[\Lambda^{1,0}\rho_{0}]+\left(  1-p\right)
\operatorname{Tr}[\Lambda^{0,0}\rho_{1}]-\frac{r}{4}+c^{1,0}\geq0,\\
p\operatorname{Tr}[\Lambda^{1,1}\rho_{0}]+\left(  1-p\right)
\operatorname{Tr}[\Lambda^{0,1}\rho_{1}]-\frac{1}{2}\left(  1-\frac{r}%
{2}\right)  +c^{1,1}\geq0,\\
I=\Lambda^{0,0}+\Lambda^{0,1}+\Lambda^{1,0}+\Lambda^{1,1}%
\end{array}
\right\}  .
\end{equation}

\end{proposition}

    \subsubsection{Optimal asymptotic rate of approximate SD-distillation}
 
We now consider the asymptotic case of approximate SD-distillation. Theorem~\ref{theo-asymp} already established that in the exact case the optimal rates for free operations being $\cds$ or $\cptp_A$ are given by the quantum Chernoff divergence.  
The following theorem shows that also when an error is allowed in the transformation, the corresponding asymptotic rates are still given by the quantum Chernoff divergence.

 \begin{theorem}[Optimal asymptotic rate of approximate SD-distillation]\label{theo-asympAprDistil}
 For all $\eps\ge0$ and $p\in(0,1)$, the optimal asymptotic rates of approximate SD-distillation under $\cptp_A$and $\cds$ maps are given by
 \begin{align}
	\label{eq_AsympResult}
	\lim_{n\to\infty}\frac{\xi^{\eps}_{d}(\rho_{XA}^{(n)})}{n}
	 =\lim_{n\to\infty}\frac{\xi^{\star,\eps}_{d}(\rho^{(n)}_{XA})}{n} = \xi(\rho_0,\rho_1),
	\end{align}
	where $\xi(\rho_0,\rho_1) = -\log\!\left(\min_{0\le s\le1}\Tr\!\left(\rho_0^s\rho_1^{1-s}\right)\right)$ denotes the quantum Chernoff divergence.
 \end{theorem}
 Here, the restriction to $p\in(0,1)$ is sensible, as for $p\in\{0,1\}$, we directly get $\xi^{\eps}_{d}(\rho_{XA}^{(n)})=\xi^{\star,\eps}_{d}(\rho_{XA}^{(n)})=\infty$ for all $n\in\N.$
 \begin{remark}
 Note that Theorem~\ref{theo-asympAprDistil} establishes the strong converse property for the task of asymptotic distillation in the RTSD. Another way of interpreting this statement is as follows: for a sequence of SD distillation protocols with rate above the asymptotic distillable-SD, the error necessarily converges to $\infty$ in the limit as $n$ becomes large.
 \end{remark}
 
 \begin{proof}[Proof of Theorem~\ref{theo-asympAprDistil}]
  As we have, for every c-q state $\rho_{XA}$,
  \begin{align}
  &\xi^{\eps}_d(\rho_{XA}) \ge\xi_d(\rho_{XA}),\\
  &\xi^{\star,\eps}_d(\rho_{XA}) \ge\xi^\star_d(\rho_{XA}),
  \end{align}
  we immediately get the lower bounds
  \begin{align}
\liminf_{n\to\infty}\frac{\xi^{\eps}_{d}(\rho_{XA}^{(n)})}{n} & \ge \xi(\rho_0,\rho_1),\\
\liminf_{n\to\infty}\frac{\xi^{\star,\eps}_{d}(\rho_{XA}^{(n)})}{n} & \ge \xi(\rho_0,\rho_1),
  \end{align}
  from Theorem~\ref{theo-asymp}.
  
  To establish the upper bounds, we use Lemma~\ref{lem:PerrDivBound}. We only prove the upper bound for $\xi^{\star,\eps}_d$, as the one for $\xi^\eps_d$ exactly follows the same lines.
  Let $M$ be such that it satisfies the constraint in the following optimization:
  \begin{align}
  \label{eq:xistardeltainproof}
  \xi^{\star,\eps}_d(\rho_{XA}^{(n)}) = \log\!\left(\sup\left\{ M \Big|\,d^\prime_{\operatorname{CDS}}(\rho^{(n)}_{XA}\mapsto\gamma_{XQ}^{(M)})\le\eps\right\}\right).
  \end{align}
  Hence, there exists a CDS map $\cN$ such that$$
D^\prime\!\left(\cN(\rho_{XA}^{(n)}),\gamma_{XQ}^{(M)}\right)\le\eps.$$ Then, by the monotonicity of the minimum error probability under CDS maps and the bound in Lemma~\ref{lem:PerrDivBound}, we see
\begin{align}
 p_{\operatorname{err}}\!\left(\rho^{(n)}_{XA}\right) \le p_{\operatorname{err}}\!\left(\cN(\rho^{(n)}_{XA})\right) \le \left(\eps+1\right)p_{\operatorname{err}}\!\left(\gamma_{XQ}^{(M)}\right)= \frac{\eps+1}{ 2M }.
\end{align}
As $M$ is arbitrary under the constraint in \eqref{eq:xistardeltainproof}, we get
\begin{align}
\xi^{\star,\eps}_{d}(\rho_{XA}^{(n)}) \le \left(-\log(p_{\operatorname{err}}(\rho_{XA}^{(n)})) - 1\right) + \log(\eps+1),
\end{align}
and hence
\begin{align}
\limsup_{n\to\infty}\frac{\xi^{\star,\eps}_{d}(\rho_{XA}^{(n)})}{n} \le \lim_{n\to\infty} \frac{\left(-\log(p_{\operatorname{err}}(\rho_{XA}^{(n)})) - 1\right) + \log(\eps+1)}{n} = \xi(\rho_0,\rho_1),
\end{align}
which finishes the proof of Theorem~\ref{theo-asympAprDistil}.
 \end{proof}
 
\section{SD-dilution}

\label{sec:SD-dilution}

We now turn to the case of dilution of symmetric distinguishability. We begin with the exact one-shot case in Section~\ref{sec:one-shot-exact-dilution}, establish some properties of relevant divergences in Section~\ref{prop-max}, provide proofs in Sections~\ref{sec:proof-OneshotdiluteCPTP} and \ref{sec:proofOneshotdiluteCDS}, consider the one-shot approximate case in Section~\ref{sec:one-shot-approx-dilution}, and evaluate asymptotic quantities in Sections~\ref{sec:asymptotic-dilution} and \ref{sec:proofs-asymptotic-dilution}.

	\subsection{One-shot exact SD-dilution}
	
	\label{sec:one-shot-exact-dilution}
	
		One-shot exact SD-dilution of a given c-q state 
	\begin{equation}\label{rho1}
	\rho_{XA}\coloneqq p|0\rangle\!\langle 0|\otimes\rho_0+\left(  1-p\right)
	|1\rangle\!\langle 1|\otimes\rho_1,
	\end{equation}
	with $p\in [0,1]$, is the task of converting an $M$-golden unit to the target state $\rho_{XA}$
	via free operations. The minimal value of $\log M$ for which this conversion is possible is equal to the \emph{one-shot exact SD-cost} for the chosen set of free operations. This is formally defined as follows:
	
	\begin{definition}
	\label{def:exact-SD-cost}
	For a set of free operations denoted by $\FO$ and $q\in[0,1]$, the one-shot exact SD-cost of the c-q state $\rho_{XA}$ defined in~\eqref{rho1} is given by
	\end{definition}
	\begin{align}
\label{eq:DiluteMilute}
\xi_c^{\operatorname{FO},q}(\rho_{XA}) \coloneqq  \log\Big(\inf\left\{M\Big| \cA\left(\gamma_{XQ}^{(M,q)}\right) = \rho_{XA},\,\cA \in {\hbox{free operations (FO)}}\, \right\}\Big).
\end{align}

For the choice
\begin{align}
	{\hbox{\FO}}\equiv \Big\{{\rm{id}} \otimes\cE \, \Big| \,  \cE\text{ CPTP on system } A\Big\}\equiv {\rm{CPTP}}_A,
\end{align}
the only sensible choice in \eqref{eq:DiluteMilute} is $q=p$, as free operations of the form $\rm{id}\otimes\cE$ cannot change the prior in the c-q state. In that case, the above quantity is called the {\em{one-shot exact SD-cost under ${\rm{CPTP}}_A$ maps}} and we simply write
  	\begin{align}
  	\label{eq:CostCPTPA}
  	  \xi_{c}(\rho_{XA}) \equiv \xi^{{\rm{CPTP}}_A,p}_{c}(\rho_{XA}).
  	\end{align}
  	Whereas for the choice ${\hbox{FO}}\equiv {\hbox{CDS}}$ and $q=1/2$, the above quantity is called the {\em{one-shot exact SD-cost under CDS maps}} and we use the notation
  	\begin{align}
  	    \xi^{\star}_c(\rho_{XA})\equiv\xi^{{\rm{CDS}},1/2}_{c}(\rho_{XA}).
  	\end{align}
  	Explicitly, for a c-q state $\rho_{XA}$ given by \eqref{rho1},
  	we then have
 \begin{align}
\label{eq:DiluteMiluteCPTPA}
\xi_c(\rho_{XA}) = \log\Big(\inf\left\{M\Big| ({\rm{id}}\otimes\cE)\left(\gamma_{XQ}^{(M,p)}\right) = \rho_{XA},\,\cE \in {\rm{CPTP}}\, \right\}\Big)
\end{align}	
and
 \begin{align}
\label{eq:DiluteMiluteCDS}
\xi^\star_c(\rho_{XA}) = \log\Big(\inf\left\{M\Big| \cN\left(\gamma_{XQ}^{(M)}\right) = \rho_{XA},\,\cN \in {\rm{CDS}}\, \right\}\Big),
\end{align}	
with $\gamma_{XQ}^{(M)} \equiv \gamma_{XQ}^{(M, 1/2)}$ as defined previously.
\begin{remark}
\label{rem:CPTPSDDil}
  Note that in the case of the free operations being ${\rm{CPTP}}_A$ and prior $p\in(0,1)$ the SD-cost, just as the distillable-SD (see Remark~\ref{rem:CPTPSDDistl}), is independent of $p$ by definition. In fact, in that case, it can be equivalently written as
  \begin{align}
   \xi_c(\rho_{XA})\equiv\xi_c(\rho_0,\rho_1) = \log\Big(\inf\left\{M\Big| \cE(\pi_M) = \rho_0\text{ and } \cE(\sigma^{(1)}\pi_M\sigma^{(1)}) = \rho_1,\,\cE\in \operatorname{CPTP} \, \right\}\Big).  
  \end{align}
  For $p\in \{0,1\}$, one easily sees that
  \begin{align}
     \xi_c(\rho_{XA}) = 0.
  \end{align}
Therefore, we restrict to the non-singular case $p\in(0,1)$ in the following Theorem~\ref{theo-diluteCPTP}.
\end{remark}

		%We prove the following theorems:
		
\begin{theorem}\label{theo-diluteCPTP}
The one-shot exact SD-cost under ${\rm{CPTP}}_A$ maps of a c-q state $\rho_{XA}$ with $p\in(0,1)$, as defined through~\eqref{rho1},  is given by
    \begin{align}
        \xi_c(\rho_{XA}) \equiv \xi_c(\rho_0,\rho_1) = \xi_{\max}(\rho_0,\rho_1),
    \end{align} 
where 
\begin{align}
\label{eq:maxdiv}
\xi_{\max}(\rho_0,\rho_1)\coloneqq  \log Q_{\max}(\rho_0,\rho_1),
\end{align}
and
\begin{align}
\label{eq:Qmax}
\nn Q_{\max}(\rho_0,\rho_1) &\coloneqq  \inf\left\{M\Big| \rho_0\le (2M-1)\rho_1,\,\, \rho_1\le (2M-1)\rho_0\right\} \\&=
\frac{1}{2}\left(2^{d_T(\rho_0,\rho_1)} + 1\right)\ge 1.
\end{align}
Here, for two positive semi-definite operators $\omega_1$, $\omega_2$,
\begin{align}\label{Thompson}
  d_T(\omega_1, \omega_2) \coloneqq    \max\Big\{D_{\max}(\omega_1\|\omega_2),D_{\max}(\omega_2\|\omega_1)\Big\} 
\end{align}
denotes the Thompson metric \cite{Thomp63}.
\end{theorem}

\begin{remark}
	  Note that $\xi_{\max}$ can be written as
	  \begin{align}
	   \xi_{\max}(\rho_0,\rho_1) &=  \max\left\{\log\!\left( 2^{D_{\max}(\rho_0,\rho_1)}+1\right),\,\log\!\left(2^{D_{\max}(\rho_1,\rho_0)}+1\right)\right\}-1.
	   \end{align}
	  The use of the subscript in $\xi_{\max}$ is motivated by the fact that it is a divergence that is essentially a symmetrized version of $D_{\max}$.

	The notation $\xi_{\max}$ is further motivated by analogy with the resource theory of asymmetric distinguishability~\cite{Wang2019states}, in which the quantity analogous to $\xi_{\max}$, arising as the cost of exact dilution of asymmetric distinguishability, is the max-relative entropy $D_{\max}$ \cite{D09}.
	
	We also note that Theorem~\ref{theo-diluteCPTP} can be infered from \cite[Lemma 3.1]{BST19}, but we include a self-contained proof below for completeness.
	\end{remark}

\begin{theorem}
\label{theo-diluteCDS}
The one-shot exact SD-cost under CDS maps of a c-q state $\rho_{XA}$, defined through~\eqref{rho1}, is given by
    \begin{align}
        \xi_c^\star(\rho_{XA}) = \xi_{\max}^\star(\rho_{XA}),
    \end{align} 
where 
\begin{align}
\label{eq:maxdivCDS}
\xi^\star_{\max}\left(\rho_{XA}\right)= \log Q^\star_{\max}\left(\rho_{XA}\right),
\end{align}
and 
\begin{align}
\label{eq:QmaxCDS1}
\nn Q^\star_{\max}(\rho_{XA}) &\coloneqq  \inf\left\{M\Big| p\rho_0\le (2M-1)(1-p)\rho_1,\,\, (1-p)\rho_1\le (2M-1)p\rho_0\right\} \\
&= \frac{1}{2}\Big( 2^{d_T(p\rho_0, (1-p)\rho_1)} + 1\Big)\ge \frac{1}{2}\max\left\{1/p,1/(1-p)\right\}.
\end{align}
\end{theorem}

\begin{remark}
Note that Theorems~\ref{theo-diluteCPTP} and \ref{theo-diluteCDS} give an operational interpretation to the Thompson metric in the context of exact one-shot dilution in the resource theory of symmetric distinguishability. To the best of our knowledge, this is the first time an operational meaning in quantum information has been given to the Thompson metric.
\end{remark}

\begin{remark}
As a consequence of Theorem~\ref{theo-diluteCPTP}, it follows that the one-shot exact SD-cost under ${\rm{CPTP}}_A$ maps can be calculated by means of a semi-definite program, due to the expression in \eqref{eq:Qmax}. As a consequence of Theorem~\ref{theo-diluteCDS}, it follows that the one-shot exact SD-cost under CDS maps can  be calculated by means of a semi-definite program, due to the expression in \eqref{eq:QmaxCDS1}.
\end{remark}

Properties of $\xi_{\max}$,  $Q_{\max}$, $\xi^\star_{\max}$, and $Q^\star_{\max}$ are discussed in Section~\ref{prop-max}. Proofs of the above theorems are given in Sections~\ref{sec:proof-OneshotdiluteCPTP} and \ref{sec:proofOneshotdiluteCDS}.

\subsection{Properties of $\xi_{\max}$, $Q_{\max}$, $\xi_{\max}^\star$, and $Q^\star_{\max}$}

\label{prop-max}

The quantity $\xi_{\max}$ satisfies the data-processing inequality under CPTP maps:
\begin{lemma}[DPI for $\xi_{\max}$]
	\label{lem:DPImax}
	Let $\cE$ be a CPTP map. Then
	\begin{align}
	\label{eq:DPImax}
	\xi_{\max}(\cE(\rho_0),\cE(\rho_1)) \le \xi_{\max}(\rho_0,\rho_1).
	\end{align}
\end{lemma}
\begin{proof}
Follows immediately from the data-processing inequality for $D_{\max}$ under CPTP maps \cite{D09}. 
\end{proof}

\bigskip
We now prove that $\xi^\star_{\max}$ is decreasing under CDS maps. For that, we find that $\xi^\star_{\max}$ can also be expressed in the following ways, as a consequence of the definition \eqref{eq:QmaxCDS1} of $Q^\star_{\max}$:
\begin{align}\label{eq:QmaxCDS}
Q^\star_{\max}(\rho_{XA}) &= \frac{1}{2}\Big(\max\left\{2^{D_{\max}(p\rho_0\|(1-p)\rho_1)},2^{D_{\max}((1-p)\rho_1\|p\rho_0)}\right\} + 1\Big) \nn\\&\nn
=\frac{1}{2}\Big(\max\left\{2^{D_{\max}(\rho_0\|\rho_1)+\log\left(p/(1-p)\right)},2^{D_{\max}(\rho_1\|\rho_0)+\log\left((1-p)/p\right)}\right\} + 1 \Big)\\&= \frac{1}{2}\Big(2^{D_{\max}(\rho_{XA}\|\left(\cF_X\otimes\rm{id}_A\right)\left(\rho_{XA}\right))}+1\Big),
\end{align}
where $\cF_X(\cdot) = \sigma^{(1)}\cdot\sigma^{(1)}$ denotes the flip channel on the system $X$.
%\smallskip

%\noindent
\begin{lemma}
\label{lem:xistarmon}
	Let $\cN \in {\rm{CDS}}$. Then
	\begin{align}
	\label{eq:DPICDS}
	\xi^\star_{\max}(\cN(\rho_{XA})) \le \xi^\star_{\max}(\rho_{XA}).
	\end{align}
\end{lemma}
\begin{proof}
Noting that every CDS channel $\cN$ commutes with the channel $\cF_X\otimes{\rm{id}}_A$, \eqref{eq:DPICDS} follows directly from the data-processing inequality for $D_{\max}$.
\end{proof}

%\section{Proofs of our results on one-shot exact SD dilution}\label{proof-dilute}
%	In this section we derive the expressions for the one-shot exact SD-costs under ${\rm{CPTP}}_A$ maps and CDS maps, as given in Theorem~\ref{theo-diluteCPTP} and Theorem~\ref{theo-diluteCDS}, respectively.

\subsection{Proof of Theorem~\ref{theo-diluteCPTP} --- One-shot exact SD-cost under ${\rm{CPTP}}_A$ maps}

\label{sec:proof-OneshotdiluteCPTP}

\begin{proof}
 We  prove that $\xi_c(\rho_{XA})= \xi_{\max}(\rho_0,\rho_1)$ and start  with  the  achievability part,  i.e., the  upper  bound $\xi_c(\rho_{XA})\le \xi_{\max}(\rho_0,\rho_1)$.   Without  loss  of  generality, suppose that $\xi_{\max}(\rho_0,\rho_1)$ is finite because, otherwise, the upper bound is trivially satisfied.
  
  Let us first consider the case $Q_{\max}(\rho_0,\rho_1)= 1$ (and hence $\xi_{\max} = 0$) for which necessarily $\rho_0=\rho_1\equiv\rho$.  Let us choose the measure-and-prepare channel $$\cE(\cdot) = \Tr(\cdot)\rho.$$ Note that
  \begin{align*}
     \cE(\pi_M) = \cE(\sigma^{(1)} \pi_M \sigma^{(1)}) = \rho \quad \forall M \ge 1.
  \end{align*}
  Hence, $$ ({\rm{id}} \otimes \cE)(\gamma_{XQ}^{(1,p)}) = \rho_{XA},$$
  which in turn implies that
   $\xi_c(\rho_{XA})\le 0 = \xi_{\max}(\rho_0,\rho_1)$. As $\xi_c(\rho_{XA})\ge 0$ by definition, this also gives $\xi_c(\rho_{XA})= \xi_{\max}(\rho_0,\rho_1)$. %\\
 
 Next consider the case in which $\rho_{XA}$ is a state for which $M=Q_{\max}(\rho_0,\rho_1)>1$ and hence $\xi_{\max}(\rho_0,\rho_1) = \log M>0$. We first prove the achievability bound $\xi_c(\rho_{XA}) \le \xi_{\max} (\rho_0, \rho_1)$. To do this, let us define
\begin{align}
&\widetilde\rho_0 = \frac{1}{2M-2}\Big((2M-1)\rho_0 -\rho_1\Big), \\
&\widetilde\rho_1 = \frac{1}{2M-2}\Big((2M-1)\rho_1 -\rho_0\Big).
\end{align} 
By assumption on $M$, we have that $\widetilde \rho_0$ and $\widetilde\rho_1$ are quantum states.
Consider now the measure-and-prepare channel $\cE$ given by
\begin{align}
\cE(\rho )= \bra{0}\rho\ket{0}\,\widetilde\rho_0 + \bra{1}\rho\ket{1}\,\widetilde\rho_1. 
\end{align}
Then 
\begin{align}
\cE(\pi_M) &= \left(1-\frac{1}{2M}\right)\widetilde\rho_0 + \frac{1}{2M} \widetilde\rho_1\nn\\
&=\frac{2M-1}{4M^2 -4M}\Big((2M-1)\rho_0 - \rho_1\Big) +\frac{1}{4M^2 -4M}\Big((2M-1)\rho_1 - \rho_0\Big) \nn\\
&= \frac{1}{4M^2-4M}\big((2M-1)^2-1\big)\rho_0 = \rho_0. 
\end{align}
By symmetry, we also get
\begin{align}
\cE(\sigma^{(1)}\pi_M\sigma^{(1)}) = \rho_1.
\end{align}
This implies that $({\rm{id}} \otimes \cE) (\gamma_{XQ}^{(M,p)}) = \rho_{XA}$, which in turn implies that
\begin{align}
\xi_c(\rho_0,\rho_1) \le \log M  = \xi_{\max}(\rho_0,\rho_1).
\end{align}
%\smallskip

To show the reverse inequality, we first note that for all $M\ge 1 $,
\begin{align}
\label{eq:ximaxUnit}
\xi_{\max}(\pi_M,\sigma^{(1)}\pi_M\sigma^{(1)}) = \log M .
\end{align}
To see this, note that 
\begin{align}
\pi_M \le (2\widetilde M-1) \sigma^{(1)}\pi_M\sigma^{(1)}
\end{align}
if and only if $M\le\widetilde M$. The reversed constraint in \eqref{eq:Qmax} is then also satisfied, which establishes that $$
Q_{\max}(\pi_M,\sigma^{(1)}\pi_M\sigma^{(1)}) = M,
$$ and hence \eqref{eq:ximaxUnit}. Let $M\ge1$ satisfy the constraint in the definition of $\xi_c(\rho_{XA})$ in \eqref{eq:DiluteMiluteCPTPA}. By definition there exists a CPTP map $\cE$ such that 
\begin{align}
\left(\rm{id}\otimes\cE\right)\left(\gamma_{XQ}^{(M, p)}\right) = \rho_{XA}.
\end{align}
By the data-processing inequality for $\xi_{\max}$, we get
\begin{align}
\xi_{\max}(\rho_0,\rho_1) \le \xi_{\max}(\pi_M,\,\sigma^{(1)}\pi_M\sigma^{(1)}) = \log M .
\end{align}
And hence, as $M$ can be chosen arbitrarily under the constraint in \eqref{eq:DiluteMilute}, we see that
\begin{align}
\xi_{\max}(\rho_0,\rho_1)\le \xi_c(\rho_{XA}),
\end{align}
and therefore in total
\begin{align}
\xi_{c}(\rho_{XA}) =  \xi_{\max}(\rho_0,\rho_1),
\end{align}
concluding the proof.
\end{proof}

\subsection{Proof of Theorem~\ref{theo-diluteCDS} --- One-shot exact SD-cost under CDS maps}

\label{sec:proofOneshotdiluteCDS}

\begin{proof}
We first prove the achievability part, i.e., $\xi^\star_c(\rho_{XA})\le\xi^\star_{\max}(\rho_{XA})$.
Without loss of generality, suppose that $\xi^\star_{\max}(\rho_{XA})$ is finite because, otherwise, the upper bound is trivally satisfied.
Moreover, note that 
\begin{align}
Q^\star_{\max}(\rho_{XA}) &=\frac{1}{2}\Big( \max\left\{2^{D_{\max}(\rho_0\|\rho_1)+\log\left(p/(1-p)\right)},2^{D_{\max}(\rho_1\|\rho_0)+\log\left((1-p)/p\right)}\right\} + 1\Big) \nn\\
&\ge \frac{1}{2}\Big(2^{\max\{\log\left(p/(1-p\right),\log\left((1-p)/p\right)\}} + 1\Big) \nn= \frac{1}{2}\Big(\max\left\{\frac{p}{1-p},\frac{1-p}{p}\right\}+1\Big) \\&=\frac{1}{2}\max\{1/p,1/(1-p)\}.
\end{align}
We first treat the case $Q^\star_{\max}(\rho_{XA})=\frac{1}{2}\max\{1/p,1/(1-p)\}$ for which necessarily $\rho_0=\rho_1\equiv\rho$. Assume without loss of generality $p\le 1-p$ (otherwise just flip the classical system $X$ of $\rho_{XA}$). Let now 
\begin{align}
  \cE_0(\cdot) = \bra{1}\cdot\ket{1}\rho,\quad\quad  \cE_1(\cdot) = \bra{0}\cdot\ket{0}\rho,
\end{align}
which are quantum operations, i.e., completely positive and trace non-increasing maps that sum to a trace-preserving map. Define the corresponding CDS map
\begin{align}
\cN = {\rm id}_{X}\otimes\cE_0 + \cF_X\otimes\cE_1.
\end{align}
Writing now $M= Q^\star_{\max}(\rho_{XA})= \frac{1}{2p}$, we see
\begin{align}
\nn\cN(\gamma^{(M)}_{XQ}) & =
\kb{0}\otimes\frac{1}{2}\left(\cE_0(\pi_{M})+ \cE_1(\sigma^{(1)} \pi_M\sigma^{(1)})\right) +\kb{1}\otimes\frac{1}{2}\left(\cE_1(\pi_M)+ \cE_0(\sigma^{(1)} \pi_M\sigma^{(1)})\right) \\ &=p\kb{0}\otimes\rho + (1-p)\kb{1}\otimes\rho = \rho_{XA}.
\end{align}
Hence, this shows  $\xi^\star_c(\rho_{XA})\le \log M  = \xi^\star_{\max}(\rho_{XA}).$

Let now $Q^\star_{\max}(\rho_{XA})>\frac{1}{2}\max\{1/p,1/(1-p)\}\ge1$ and define for $M = Q^\star_{\max}(\rho_{XA})$ the operators
\begin{align}
&\widetilde\rho_0 = \frac{1}{2Mp-1}\Big(p(2M-1)\rho_0 -(1-p)\rho_1\Big), \\
&\widetilde\rho_1 = \frac{1}{2M(1-p)-1}\Big((1-p)(2M-1)\rho_1 -p\rho_0\Big).
\end{align}
By assumption on $M$ we have that $\widetilde \rho_0$ and $\widetilde\rho_1$ are states.
Moreover, let
\begin{align}
q = \frac{2Mp -1}{2M-2} ,
\end{align}
and hence 
\begin{align}
1-q = \frac{2M(1-p) -1}{2M-2}.
\end{align}
Consider now the quantum operations
\begin{align}
&\cE_0(\rho )= q\bra{0}\rho\ket{0}\,\widetilde\rho_0 + (1-q)\bra{1}\rho\ket{1}\,\widetilde\rho_1, \\
&\cE_1(\rho )= (1-q)\bra{0}\rho\ket{0}\,\widetilde\rho_1 + q\bra{1}\rho\ket{1}\,\widetilde\rho_0.
\end{align}
Note that $\cE_0 + \cE_1$ is trace preserving.
The corresponding CDS map is given by 
\begin{align}
\cN = {\rm id}_X\otimes\cE_0 + \cF_X\otimes\cE_1.
\end{align}
Then 
\begin{align}
\cN(\gamma_{XQ}^{(M)}) = \kb{0}\otimes\frac{1}{2}\Big(\cE_0(\pi_M) +\cE_1(\sigma^{(1)}\pi_M\sigma^{(1)})\Big) + \kb{1}\otimes\frac{1}{2}\Big(\cE_1(\pi_M) +\cE_0(\sigma^{(1)}\pi_M\sigma^{(1)})\Big).
\end{align}
Moreover, we see
\begin{align}
\frac{1}{2}\Big(\cE_0(\pi_M) +\cE_1(\sigma^{(1)}\pi_M\sigma^{(1)})\Big) & = q\left(1-\frac{1}{2M}\right)\widetilde\rho_0 + \frac{1-q}{2M} \widetilde\rho_1\nn\\
&=\frac{2M-1}{4M^2 -4M}\Big(p(2M-1)\rho_0 - (1-p)\rho_1\Big)\nn\\
& \qquad +\frac{1}{4M^2 -4M}\Big((1-p)(2M-1)\rho_1 - p\rho_0\Big) \nn\\
&= \frac{p}{4M^2-4M}\big((2M-1)^2-1\big)\rho_0 = p\rho_0. 
\end{align}
By symmetry we also get
\begin{align}
\frac{1}{2}\Big(\cE_1(\pi_M) +\cE_0(\sigma^{(1)}\pi_M\sigma^{(1)})\Big) = (1-p)\rho_1.
\end{align}
This proves the achievability
\begin{align}
\xi^\star_c(\rho_{XA}) \le \xi_{\max}^\star(\rho_{XA}).
\end{align}

For the other inequality, we first note that \begin{align}
\xi^\star_{\max}(\gamma^{(M)}_{XQ}) = \log M,
\end{align}
which follows from the observation $\xi^\star_{\max}(\gamma_{XQ}^{(M)}) = \xi_{\max}(\pi_M,\sigma^{(1)}\pi_M\sigma^{(1)})$ and the arguments in the case of free operations being ${\rm{CPTP}}_A$ maps (in particular consider the discussion of \eqref{eq:ximaxUnit}).
Let $M\ge1$ satisfy the constraint in the definition of $\xi^\star_c(\rho_{XA})$ in \eqref{eq:DiluteMiluteCDS}. By definition there exists a CDS map $\cN$ such that \begin{align}
\cN(\gamma_{XQ}^{(M)}) = \rho_{XA}.
\end{align}
We use monotonicity of $\xi^\star_{\max}$ under CDS maps (compare Lemma~\ref{lem:xistarmon}) to get that
\begin{align}
    \xi^\star_{\max}(\rho_{XA}) \le \xi^\star_{\max}(\gamma_{XQ}^{(M)}) = \log M .
\end{align}
As $M$ is arbitrary, we see that $\xi^\star_{\max}(\rho_{XA}) \le \xi^\star_c(\rho_{XA})$ and hence $\xi^\star_{\max}(\rho_{XA}) = \xi^\star_c(\rho_{XA}).$
\end{proof}

	\subsection{One-shot approximate SD-dilution}
	
	\label{sec:one-shot-approx-dilution}
	
	 We can now define the \emph{one-shot approximate SD-cost} for a general c-q state
	 \begin{align}
	  \rho_{XA} = p\kb{0}\otimes\rho_0 + (1-p)\kb{1}\otimes\rho_1.
	 \end{align}
  \begin{definition}
  For $\eps\ge0$ and golden unit $\gamma_{XQ}^{(M,q)}$, the \emph{one-shot approximate SD-cost} of the c-q state $\rho_{XA}$ is given by
  \begin{align}
  \label{eq:oneApproxDil}
\xi^{\FO,q,\eps}_{c}(\rho_{XA}) &\coloneqq \log\Big(\inf\left\{M\Big| d^\prime_{\operatorname{FO}}(\gamma_{XQ}^{(M,q)}\to\rho_{XA})\le \eps \right\}\Big)\nn \\
&=\log\Big(\inf\left\{M\Big|\,  D^\prime\!\left(\cA\left(\gamma_{XQ}^{(M,q)}\right), \rho_{XA}\right)\le \eps,\,\,\cA\in {\rm{FO}}\right\}\Big),
  \end{align}
  where the minimum conversion error $d^\prime_{\operatorname{FO}}$ is defined in Definition~\ref{def:conversiondistance}.
  	For the choice
  	\begin{align}
  	{\hbox{\FO}}\equiv \Big\{{\rm{id}} \otimes\cE\Big| \cE \text{ CPTP on system A}\Big\}\equiv {\rm{CPTP}}_A ,
  	\end{align}
  	the only sensible choice in \eqref{eq:oneApproxDistill} is $q=p$, as free operations of the form $\rm{id}\otimes\cE$ cannot change the prior in the c-q state. In that case we simply write
  	\begin{align}
  	  \xi^{\eps}_{c}(\rho_{XA}) \equiv \xi^{{\rm{CPTP}}_A,p,\eps}_{c}(\rho_{XA}).
  	\end{align}
  	Whereas for the choice ${\hbox{\FO}}\equiv {\hbox{\cds}}$ and $q=1/2$, we use the notation
  	\begin{align}
  	    \xi^{\star,\eps}_c(\rho_{XA})\equiv\xi^{{\rm{CDS}},1/2,\eps}_{c}(\rho_{XA}).
  	\end{align}
  \end{definition}
  
  In the case of the dilution task, the \emph{one-shot approximate SD-cost} can be directly obtained from the corresponding exact quantity.
  
  \begin{lemma}
  \label{lem:ApproxDilEquality}
  For $\FO$ being c-q state preserving and $\eps\ge0$, we have
\begin{align}
 \label{eq:ApproxDil}
\xi^{\FO,q,\eps}_{c}(\rho_{XA}) = \inf_{\widetilde\rho_{XA}\in B^\prime_\eps(\rho_{XA})} \xi^{\FO,q}_{c}(\widetilde\rho_{XA}),
\end{align}
where we have defined the ball of c-q states\footnote{Note that \eqref{eq:ApproxDil} is also true if we replace $B'_\eps(\rho_{XA})$ with the full $D'$-ball with radius $\eps$ of all linear operators and not just c-q states. The reason for that is that $\xi^{\FO,q}_{c}(\widetilde\rho_{XA})$ is infinite for $\widetilde\rho_{XA}$ not a c-q state because the set of free operations $\FO$ is assumed to be c-q state preserving.} around $\rho_{XA}$ with radius $\eps$ with respect to the scaled trace distance $D^\prime(\cdot, \cdot)$ 
\begin{align}
B^\prime_\eps(\rho_{XA}) = \Big\{\widetilde\rho_{XA}\text{ c-q state}\Big|  D^\prime(\widetilde\rho_{XA}, \rho_{XA})\le\eps\Big\}.
\end{align}
Hence, in particular we get for free operations being $\cptp_A$ or $\cds$ maps
\begin{align}
\xi^\eps_c(\rho_{XA}) &= \inf_{\widetilde\rho_{XA}\in B^\prime_\eps(\rho_{XA})}  \xi_c(\widetilde\rho_{XA}), \\
\xi^{\star,\eps}_c(\rho_{XA}) &= \inf_{\widetilde\rho_{XA}\in B^\prime_\eps(\rho_{XA})}  \xi^\star_c(\widetilde\rho_{XA}).
\end{align}
\end{lemma}
\begin{proof}
The proof simply follows by
\begin{align}
\label{eq:ProofApproximSDdilute}
\xi^{\FO,q,\eps}_{c}(\rho_{XA}) & = \log\Big(\inf\left\{M\Big| d^\prime_{\operatorname{FO}}( \gamma_{XQ}^{(M,q)}\to\rho_{XA})\le \eps \right\}\Big)\nn \\
&=\log\Big(\inf\left\{M\Big|\,  D^\prime\!\left(\cA\left(\gamma_{XQ}^{(M,q)}\right), \rho_{XA}\right)\le \eps,\,\,\cA\in {\rm{FO}}\right\}\Big)\nn\\
&= \nn\log\Big(\inf\left\{M\Big|\, \cA\left(\gamma_{XQ}^{(M,q)}\right) =\widetilde\rho_{XA},\,\,\widetilde\rho_{XA}\in B^\prime_\eps(\rho_{XA}),\, \cA\in \FO\right\}\Big) \\
&= \inf_{\widetilde\rho_{XA}\in B^\prime_\eps(\rho_{XA})}\log\Big(\inf\left\{M\Big| \cA\left(\gamma_{XQ}^{(M,q)}\right) =\widetilde\rho_{XA},\,\,\cA\in \FO\right\}\Big)\nn \\
&= \inf_{\widetilde\rho_{XA}\in B^\prime_\eps(\rho_{XA})}\xi^{\FO,q}_c(\widetilde\rho_{XA}).
\end{align}
Here, for the third equality we have used that  $\widetilde\rho_{XA}=\cA\left(\gamma_{XQ}^{(M,q)}\right)$ is a c-q state because $\FO$ is c-q state preserving. Hence $D^\prime\!\left(\widetilde\rho_{XA}, \rho_{XA}\right)\le\eps$ already implies $\widetilde\rho_{XA}\in B^\prime_\eps(\rho_{XA})$.
\end{proof}

\subsection{Optimal asymptotic rates of exact and approximate SD-dilution}

\label{sec:asymptotic-dilution}

Consider the c-q state
$$
\rho_{XA}^{(n)}\coloneqq  p \kb{0} \otimes \rho_0^{\otimes n} + (1-p) \kb{1} \otimes \rho_1^{\otimes n}.
$$
Similarly, as in the case of distillation, we are interested in the asymptotic quantities
\begin{align}
\label{eq:Ascos}
\limsup_{n\to\infty}\frac{\xi_{c}(\rho_{XA}^{(n)})}{n}\quad ; \quad \limsup_{n\to\infty}\frac{\xi_{c}^\star(\rho_{XA}^{(n)})}{n}.
\end{align}
Using Theorem~\ref{theo-diluteCPTP} and Theorem~\ref{theo-diluteCDS} and the additivity of $D_{\max}$ and hence also of the Thompson metric $d_T$, we can directly read off that both limits in \eqref{eq:Ascos} exist and are given by the Thompson metric:
\begin{theorem}[Exact asymptotic SD-cost]\label{theo-exact-asympDil} For all $p\in(0,1)$, the optimal asymptotic rates of exact SD-dilution under $\cptp_A$ and $\cds$ maps is given by
\begin{align}
	\lim_{n\to\infty}\frac{\xi_{c}(\rho_{XA}^{(n)})}{n}= \lim_{n\to\infty}\frac{\xi_{c}^\star(\rho_{XA}^{(n)})}{n} = d_T(\rho_0,\rho_1).
	\label{eq:thompson-metric-op-int}
\end{align}
\end{theorem}

Hence, unlike the case of distillation, the optimal asymptotic rates in the case of exact dilution  do not match the quantum Chernoff bound, as they are too large. We also note here that \eqref{eq:thompson-metric-op-int} gives an operational interpretation of the Thompson metric in quantum information theory.

However, allowing errors with respect to the scaled trace distance $D^\prime$ in the conversion, the corresponding approximate quantities converge to the Chernoff divergence.
\begin{theorem}[Approximate asymptotic SD-cost]\label{theo-asympDil}
For all $\eps>0$ and $p\in(0,1)$, the optimal asymptotic rates of approximate SD-dilution under $\cptp_A$ and $\cds$ maps is given by
\begin{align}
	\label{eq_AsympResultDil}
	\lim_{n\to\infty}\frac{\xi^{\eps}_{c}(\rho_{XA}^{(n)})}{n}
	 =\lim_{n\to\infty}\frac{\xi^{\star,\eps}_{c}(\rho_{XA}^{(n)})}{n} = \xi(\rho_0,\rho_1),
	\end{align}
	where $\xi(\rho_0,\rho_1) = -\log\!\left(\min_{0\le s\le1}\Tr\!\left(\rho_0^s\rho_1^{1-s}\right)\right)$ denotes the quantum Chernoff divergence. In particular this gives
	\begin{align}
	\lim_{\eps\to 0}\lim_{n\to\infty}\frac{\xi^{\eps}_{c}(\rho_{XA}^{(n)})}{n}
	 =\lim_{\eps\to 0}\lim_{n\to\infty}\frac{\xi^{\star,\eps}_{c}(\rho_{XA}^{(n)})}{n} = \xi(\rho_0,\rho_1).
	\end{align}
\end{theorem}

Here, the restriction to $p\in(0,1)$ is sensible, as for $p\in\{0,1\}$, we directly get $\xi^{\eps}_{c}(\rho_{XA}^{(n)})= 0$ and $\xi^{\star,\eps}_{c}(\rho_{XA}^{(n)})=\infty$ for all $n\in\N$ and $\eps\ge 0$.

\begin{remark}
The fact that the limits in \eqref{eq_AsympResultDil} hold without any restriction on the value of $\varepsilon > 0 $ implies that the strong converse holds for the asymptotic SD-cost. Another way of interpreting this statement is as follows: for a sequence of SD dilution protocols with rate below the asymptotic SD-cost, the error necessarily converges to infinity as $n\to\infty$.
\end{remark}

\begin{remark}
\label{rem:asymp-rev-1}
Given that the asymptotic distillable-SD and SD-cost are equal to the quantum Chernoff divergence, it follows that the resource theory of symmetric distinguishability is asymptotically reversible. This means that, in the asymptotic limit of large $n$, one can convert the  source state $\rho^{(n)}_{XA}$ to a target state $\sigma^{(m)}_{XA'}$ at a rate $\frac{m}{n}$ given by the ratio of the Chernoff divergences. Then one can go back to $\rho^{(n)}_{XA}$ at the inverse rate with no loss (in the asymptotic limit). The procedure to do so, for the first aforementioned state conversion, is to distill golden-unit states from $\rho^{(n)}_{XA}$ at a rate equal to the Chernoff divergence. Then we dilute these golden-unit states to $\sigma^{(m)}_{XA'}$, and the overall conversion rate is equal to the ratio of Chernoff divergences. Then we go back from $\sigma^{(m)}_{XA'}$ to $\rho^{(n)}_{XA}$ in a similar manner, and there is no loss in the asymptotic limit. We discuss these points in much more detail in Section~\ref{sec:asymp}.

We note here that other resource theories, such as pure-state bipartite entanglement \cite{BBPS96}, coherence \cite{Winter_2016}, thermodynamics \cite{BHORS13}, and asymmetric distinguishability \cite{Wang2019states} are all asymptotically reversible in a similar sense.
\end{remark}

\subsection{Proof of Theorem~\ref{theo-asympDil}}

\label{sec:proofs-asymptotic-dilution}
 
 \subsubsection{Lower bound for the asymptotic SD-cost in Eq.~\eqref{eq_AsympResultDil}}
 
Using an argument similar to that given in  the proof of Theorem~\ref{theo-asympAprDistil}, we establish  the following asymptotic lower bound on the approximate SD-cost:
\begin{lemma}
	\label{lem:LowerBound}
	For all $\eps \ge0$, 
	\begin{align}
	\label{eq:UpperBoundConv}
	  \liminf_{n\to\infty}\frac{\xi^\eps_{c}(\rho_{XA}^{(n)})}{n}&\ge\xi(\rho_0,\rho_1), \\
	  \liminf_{n\to\infty}\frac{\xi^{\star,\eps}_{c}(\rho_{XA}^{(n)})}{n}&\ge\xi(\rho_0,\rho_1).
\end{align}
\end{lemma}
\begin{proof}
We only consider the lower bound for $\xi^{\star,\eps}_c$, as the one for $\xi^\eps_c$ exactly follows the same line of reasoning.
Let $M$ be such that it satisfies the constraint in the following optimization:
\begin{align}
\label{eq:xistarapproxinproof}
\xi^{\star,\eps}_c(\rho_{XA}^{(n)}) = \log\!\left(\inf\left\{M\Big|\,d^\prime_{\operatorname{CDS}}(\gamma_{XQ}^{(M)}\to\rho^{(n)}_{XA})\le\eps\right\}\right).
\end{align}
Hence, there exists a {\rm{CDS}} map  $\cN$ satisfying
$$
D^\prime\!\left(\cN(\gamma_{XQ}^{(M)}),\rho_{XA}^{(n)}\right)\le\eps.$$ Then, by the monotonicity of the minimum error probability under CDS maps and Lemma~\ref{lem:PerrDivBound}, we see
\begin{align}
 \frac{1}{2M} = p_{\operatorname{err}}\!\left(\gamma_{XQ}^{(M)}\right) \le p_{\operatorname{err}}\!\left(\cN(\gamma_{XQ}^{(M)})\right) \le (\eps+1)p_{\operatorname{err}}(\rho^{(n)}_{XA}).
\end{align}
As $M$ is arbitrary under the constraint in \eqref{eq:xistarapproxinproof}, we get
\begin{align}
\xi^{\star,\eps}_{c}(\rho_{XA}^{(n)}) \ge -\log(p_{\operatorname{err}}(\rho_{XA}^{(n)})) -\log(2(\eps+1)),
\end{align}
and hence
\begin{align}
\liminf_{n\to\infty}\frac{\xi^{\star,\eps}_{c}(\rho_{XA}^{(n)})}{n} \ge \lim_{n\to\infty} \frac{-\log(p_{\operatorname{err}}(\rho_{XA}^{(n)}))-\log(2(\eps+1))}{n} = \xi(\rho_0,\rho_1),
\end{align}
concluding the proof.
\end{proof}

\subsubsection{Upper bound for the asymptotic SD-cost in Eq.~\eqref{eq_AsympResultDil} and smoothed Thompson metric}

  Denote the set of sub-normalised states on a Hilbert space $\cH$ by
  \begin{align} 
  \cS_{\le}(\cH) \coloneqq \left\{\omega\in\cB(\cH)\,\Big|\,\,\omega \ge 0,\,\,\Tr(\omega)\le 1\right\},
  \end{align}
  with $\cB(\cH)$ the set of bounded operators on $\cH$.
  Moreover, for a sub-normalised state $\omega$ on $\cH$, we define the trace ball of sub-normalised states with radius $\eps\ge0$ around $\omega$ as 
  \begin{align}
  \cB_\eps(\omega) \coloneqq \left\{\widetilde\omega \in\cS_{\le}(\cH)\, \Big|\,\,\frac{1}{2}\|\,\omega-\widetilde\omega\|_1\le\eps\right\}.
  \end{align}
 
In order to prove the desired upper bound on the asymptotic dilution cost \eqref{eq_AsympResultDil}, we  consider the following smoothed quantity:
\begin{align}
d^\eps_T(\omega_0,\omega_1) & \coloneqq \inf_{\substack{\widetilde\omega_0 \in \cB_\eps(\omega_0),\\\widetilde\omega_1 \in \cB_\eps(\omega_1)}}d_T(\widetilde\omega_0,\widetilde\omega_1)\\&= \inf_{\substack{\widetilde\omega_0 \in \cB_\eps(\omega_0),\\\widetilde\omega_1 \in \cB_\eps(\omega_1)}}\max\Big\{D_{\max}(\widetilde\omega_0\|\widetilde\omega_1),D_{\max}(\widetilde\omega_1\|\widetilde\omega_0)\Big\}.
\end{align}
We get the following upper bound on the smoothed Thompson metric:

\begin{lemma}
	\label{lem:ThompsonMetricSymBound}
	For all $\omega_0,\omega_1\in\cS_{\le}(\cH)$ and $\eps\in(0,1]$, we have the bound 
	\begin{align}
	\label{eq:ThompsonBound}
	d^\eps_T(\omega_0,\omega_1)\le \log\!\left(\frac{4}{\eps} \right).
	\end{align}
	For the special case of $\Tr(\omega_0) + \Tr(\omega_1)=1$, we get the slightly stronger bound
	\begin{align}
	\label{eq:ThompsonBoundsum1}
	\inf_{\substack{\widetilde\omega_0 \in \cB_\eps(\omega_0),\,\widetilde\omega_1 \in \cB_\eps(\omega_1),\\\Tr(\widetilde\omega_0)+\Tr(\widetilde\omega_1)=1}}d_T(\widetilde\omega_0,\widetilde\omega_1)\le \log\!\left(\frac{4}{\eps}\right) ,
	\end{align}
	and for the case $\Tr(\omega_0)=1$ and $\Tr(\omega_1)=1$, we get
	\begin{align}
	\label{eq:ThompsonBound1}
	\inf_{\substack{\widetilde\omega_0 \in \cB_\eps(\omega_0),\,\widetilde\omega_1 \in \cB_\eps(\omega_1),\\\Tr(\widetilde\omega_0)=1,\,\Tr(\widetilde\omega_1)=1}}d_T(\widetilde\omega_0,\widetilde\omega_1)\le \log\!\left(\frac{2}{\eps}\right).
	\end{align}
\end{lemma}
\begin{proof}
	For $\eps\in(0,1]$, fix
	\begin{align}
	\lambda_0 &= \max\left\{\frac{2\Tr((\omega_1-\omega_0)_+)}{\eps},1\right\},\nn\\
	\lambda_1 &= \max\left\{\frac{2\Tr((\omega_0-\omega_1)_+)}{\eps},1\right\},
	\end{align}
	and write
	\begin{align}
	    0\le \eps_0 & \coloneqq  \Tr((\omega_1 -\omega_0)_+)/\lambda_0\le\eps/2 , \\
	     0\le  \eps_1 & \coloneqq  \Tr((\omega_0 -\omega_1)_+)/\lambda_1\le\eps/2 .
	\end{align}
	Define now the positive semi-definite operators
	\begin{align}
	\widetilde\omega_0 & \coloneqq \frac{\omega_0 + \frac{1}{\lambda_0}(\omega_1-\omega_0)_+}{1+\eps_0+\eps_1}, \\
	\label{eq:tildewomega1def}
	\widetilde\omega_1 & \coloneqq \frac{\omega_1 + \frac{1}{\lambda_1}(\omega_0-\omega_1)_+}{1+\eps_0+\eps_1},
	\end{align}
	where $A_+$ denotes the positive part of a self-adjoint operator $A.$
	Firstly, note that by definition 
	\begin{align}
	\Tr(\widetilde\omega_0) =\frac{\Tr(\omega_0) + \eps_0}{1+\eps_0+\eps_1}\le 1 ,
	\end{align}
	and analogously $\Tr(\widetilde\omega_1)\le1$, which gives $\widetilde\omega_0,\widetilde\omega_1\in\cS_{\le}(\cH)$.
	Moreover, note that $\widetilde\omega_0 \in\cB_{\eps}(\omega_0)$ and  $\widetilde\omega_1 \in\cB_{\eps}(\omega_1)$, which follows by
	\begin{align}
	\nonumber\frac{1}{2}\left\|\omega_0 - \widetilde\omega_0\right\|_1 &= \frac{1}{2(1+\eps_0+\eps_1)}\left\|(\eps_0+\eps_1)\omega_0-\frac{1}{\lambda_0}(\omega_1-\omega_0)_+\right\|_1 \\&\le\nn \frac{1}{2(1+\eps_0+\eps_1)}\left((\eps_0+\eps_1)\|\omega_0\|_1 + \frac{\|(\omega_1-\omega_0)_+\|_1}{\lambda_0}\right)\\& \le \frac{2\eps_0+\eps_1}{2(1+\eps_0+\eps_1)} \le \eps ,
	\end{align}
	and analogously for $\widetilde\omega_1$.
	
	We now note that
	\begin{align}
	\label{eq:omega0bound1}
	\lambda_0\widetilde\omega_0 \nn&=\frac{1}{1+\eps_0+\eps_1}\Big(\lambda_0\omega_0 + (\omega_1-\omega_0)_+\Big) \ge \frac{1}{1+\eps_0+\eps_1}\Big((\lambda_0-1)\omega_0 + \omega_1\Big) \\&\ge \frac{\omega_1}{1+\eps_0+\eps_1},
	\end{align}
	where we have used $\lambda_0\ge1$. Furthermore, we see
	\begin{align}
	\label{eq:omega0bound2}
	&\frac{\lambda_0}{\lambda_1}\widetilde\omega_0 - \frac{1}{\lambda_1(1+\eps_0+\eps_1)}(\omega_0-\omega_1)_+ \nn\\&\nn= \frac{1}{1+\eps_0+\eps_1}\left(\frac{\lambda_0}{\lambda_1}\omega_0 +\frac{1}{\lambda_1}\Big((\omega_1-\omega_0)_+ - (\omega_0-\omega_1)_+\Big)\right) \\&\nn=\frac{1}{1+\eps_0+\eps_1}\left(\frac{\lambda_0}{\lambda_1}\omega_0 +\frac{1}{\lambda_1}\Big((\omega_1-\omega_0)_+ - (\omega_1-\omega_0)_-\Big)\right)\\&\nn=\frac{1}{1+\eps_0+\eps_1}\left(\frac{\lambda_0}{\lambda_1}\omega_0 + \frac{1}{\lambda_1}(\omega_1-\omega_0)\right) \\&= \frac{1}{1+\eps_0+\eps_1}\left(\frac{\lambda_0-1}{\lambda_1}\omega_0 + \frac{1}{\lambda_1}\omega_1\right) \ge 0.
	\end{align}
	Here, in the third line, we have denoted the negative part of an self-adjoint operator $A$ by $A_-$ and used $(-A)_+ = A_-$ and in the fourth line we used $A = A_+ -A_-$.
	Hence, combining \eqref{eq:omega0bound1} and \eqref{eq:omega0bound2} together with the definition of $\widetilde\omega_1$ \eqref{eq:tildewomega1def}, this in total gives
	\begin{align}
	\label{eq:Dmaxessence0}
	\left(\lambda_0 +\frac{\lambda_0}{\lambda_1}\right) \widetilde\omega_0 \ge\widetilde\omega_1.
	\end{align}
	Analogously we get
	\begin{align}
	\label{eq:Dmaxessence1}
	\left(\lambda_1+\frac{\lambda_1}{\lambda_0}\right)\widetilde\omega_1 \ge \widetilde\omega_0.
	\end{align} 
	Therefore, using $\Tr((\omega_0-\omega_1)_+) \le \Tr(\omega_0) \le 1$ and $\Tr((\omega_1-\omega_0)_+) \le \Tr(\omega_1) \le 1$ and hence $1\le\lambda_0,\lambda_1 \le \max\{2/\eps,1\}=2 /\eps$, \eqref{eq:Dmaxessence0} and	\eqref{eq:Dmaxessence1} give
	 \begin{align}
	\label{eq:Dmaxomega1}D_{\max}\!\left(\widetilde\omega_1\|\widetilde\omega_0\right) \le \log\!\left(\lambda_0 +\frac{\lambda_0}{\lambda_1}\right)\le\log\!\left(\frac{4}{\eps}\right)\end{align} 
	and \begin{align}\label{eq:Dmaxomega2}D_{\max}(\widetilde\omega_0\|\widetilde\omega_1) \le \log\!\left(\lambda_1+\frac{\lambda_1}{\lambda_0}\right)\le\log\!\left(\frac{4}{\eps}\right). \end{align}
	Therefore, we obtain for the Thompson metric of $\widetilde\omega_0$ and $\widetilde\omega_1$
	\begin{align}
	\nn d_T(\widetilde\omega_0,\widetilde\omega_1) 
	&=\max\Big\{D_{\max}(\widetilde\omega_0\|\widetilde\omega_1),D_{\max}(\widetilde\omega_1\|\widetilde\omega_0)\Big\} \\&\le \log\!\left(\frac{4}{\eps}\right).
	\end{align}
	Hence, by definition of the smoothed Thompson metric this shows \eqref{eq:ThompsonBound}. Moreover, noting that in case $\Tr(\omega_0)+\Tr(\omega_1) =1$, we have by construction $$\Tr(\widetilde\omega_0)+\Tr(\widetilde\omega_1)=\frac{\Tr(\omega_0)+\Tr(\omega_1)+\eps_0+\eps_1}{1+\eps_0+\eps_1}=1,$$ which immediately also gives \eqref{eq:ThompsonBoundsum1}.
	
	In order to also conclude \eqref{eq:ThompsonBound1}, we slightly change the above construction in the following way: First note that in the case $\Tr(\omega_0) = \Tr(\omega_1) = 1$ we have $\Tr((\omega_0-\omega_1)_+) = \Tr((\omega_1-\omega_0)_+)$. Let now 
	\begin{align}
	\lambda = \max\left\{\frac{\Tr((\omega_0-\omega_1)_+)}{\eps},1\right\},
	\end{align}
	and write $0\le \eps^\prime = \Tr((\omega_0-\omega_1))_+)/\lambda= \Tr((\omega_1-\omega_0))_+)/\lambda \le \eps$. Define now the positive semi-definite operators
	\begin{align}
	\widetilde\omega_0 = \frac{\omega_0 + \frac{1}{\lambda}(\omega_1-\omega_0)_+}{1+\eps^\prime},
	\end{align}
	and
	\begin{align}
		\widetilde\omega_1 = \frac{\omega_1 + \frac{1}{\lambda}(\omega_0-\omega_1)_+}{1+\eps^\prime}.
	\end{align}
	Note that by $\Tr(\omega_0) = \Tr(\omega_1) = 1$ we get $\Tr(\widetilde\omega_0) = \Tr(\widetilde\omega_1) = 1$. Moreover, similarly to the above, we see $\widetilde\omega_0 \in \cB_\eps(\omega_0)$ and $\widetilde\omega_1 \in \cB_\eps(\omega_1).$ Lastly, by using the same arguments as in \eqref{eq:omega0bound1} and \eqref{eq:omega0bound2} we see that 
	\begin{align}
	\left(\lambda+1\right)\widetilde\omega_0 \ge \widetilde\omega_1
	\end{align}
	and
	\begin{align}
	\left(\lambda+1\right)\widetilde\omega_1 \ge \widetilde\omega_0,
	\end{align}
	which gives by using $\lambda \le1/\eps$
	 \begin{align}
D_{\max}\!\left(\widetilde\omega_1\|\widetilde\omega_0\right) \le \log\!\left(\lambda+1\right)\le\log\!\left(\frac{1+\eps}{\eps}\right) \le \log\!\left(\frac{2}{\eps}\right)\end{align} 
	and \begin{align}D_{\max}(\widetilde\omega_0\|\widetilde\omega_1) \le \log\!\left(c+1\right)\le\log\!\left(\frac{1+\eps}{\eps}\right) \le \log\!\left(\frac{2}{\eps}\right).
	 \end{align}
	 Then we conclude \eqref{eq:ThompsonBound1}.
\end{proof}

\bigskip 
The desired upper bound \eqref{eq_AsympResultDil} in Theorem~\ref{theo-asympDil} can now be deduced from the following lemma.

\begin{lemma}
For all $\eps>0$, we have
	\begin{align}
	\label{eq:AEPupCPTPA}
	\limsup_{n\to\infty}\frac{\xi^{\eps}_c(\rho^{(n)}_{XA})}{n} &\le \xi(\rho_0,\rho_1), \\
	\limsup_{n\to\infty}\frac{\xi^{\star,\eps}_c(\rho^{(n)}_{XA})}{n} &\le \xi(\rho_0,\rho_1). \label{eq:AEPupCDS}
	\end{align}
\end{lemma}
\begin{proof}
We prove the statement only for $\xi^{\star,\eps}_c$, as the proof for $\xi^{\eps}_c$ follows exactly along the same lines.
Let $A_n$ denote the quantum system consisting of $n$ copies of the quantum system $A$.
Moreover, for a generic c-q state $\widetilde\rho_{XA_n}$, we use the notation
\begin{align}
\widetilde\rho_{XA_n} = \tilde p\kb{0}\otimes\tilde\rho_0 + (1-\tilde p) \kb{1}\otimes\tilde\rho_1.
\end{align}
By using Theorem~\ref{theo-diluteCDS} and Lemma~\ref{lem:ApproxDilEquality}, we see that
\begin{align}
\xi_c^{\eps,\star}(\rho_{XA}^{(n)}) &=\nn \inf_{\substack{\widetilde\rho_{XA_n}\,\text{c-q state},\\ D^\prime(\widetilde\rho_{XA_n}, \rho^{(n)}_{XA}) \le\eps}} \xi^\star_c(\widetilde\rho_{XA_n}) = \inf_{\substack{\widetilde\rho_{XA_n}\,\text{c-q state},\\ D^\prime(\widetilde\rho_{XA_n}, \rho^{(n)}_{XA}) \le\eps}} \xi^\star_{\max}(\widetilde\rho_{XA_n}) \\
&=   \inf_{\substack{\widetilde\rho_{XA_n}\,\text{c-q state},\\ D^\prime(\widetilde\rho_{XA_n}, \rho^{(n)}_{XA}) \le\eps}} \log\!\left(\,2^{d_T(\tilde p\tilde\rho_0,(1-\tilde p)\tilde\rho_1)}+1\,\right)-1.
\end{align}

Define now 
\begin{align}
\widetilde\eps(\eps,n) \coloneqq \frac{\eps p_{\operatorname{err}}(\rho_{XA}^{(n)})}{2}. 
\end{align}
Note that if a given $\widetilde\rho_{XA_n}\equiv(\tilde p,\tilde\rho_0,\tilde\rho_1)$ is such that $\tilde p\tilde\rho_0\in \cB_{\widetilde\eps(\eps,n)}(p\rho^{\otimes n}_0)$ and \\ $(1-\tilde p)\tilde\rho_1\in \cB_{\widetilde\eps(\eps,n)}((1-p)\rho^{\otimes n}_1)$, this implies that $\frac{1}{2}\|\widetilde\rho_{XA_n}-\rho^{(n)}_{XA}\|_1\le 2 \widetilde\eps(\eps,n)$ and hence 
\begin{align}
D^\prime(\widetilde\rho_{XA_n}, \rho^{(n)}_{XA}) \le\eps.
\end{align}
Using that, by writing $\widetilde\omega_0 = \tilde p\tilde\rho_0$ and $\widetilde\omega_1 = (1-\tilde p)\tilde\rho_1$, we see that
\begin{align}
\xi_c^{\eps,\star}(\rho_{XA}^{(n)})+1 &= \inf_{\substack{\widetilde\rho_{XA_n}\,\text{c-q state},\nn\\ D^\prime(\widetilde\rho_{XA_n}, \rho^{(n)}_{XA}) \le\eps}} \log\!\left(\,2^{d_T(\tilde p\tilde\rho_0,(1-\tilde p)\tilde\rho_1)}+1\,\right) \\&\le \nn \inf_{\substack{\widetilde\omega_0 \in \cB_{\widetilde\eps(\eps,n)}(p\rho^{\otimes n}_0),\,\widetilde\omega_1 \in \cB_{\widetilde\eps(\eps,n)}((1-p)\rho^{\otimes n}_1)\\\Tr(\widetilde\omega_0)+\Tr(\widetilde\omega_1)=1}}\log\!\left(\,2^{d_T(\widetilde\omega_0,\widetilde\omega_1)}+1\,\right) \\&
=\log\!\left(\inf_{\substack{\widetilde\omega_0 \in \cB_{\widetilde\eps(\eps,n)}(p\rho^{\otimes n}_0),\,\widetilde\omega_1 \in \cB_{\widetilde\eps(\eps,n)}((1-p)\rho^{\otimes n}_1)\\\Tr(\widetilde\omega_0)+\Tr(\widetilde\omega_1)=1}}\,2^{d_T(\widetilde\omega_0,\widetilde\omega_1)}+1\,\right) \nn\\&\le\log\!\left(\frac{4}{\widetilde\eps(\eps,n)} +1\right)\nn\\ &=\log\!\left(\frac{1}{\widetilde\eps(\eps,n)}\right) + \log\!\left(4+\widetilde\eps(\eps,n)\right) \nn\\&\le \log\!\left(\frac{1}{\widetilde\eps(\eps,n)}\right) + \log\!\left(4+\eps\right),
\end{align}
where we have used equation \eqref{eq:ThompsonBoundsum1} in Lemma~\ref{lem:ThompsonMetricSymBound} for the fourth line and $\widetilde\eps(\eps,n)\le\eps$ for the last inequality.
Hence, we finally see by definition of $\widetilde\eps(\eps,n)$
\begin{align}
\label{eq:finally}
\frac{\xi_c^{\eps,\star}(\rho_{XA}^{(n)})}{n} &\le \nn\frac{\log\!\left(\frac{1}{\widetilde\eps(\eps,n)}\right)}{n} + \frac{\log\!\left(4+\eps\right)-1}{n}  \\
&=  \frac{-\log(p_{\operatorname{err}}(\rho_{XA}^{(n)}))}{n} -\frac{ \log\!\left(\eps/2\right)}{n} +\frac{\log\!\left(4+\eps\right)-1}{n}.
\end{align}
Noting once again that $$\lim_{n\to\infty}
\frac{-\log(p_{\operatorname{err}}(\rho_{XA}^{(n)}))}{n} = \xi(\rho_0,\rho_1)$$ and  the other terms on the right-hand side of \eqref{eq:finally} are $\mathcal{O}(1)$ in $n$ finishes the proof. 
\end{proof}

\section{Examples}

\label{sec:examples}

In this section, we detail a few examples of the RTSD\ to illustrate some of
the key theoretical concepts developed in the previous sections. We begin with
a first example. Let $\rho_{0}$ and $\rho_{1}$ be the following states:%
\begin{align}
\rho_{0}  & \coloneqq \mathcal{A}^{\gamma,N}(|0\rangle\!\langle0|),\label{eq:example-1-state-1}\\
\rho_{1}  & \coloneqq \mathcal{A}^{\gamma,N}(|1\rangle\!\langle1|),
\label{eq:example-1-state-2}
\end{align}
where $\mathcal{A}^{\gamma,N}$ is the generalized amplitude damping channel,
defined as%
\begin{equation}
\mathcal{A}^{\gamma,N}(\omega)\coloneqq \sum_{i=0}^{3}A_{i}\omega A_{i}^{\dag},
\end{equation}
with $\gamma,N\in\left[  0,1\right]  $ and%
\begin{align}
A_{0}  & \coloneqq \sqrt{1-N}\left(  |0\rangle\!\langle0|+\sqrt{1-\gamma}|1\rangle\!
\langle1|\right)  ,\\
A_{1}  & \coloneqq \sqrt{\gamma\left(  1-N\right)  }|0\rangle\!\langle1|,\\
A_{2}  & \coloneqq \sqrt{N}\left(  \sqrt{1-\gamma}|0\rangle\!\langle0|+|1\rangle\!
\langle1|\right)  ,\\
A_{3}  & \coloneqq \sqrt{\gamma N}|1\rangle\!\langle0|.
\end{align}
The parameter $\gamma\in[0,1]$ is a damping parameter and $N \in [0,1]$ is a thermal noise parameter. The generalized amplitude
damping channel models the dynamics of a two-level system in contact with a
thermal bath at non-zero temperature \cite{NC10} and can be used as a
phenomenological model for relaxation noise in superconducting qubits \cite{CB08}.
See \cite{KSW20} for an in-depth study of the information-theoretic properties of
this channel and for a discussion of how this channel can be interpreted as a
qubit thermal attenuator channel. We choose the prior probabilities for the
states $\rho_{0}$ and $\rho_{1}$ to be $q$ and $1-q$ respectively, with $q=1/3$, so that the c-q state
describing the elementary quantum source is%
\begin{equation}
\rho_{XA}\coloneqq q|0\rangle\!\langle0|\otimes\rho_{0}+\left(  1-q\right)
|1\rangle\!\langle1|\otimes\rho_{1}.
\end{equation}

\begin{figure}
[ptb]
\begin{center}
\includegraphics[
width=4in
]%
{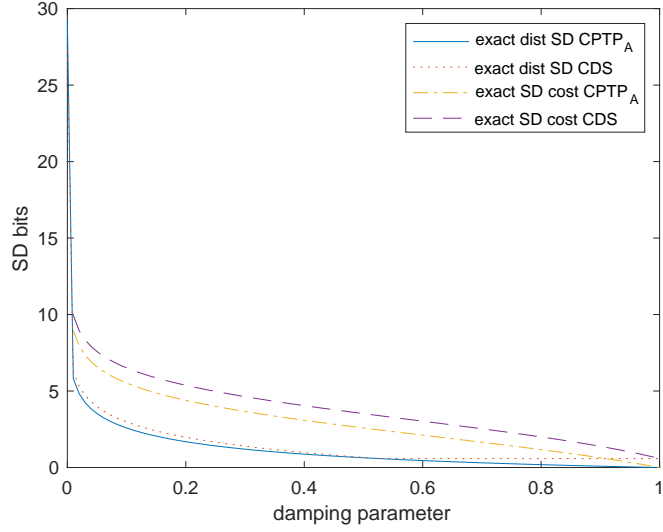}%
\caption{Various operational measures of symmetric distinguishability for the states in \eqref{eq:example-1-state-1}--\eqref{eq:example-1-state-2} with prior $q=1/3$, as a function of the damping parameter $\gamma$. The noise parameter $N=0.1$.}%
\label{fig:example-1}%
\end{center}
\end{figure}
In Figure~\ref{fig:example-1}, we set the thermal noise parameter $N=0.1$ and
plot the exact one-shot distillable-SD of $\rho_{XA}$ under CPTP$_{A}$ maps,
and under CDS maps, and the exact
one-shot SD-cost of $\rho_{XA}$ under CPTP$_{A}$ maps, and under CDS maps. As expected, when the damping parameter
$\gamma$ increases, each measure of SD decreases. The exact
distillable-SD\ and SD-cost under CDS maps do not decrease to zero due to
the non-uniform prior ($q=1/3$). However, they do decrease to zero
under CPTP$_{A}$ maps because
the prior $q$ does not play a role in this case. Additionally, the
SD-cost under CPTP$_{A}$ maps is strictly larger than the
distillable-SD\ under CPTP$_{A}$ maps for all $\gamma\in(0,1)$, demonstrating
that the RTSD\ is not reversible in this one-shot scenario. The same holds for
distillable-SD\ and SD-cost under CDS maps.

The SD-cost under CPTP$_{A}$ maps is smaller than that under CDS maps because
we use different golden units in these two cases. In this context,
recall Definition~\ref{def:exact-SD-cost}. It is not possible for CPTP$_{A}$ maps to change the
prior $q$. So we are forced to use the golden unit with prior $q$, i.e.,
$\gamma^{(M,q)}$, which in this case we chose to be $q=1/3$. CDS maps,
however, can change the prior, and as mentioned in Definition~\ref{def:exact-SD-cost}, we pick the
prior of the golden unit to be the canonical choice of $1/2$. Also, note that
$\gamma^{(M,1/3)}$ has more SD than $\gamma^{(M,1/2)}$; i.e., it dominates
$\gamma^{(M,1/2)}$ in the preorder of SD and hence can be transformed into the prior
$1/2$ golden unit via CDS (see Lemma~\ref{lem:GoldenUnitPreorder}). So the SD costs under CPTP$_{A}$ and CDS maps are
different, as we are paying with a less valuable currency in the case of CDS maps.

\begin{figure}
[ptb]
\begin{center}
\includegraphics[
width=4in
]%
{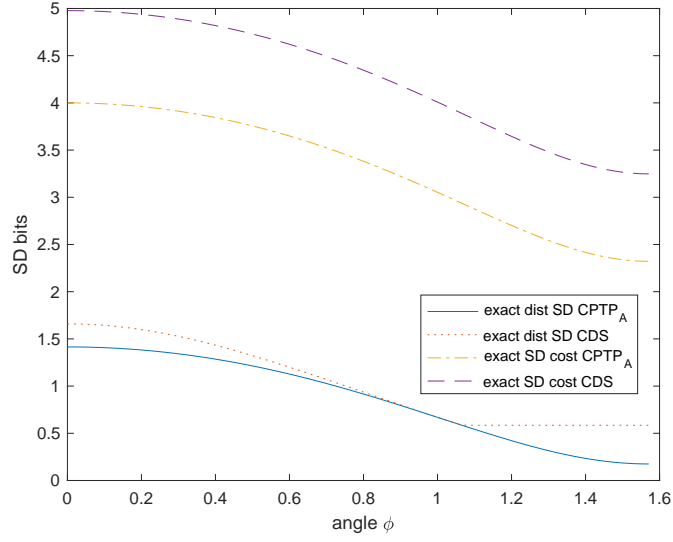}%
\caption{Various operational measures of symmetric distinguishability for the states in \eqref{eq:example-2-state-1}--\eqref{eq:example-2-state-2} with prior $q=1/3$, as a function of the angle parameter $\phi$. The damping parameter $\gamma=1/4$ and the noise parameter $N=0.1$.}%
\label{fig:example-2}%
\end{center}
\end{figure}

We next consider the following example:%
\begin{align}
\rho_{0}  & \coloneqq \mathcal{A}^{\gamma,N}(|0\rangle\!\langle0|), \label{eq:example-2-state-1}\\
\rho_{1}  & \coloneqq e^{i\phi\sigma^{(1)}}\mathcal{A}^{\gamma,N}(|1\rangle\!
\langle1|)e^{-i\phi\sigma^{(1)}},
\label{eq:example-2-state-2}
\end{align}
where the angle $\phi\in\left[  0,\pi/2\right]$, 
We choose the prior $q$ to be the same (i.e., $q=1/3$). All of the quantities
mentioned above are plotted in Figure~\ref{fig:example-2}. As the angle $\phi$
increases from zero to $\pi/2$, the states $\mathcal{A}^{\gamma,N}(|0\rangle\!\langle0|)$ and
$e^{i\phi\sigma^{(1)}}\mathcal{A}^{\gamma,N}(|1\rangle\!
\langle1|)e^{-i\phi\sigma^{(1)}}$ become less
distinguishable and become the same state when $\phi=\pi/2$. Thus, we expect
for the various measures of SD\ to decrease as $\phi$ increases from zero to
$\pi/2$. Similar statements as given above apply regarding the difference
between the SD\ quantities under CPTP$_{A}$ and CDS maps.%

\begin{figure}
[ptb]
\begin{center}
\includegraphics[
width=4in
]%
{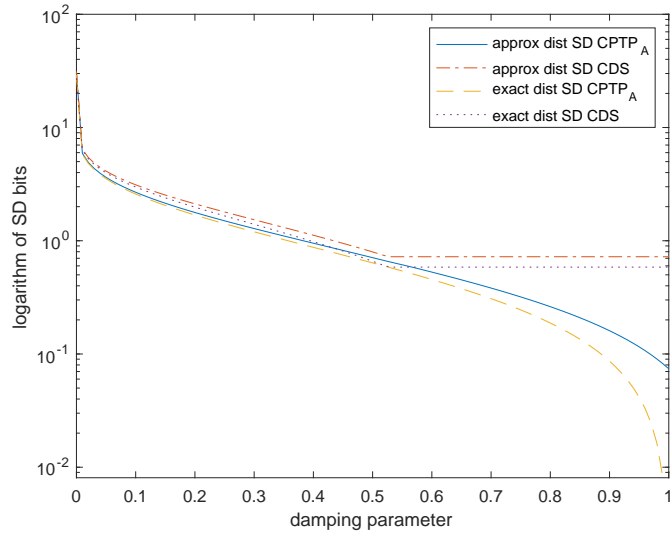}%
\caption{Logarithm of various operational measures of symmetric distinguishability for the states in \eqref{eq:example-1-state-1}--\eqref{eq:example-1-state-2} with prior $q=1/3$, as a function of the damping parameter $\gamma$. The noise parameter $N=0.1$, and the approximation error $\varepsilon = 0.1$. We have plotted the logarithm of the number of SD bits in order to distinguish the curves more clearly.}%
\label{fig:example-3}%
\end{center}
\end{figure}

As another example, we plot the logarithm of the one-shot approximate distillable-SD of the states in \eqref{eq:example-1-state-1}--\eqref{eq:example-1-state-2} under both CPTP$_A$ and CDS maps, as a function of the damping parameter~$\gamma$. We set the approximate error $\varepsilon = 0.1$, the prior probability $q=1/3$, and the noise parameter $N=0.1$. For reference, we also plot the logarithm of the exact distillable-SD under both CPTP$_A$ and CDS maps. See Figure~\ref{fig:example-3}. We have plotted the logarithm of the number of SD bits in order to distinguish the curves more clearly. The difference in the behavior of the curves has to do with the fact that CDS maps can change the prior probability while CPTP$_A$ maps cannot. Here, the distillable-SD under $\cds$ maps (both in the approximate and exact cases) flattens out for values of the damping parameter greater than $\gamma\approx 0.5$ as in this case symmetric distinguishability of the considered box is exclusively due to the non-uniform prior ($q=1/3$) and does not decrease further even if the quantum states themselves become less distinguishable.

\begin{figure}
[ptb]
\begin{center}
\includegraphics[
width=4in
]%
{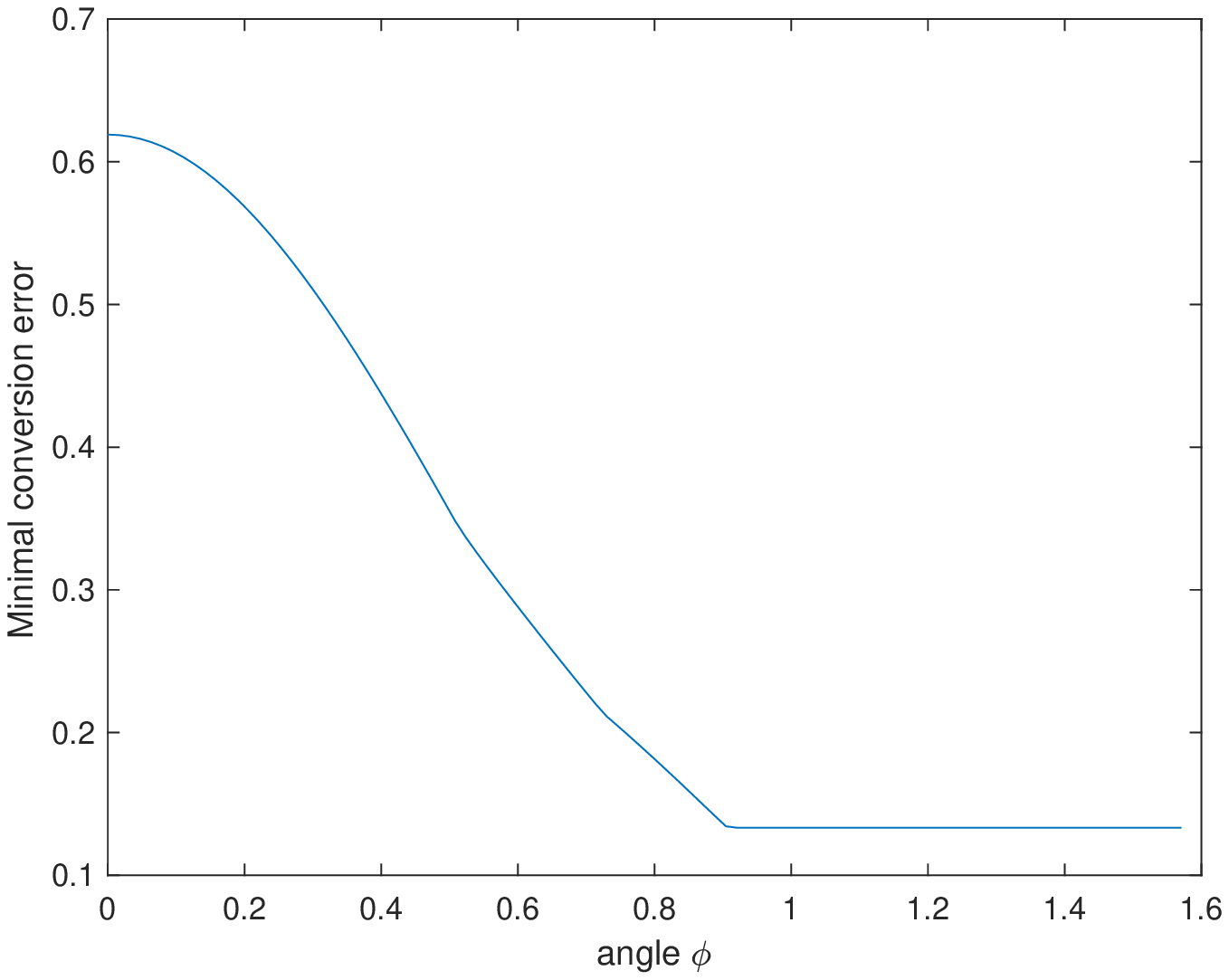}%
\caption{Minimum conversion error when transforming $(1/3, \mathcal{A}^{\gamma_1,N_1}(|0\rangle\!\langle0|),\mathcal{A}^{\gamma_1,N_1}(|1\rangle\!\langle1|))$ to $(1/4, \mathcal{A}^{\gamma_2,N_2}(|0\rangle\!\langle0|),e^{i\phi\sigma^{(1)}}\mathcal{A}^{\gamma_2,N_2}(|1\rangle\!\langle1|)e^{-i\phi\sigma^{(1)}})$ as a function of the angle $\phi \in [0,\pi/2]$, with $\gamma_1 = 0.5$, $N_1 = 0.3$, $\gamma_2 = 0.25$, and $N_2 = 0.1$.}%
\label{fig:example-4}%
\end{center}
\end{figure}

As a final example, we plot the minimum conversion error in \eqref{eq:opt-conv-error} when transforming the box $$(1/3, \mathcal{A}^{\gamma_1,N_1}(|0\rangle\!\langle0|),\mathcal{A}^{\gamma_1,N_1}(|1\rangle\!\langle1|))$$ to the box  $$(1/4, \mathcal{A}^{\gamma_2,N_2}(|0\rangle\!\langle0|),e^{i\phi\sigma^{(1)}}\mathcal{A}^{\gamma_2,N_2}(|1\rangle\!\langle1|)e^{-i\phi\sigma^{(1)}})$$ as a function of the angle $\phi \in [0,\pi/2]$, with $\gamma_1 = 0.5$, $N_1 = 0.3$, $\gamma_2 = 0.25$, and $N_2 = 0.1$. To do so, we make use of the semi-definite program from Proposition~\ref{prop:SDP-conv-err}. The minimum conversion error is plotted in Figure~\ref{fig:example-4} as a function of the angle $\phi$. Intuitively, for small values of the angle~$\phi$, it should be more difficult to perform the conversion because the states in the first box are less distinguishable than those in the second, and so we expect the error to be higher. However, as the angle $\phi$ increases, the states in the second box become less distinguishable and so the transformation becomes easier. The difference in the prior probabilities of the boxes is a fundamental limitation that cannot be overcome, even as $\phi$ becomes closer to $\pi/2$, so that the minimum conversion error plateaus for angle values greater than $\approx 0.9$.

All Matlab programs that generate the above plots (along with the semi-definite programs) are available with the arXiv ancillary files of this paper.

\section{Asymptotic transformation task}

\label{sec:asymp}

Let $\rho_{XA}$ and $\sigma_{XB}$ be c-q states explicitly given by
\begin{align}
\rho_{XA} &= p\kb{0}\otimes\rho_0 + (1-p)\kb{1}\otimes\rho_1 \label{r1} , \\
\sigma_{XB} &= q\kb{0}\otimes\sigma_0 + (1-q)\kb{1}\otimes\sigma_1, \label{o1}
\end{align}
with $\rho_0,\rho_1$ states of a quantum system $A$, and $\sigma_0,\sigma_1$ states of a quantum system~$B$. Moreover, we assume that $p,q\in[0,1]$. We use the short-hand notation $\vec{p}\coloneqq  (p, 1-p)$ and $\vec{q} \coloneqq  (q, 1-q)$ for the prior (distribution) of $\rho_{XA}$ and $\sigma_{XB}$ respectively and write $\vec{p} \succ \vec{q}$ if $\vec{p}$ majorizes~$\vec{q}$. Let $\rho_{XA}^{(n)}$ and $\sigma_{XB}^{(m)}$ be as follows:
\begin{align}
\rho_{XA}^{(n)} &= p\kb{0}\otimes\rho_0^{\otimes n} + (1-p)\kb{1}\otimes\rho_1^{\otimes n}  , \\
\sigma_{XB}^{(m)} &= q\kb{0}\otimes\sigma_0^{\otimes m} + (1-q)\kb{1}\otimes\sigma_1^{\otimes m}. 
\end{align}

\begin{definition}
	Let $\rho_{XA},\sigma_{XB}$ be c-q states, and let $\FO$ denote the set of free operations.
	For $n,m\in\N$ and $\eps>0$, we say that there exists a $(n,m,\eps)$ {\em{$\FO$-transformation protocol}} for the states $\rho_{XA}$ and $\sigma_{XB}$ if
	\begin{align}
	d^\prime_{\operatorname{FO}}(\rho^{(n)}_{XA}\mapsto\sigma^{(m)}_{XB})\le\eps.
	\end{align}
	 That is, there exists an $\cA\in \FO$ such that
	\begin{align}
	D^\prime\!\left(\cA(\rho_{XA}^{(n)}),\sigma^{(m)}_{XB}\right) \le \eps.
	\end{align}
    We denote such a transformation protocol in short by the notation $\rho_{XA} \mapsto \sigma_{XB}$.
\end{definition}

\begin{definition}
	\label{def:AchRate}
	The rate $R\ge 0$ is an achievable rate for the transformation $\rho_{XA} \mapsto \sigma_{XB}$ under free operations $\FO$, if for all $\eps,\delta>0$ and $n\in\N$ large enough there exists an $(n,\lfloor n(R-\delta)\rfloor,\eps)$ $\FO$-transformation protocol. The optimal rate is given by the supremum over all achievable rates, and is denoted by
	\begin{align}
	R_{\operatorname{FO}}(\rho_{XA} \mapsto \sigma_{XB}) = \sup\left\{R\ge 0\Big| R\text{ achievable rate under $\FO$}\right\}.
	\end{align}
	In particular, in the case of free operations being ${\rm{CPTP}}_A$ we write
	\begin{align}
	\RA(\rho_{XA} \mapsto \sigma_{XB}) \equiv \RA_{\operatorname{CPTP}_A}(\rho_{XA} \mapsto \sigma_{XB}),
	\end{align}
	and in the case of free operations being ${\rm{CDS}}$
	we write
	\begin{align}
	\RC (\rho_{XA} \mapsto \sigma_{XB}) \equiv \RA_{\operatorname{CDS}}(\rho_{XA} \mapsto \sigma_{XB}).
	\end{align}
\end{definition}
Note that by the inclusion ${\rm{CPTP}}_A \subset {\rm{CDS}}$ we immediately get the inequality
\begin{align}
\RA(\rho_{XA} \mapsto \sigma_{XB}) \le \RC(\rho_{XA} \mapsto \sigma_{XB}).
\end{align}
\begin{definition}[Strong converse rate]
	\label{def:StrongConvRate}
	The rate $R\ge0$ is a strong converse rate for the transformation $\rho_{XA} \mapsto \sigma_{XB}$ under free operations $\FO$, if for all $\eps,\delta>0$ and $n\in\N$ large enough there does not exist an $(n, \lceil n(R+\delta)\rceil,\eps)$ $\FO$-transformation protocol. The optimal strong converse rate is given by the infimum over all strong converse rates, and is denoted by
	\begin{align}
	\widetilde R_{\operatorname{FO}}(\rho_{XA} \mapsto \sigma_{XB})=\inf\left\{R\ge 0\,\Big|\, R\text{ strong converse rate under $FO$}\right\}. 
	\end{align}	
	In particular, in the case of free operations being ${\rm{CPTP}}_A$ we write 
	\begin{align}
	\widetilde \RA(\rho_{XA} \mapsto \sigma_{XB})\equiv\widetilde R_{\operatorname{CPTP}_A}(\rho_{XA} \mapsto \sigma_{XB}),
	\end{align}
	and in the case of free operations being ${\rm{CDS}}$
	we write
	\begin{align}
	\widetilde \RC(\rho_{XA} \mapsto \sigma_{XB})\equiv\widetilde R_{\operatorname{CDS}}(\rho_{XA} \mapsto \sigma_{XB}).
	\end{align}
\end{definition}

By definition, we have
\begin{align}
 R_{\operatorname{FO}}(\rho_{XA} \mapsto \sigma_{XB}) \le\widetilde R_{\operatorname{FO}}(\rho_{XA} \mapsto \sigma_{XB}) ,
\end{align}
and, moreover, by again using the fact that ${\rm{CPTP}}_A\subset {\rm{CDS}}$, we get the inequality
\begin{align}
\widetilde \RA (\rho_{XA} \mapsto \sigma_{XB}) \le \widetilde \RC (\rho_{XA} \mapsto \sigma_{XB}).
\end{align}
The following theorem gives expressions for the optimal achievable and strong converse rates for the transformation $\rho_{XA} \mapsto \sigma_{XB}$ under both ${\rm{CDS}}$ and ${\rm{CPTP}}_A$. Note that in the following
we interpret $\frac{\infty}{\infty}$ as $\infty$.
\begin{theorem}
	\label{thm:transformationrate}
	Let $\rho_{XA}$ and $\sigma_{XB}$ be the c-q states defined through \eqref{r1} and \eqref{o1} with $p,q\in(0,1)$. 
	\smallskip
	
	\noindent
	For free operations being $\cds$ we have: for $ \xi(\sigma_0,\sigma_1)>0$ 
	\begin{align}
	\RC (\rho_{XA} \mapsto \sigma_{XB}) = \widetilde \RC(\rho_{XA} \mapsto \sigma_{XB}) = \frac{\xi(\rho_0,\rho_1)}{\xi(\sigma_0,\sigma_1)}.
	\end{align}
	For $\xi(\sigma_0,\sigma_1)=0$ and $\xi(\rho_0,\rho_1)>0$ we have 
	\begin{align}
	  \RC(\rho_{XA} \mapsto \sigma_{XB}) = \widetilde \RC(\rho_{XA} \mapsto \sigma_{XB}) = \infty.
	\end{align}%\RC(\rho_{XA} \mapsto \sigma_{XB}) = \widetilde \RC(\rho_{XA} \mapsto \sigma_{XB}) = \infty$ in the following cases:
	For $ \xi(\sigma_0,\sigma_1)=\xi(\rho_0,\rho_1)=0 $, we have 
	\begin{align}
	    \RC(\rho_{XA} \mapsto \sigma_{XB}) &= \widetilde \RC(\rho_{XA} \mapsto \sigma_{XB}) = \infty, \quad\quad\quad\,\,\,\,\, {\hbox{ if}} \,\, \vec{p} \succ \vec{q},\,\, \\
	    %the prior of}} \,\,\rho_{XA} \,\, {\hbox{is less uniform than the prior of}} \,\, \sigma_{XB};\nonumber\\
	    {\hbox{whereas}} \quad \,\, \RC(\rho_{XA} \mapsto \sigma_{XB}) &=0 \,\,\, {\hbox{and}} \,\,\, \widetilde \RC(\rho_{XA} \mapsto \sigma_{XB}) =\infty, \quad{\hbox{else.}}
	\end{align}
%	\begin{enumerate}
	    %\item {\hbox{For $\xi(\sigma_0,\sigma_1)=0$ and $\xi(\rho_0,\rho_1)>0$ 
	    %\item \xi(\sigma_0,\sigma_1)=\xi(\rho_0,\rho_1)=0 \quad \hbox{and the prior of $\rho_{XA}$ is less uniform than the prior of $\sigma_{XB}} 
	%\end{enumerate}
	%for $\xi(\sigma_0,\sigma_1)=0$ and $\xi(\rho_0,\rho_1)>0$ we get $\RC(\rho_{XA} \mapsto \sigma_{XB}) = \widetilde \RC(\rho_{XA} \mapsto \sigma_{XB}) = \infty$
	%and in the case $\xi(\sigma_0,\sigma_1)=\xi(\rho_0,\rho_1)=0$ we get $\RC(\rho_{XA} \mapsto \sigma_{XB}) = \widetilde \RC(\rho_{XA} \mapsto \sigma_{XB}) = \infty$  if the prior of $\rho_{XA}$ is more singular than the prior of $\sigma_{XB}$ and $\RC(\rho_{XA} \mapsto \sigma_{XB}) =0$ and $\widetilde \RC(\rho_{XA} \mapsto \sigma_{XB}) =\infty$ otherwise.
	
	\noindent For free operations being $\cptp_A$ we have: in the case of $\rho_{XA}$ and $\sigma_{XB}$ having equal priors
	\begin{align}
	\label{eq:cooliCPTPA}
	\RA(\rho_{XA} \mapsto \sigma_{XB}) = \widetilde \RA(\rho_{XA} \mapsto \sigma_{XB}) = \frac{\xi(\rho_0,\rho_1)}{\xi(\sigma_0,\sigma_1)}.
	\end{align}
	Here, we interpreted $\frac{0}{0}$ as $\infty$.

	In the case of the priors being different we get 
	\begin{align}
	\label{eq:CPTPAspecial1}\RA(\rho_{XA} \mapsto \sigma_{XB}) &= \widetilde \RA(\rho_{XA} \mapsto \sigma_{XB}) = 0, \quad\quad\quad\text{if } \xi(\sigma_0,\sigma_1)>0, \\ \label{eq:CPTPAspecial2}	\RA(\rho_{XA} \mapsto \sigma_{XB}) &=0,\,\,\, \widetilde \RA(\rho_{XA} \mapsto \sigma_{XB}) = \infty,\quad \text{if } \xi(\sigma_0,\sigma_1)=0.
	\end{align}
\end{theorem}
\begin{remark}
For simplicity we excluded in Theorem~\ref{thm:transformationrate} the case of singular priors, i.e. $p\in\{0,1\}$ or $q\in\{0,1\}$. For completeness we now state the corresponding results on optimal and strong converse rates in these cases:

For free operations being CDS we have
\begin{align}
	\RC (\rho_{XA} \mapsto \sigma_{XB}) = \widetilde \RC(\rho_{XA} \mapsto \sigma_{XB}) = \infty,\quad\text{if }p\in\{0,1\} ,
\end{align}
and for $p\in(0,1)$ and $q\in\{0,1\}$ we have 
	\begin{align}
	    \RC(\rho_{XA} \mapsto \sigma_{XB}) &= \widetilde \RC(\rho_{XA} \mapsto \sigma_{XB}) = \infty, \quad {\hbox{if}} \,\, \xi(\rho_0,\rho_1)=\infty,\,\, \\
	    %the prior of}} \,\,\rho_{XA} \,\, {\hbox{is less uniform than the prior of}} \,\, \sigma_{XB};\nonumber\\
	    {\hbox{whereas}} \quad \,\, \RC(\rho_{XA} \mapsto \sigma_{XB}) &= \widetilde \RC(\rho_{XA} \mapsto \sigma_{XB}) = 0, \quad\,\,\, {\hbox{if}} \,\, \xi(\rho_0,\rho_1)<\infty.
	\end{align}
For free operations being $\cptp_A$ and $q\in\{0,1\}$ we have
\begin{align}
	R (\rho_{XA} \mapsto \sigma_{XB}) &= \widetilde R(\rho_{XA} \mapsto \sigma_{XB}) = \infty,\quad\text{if }p=q,\\
	R (\rho_{XA} \mapsto \sigma_{XB}) &= \widetilde R(\rho_{XA} \mapsto \sigma_{XB}) = 0,\quad\,\,\,\text{if }p\neq q.
\end{align}
And for $p\in\{0,1\}$ and $q\in(0,1)$ we have similarly to \eqref{eq:CPTPAspecial1} and \eqref{eq:CPTPAspecial2}
\begin{align}
	R (\rho_{XA} \mapsto \sigma_{XB}) &= \widetilde R(\rho_{XA} \mapsto \sigma_{XB}) = 0,\quad\quad\,\,\,\,\,\,\text{if }\xi(\sigma_0,\sigma_1)>0,\\
	R (\rho_{XA} \mapsto \sigma_{XB}) &=0,\,\,\, \widetilde R(\rho_{XA} \mapsto \sigma_{XB}) = \infty,\quad\text{if }\xi(\sigma_0,\sigma_1)=0.
\end{align}
\end{remark}
\begin{remark}
Further to what was already stated in Remark~\ref{rem:asymp-rev-1}, Theorem~\ref{thm:transformationrate} expresses the fact that the resource theory of symmetric distinguishability is asymptotically reversible. Indeed, the optimal asymptotic rate at which one can convert $\rho_{XA}$ to $\sigma_{XB}$ is equal to the ratio of quantum Chernoff divergences. The rate at which one can convert back is thus equal to the reciprocal of the forward rate. Since the product of these two rates is equal to one, we conclude that the RTSD is asymptotically reversible.
\end{remark}

\subsection{Proof of Theorem~\ref{thm:transformationrate}}

\subsubsection{Achievability}

We start the proof of Theorem~\ref{thm:transformationrate} by proving the achievability part. In particular, we show the following lemma.
\begin{lemma}
	Let $\rho_{XA}$ and $\sigma_{XB}$ be the c-q states defined through \eqref{r1} and \eqref{o1} with $p,q\in(0,1)$.
	For $ \xi(\sigma_0,\sigma_1)>0$, we have
	\begin{align}
	\label{eq:AchCDS}
	\RC(\rho_{XA} \mapsto \sigma_{XB}) &\ge \frac{\xi(\rho_0,\rho_1)}{\xi(\sigma_0,\sigma_1)}.
	\end{align}
	    Moreover, for $\xi(\sigma_0,\sigma_1)=0$ and $\xi(\rho_0,\rho_1)>0$, we get $\RC(\rho_{XA} \mapsto \sigma_{XB}) = \infty$ and in the case $\xi(\sigma_0,\sigma_1)=\xi(\rho_0,\rho_1)=0$ we get $\RC(\rho_{XA} \mapsto \sigma_{XB}) = \infty$  if $\vec{p}\succ \vec{q}$, and $\RC(\rho_{XA} \mapsto \sigma_{XB}) =0$ otherwise.

	Moreover, in the case of $\rho_{XA}$ and $\sigma_{XB}$ having the same priors we have
	\begin{align}
	\label{eq:AchCPTPA}
	\RA(\rho_{XA} \mapsto \sigma_{XB}) &\ge \frac{\xi(\rho_0,\rho_1)}{\xi(\sigma_0,\sigma_1)}.
	\end{align}	
	Here, we interpreted  $\frac{\infty}{\infty}$ and $\frac{0}{0}$ as $\infty$. In the case of priors being different we have $\RA(\rho_{XA} \mapsto \sigma_{XB})=0.$
\end{lemma}

\begin{proof}
We prove the result for free operations being ${\rm{CDS}}$, since for free operations being ${\rm{CPTP}}_A$ the proof follows the same lines. Let us first consider the case $\xi(\sigma_0,\sigma_1)=0$, in which case necessarily $\sigma_0 = \sigma_1\equiv\sigma$. Moreover, first assume that $\xi(\rho_0,\rho_1)=0$ (which implies that $\rho_0 = \rho_1\equiv\rho$) and $\vec{p} \succ \vec{q}$. In that case, there exists a $\lambda\in[0,1]$ such that 
\begin{align*}
\lambda p + (1-\lambda)(1-p) = q ,
\end{align*}
and consequently 
\begin{align*}
\lambda (1-p) + (1-\lambda)p = 1-q.
\end{align*}
Therefore, considering for every $m\in\N$
\begin{align*}
\cE^{(m)}_0(\cdot) = \lambda\Tr(\cdot)\sigma^{\otimes m},\quad\quad  \cE^{(m)}_1(\cdot) = (1-\lambda)\Tr(\cdot)\sigma^{\otimes m},
\end{align*}
which are quantum operations summing to a CPTP map, and the corresponding CDS map 
\begin{align}
\label{eq:MoreSingularPriorconstruction}
\cN^{(m)} = {\rm{id}}\otimes\cE^{(m)}_0 + \cF\otimes\cE^{(m)}_1,
\end{align}
we get 
\begin{align*}
\cN^{(m)}(\rho_{XA}) &= \left(\lambda p+(1-\lambda)(1-p)\right)\kb{0}\otimes\sigma^{\otimes m}  +\left(\lambda (1-p)+(1-\lambda)p\right)\kb{1}\otimes\sigma^{\otimes m} \\&= \sigma_{XB}^{(m)}.
\end{align*}
Consequently, $R^{\star}(\rho_{XA}\mapsto\sigma_{XB}) = \infty$, since $\rho_{XA}$ can be transformed to $\sigma_{BX}^{(m)}$ without error via a CDS map for an arbitrary $m \in {\mathbb{N}}$.
\smallskip

Consider now the case $\xi(\rho_0,\rho_1)=\xi(\sigma_0,\sigma_1)=0$ and %$\rho_{XA}$ not having more singular prior
$\vec{p}\not\succ \vec{q}$. In this case, using Lemma~\ref{lem:PerrDivBound}, we see that for all $n,m\in\N$
\begin{align*}
\min\{p,(1-p)\}&=p_{\operatorname{err}}(\rho_{XA})=p_{\operatorname{err}}(\rho^{(n)}_{XA}) \le \left(d^\prime_{\operatorname{CDS}}(\rho^{(n)}_{XA}\mapsto\sigma^{(m)}_{XB})+1\right)p_{\operatorname{err}}(\sigma^{(m)}_{XB}) \\& =\left(d^\prime_{\operatorname{CDS}}(\rho^{(n)}_{XA}\mapsto\sigma^{(m)}_{XB})+1\right)p_{\operatorname{err}}(\sigma_{XB}) \nonumber\\
&= \left(d^\prime_{\operatorname{CDS}}(\rho^{(n)}_{XA}\mapsto\sigma^{(m)}_{XB})+1\right)\min\{q,1-q\} ,
	\end{align*}
and since $\min\{p,(1-p)\} = (c+1)\min\{q,(1-q)\}$ for some $c>0$, we get
\begin{align*}
d^\prime_{\operatorname{CDS}}(\rho^{(n)}_{XA}\mapsto\sigma^{(m)}_{XB}) \ge c >0
\end{align*}
for all $n,m\in\N$, and hence $R^\star(\rho_{XA}\mapsto\sigma_{XB}) = 0$.

 Now consider the case in which $\xi(\sigma_0,\sigma_1)=0$  but $\xi(\rho_0,\rho_1)>0$. Assume without loss of generality that $q\le 1-q$, and for all $n\in\N$, let $\{\Lambda^{(n)}_0,\Lambda^{(n)}_1\}$ denote the optimal POVM for discriminating the quantum states $\rho_0^{\otimes n}$ and $\rho_1^{\otimes n}$ of the c-q state $\rho^{(n)}_{XA}$. Moreover, for all $m\in\N$, let
\begin{align}
\cE^{(n,m)}_0(\cdot) = \Tr(\Lambda^{(n)}_1\cdot)\sigma^{\otimes m},\quad\quad  \cE^{(n,m)}_1(\cdot) = \Tr(\Lambda^{(n)}_0\cdot)\sigma^{\otimes m},
\end{align}
which are quantum operations summing to a CPTP map, and define the corresponding CDS map
\begin{align*}
\cN^{(n,m)} = {\rm{id}}\otimes\cE^{(n,m)}_0 + \cF \otimes\cE^{(n,m)}_1.
\end{align*}
This gives 
\begin{align*}
\cN^{(n,m)}(\rho_{XA}^{(n)}) = p_{\operatorname{err}}(\rho_{XA}^{(n)})\kb{0}\otimes\sigma^{\otimes m} +  (1-p_{\operatorname{err}}(\rho_{XA}^{(n)}))\kb{1}\otimes\sigma^{\otimes m}.
\end{align*}
Now choose $n$ large enough such that $p_{\operatorname{err}}(\rho_{XA}^{(n)})\le q$, which gives that the prior of $\cN^{(n,m)}(\rho_{XA}^{(n)})$ majorises the prior of $\sigma_{XB}^{(m)}$. Using the above, i.e., the construction around \eqref{eq:MoreSingularPriorconstruction}, we can find a CDS map transforming $\cN^{(n,m)}(\rho_{XA}^{(n)})$ to $\sigma_{XB}^{(m)}$ for all $m$ without error. This implies that $R^\star(\rho_{XA}\mapsto\sigma_{XB}) = \infty$.
\smallskip

 Next consider the case in which $\xi(\sigma_0,\sigma_1)>0$. We can assume without loss of generality that $\xi(\rho_0,\rho_1)> 0$, since otherwise \eqref{eq:AchCPTPA} is trivially satisfied. 
 The case $\xi(\rho_0,\rho_1)=\infty$ follows by Lemma~\ref{lem:InfToInf}, i.e., by the fact that we can transform any infinite-resource state to any other c-q state via CDS maps without error. Note here that $\xi(\rho_0,\rho_1)=\infty$ if and only if $\rho_0$ and $\rho_1$ have orthogonal supports.
Furthermore, the case in which $\xi(\rho_0,\rho_1)<\infty$ and $\xi(\sigma_0,\sigma_1) = \infty$ follows from the fact that any transformation from a finite resource to an infinite-resource has infinite error with respect to the scaled trace distance $D^\prime$.
 \smallskip
 
 Let us therefore now finally consider the case in which $0<\xi(\rho_0,\rho_1),\,\xi(\sigma_0,\sigma_1)<\infty$, and fix $\eps,\delta>0$.
	By Theorem~\ref{theo-asymp} we can find for all $\delta_1>0$, an $N_1\in\N$ such that for all $n\ge N_1$ there exists a CDS map $\cN_1$ such that
	\begin{align*}
	\cN_1(\rho^{(n)}_{XA}) = \gamma_{XQ}^{(M_n)},
	\end{align*}
	where $\gamma_{XQ}^{(M_n)}$ is the $M_n$-golden unit defined in \eqref{golden1/2}
	with 
	\begin{align}
	\label{eq:GoodDistill}
	\frac{\log M_n}{n} \ge \xi(\rho_0,\rho_1) -\delta_1.
	\end{align}
	Since $\xi(\rho_0,\rho_1)>0$, without loss of generality even $\xi(\rho_0,\rho_1)>\delta_1$ by picking $\delta_1$ small enough, we also see that the sequence $(M_n)_{n\ge N}$ goes to infinity for $n\to\infty$. For each $M_n$ consider $m_n\in\N$ to be the unique number such that
	\begin{align}
	\label{eq:M_nbounds}
	\xi^{\star,\eps}_c(\sigma_{XB}^{(m_n)}) \le \log M_n \le \xi^{\star,\eps}_c(\sigma_{XB}^{(m_n+1)}).
	\end{align}
	First, note that  we also get $m_n\xrightarrow[n\to\infty]{}\infty$ because $M_n\xrightarrow[n\to\infty]{}\infty$.
 	Also note that as each $\xi^{\star,\eps}_c(\sigma_{XB}^{(m_n)})$ is finite there exists a CDS map $\cN_2$ such that
	\begin{align*}
	D^\prime\!\left(\cN_2(\gamma_{XQ}^{(M_n)}),\sigma_{XB}^{(m_n)}\right)\le \eps
	\end{align*}
	and hence in total
	\begin{align*}
	D^\prime\!\left(\cN_2\circ\cN_1(\rho_{XA'}^{(n)}),\sigma_{XB}^{(m_n)}\right)\le \eps.
	\end{align*}
	Moreover, we have by \eqref{eq:GoodDistill} and \eqref{eq:M_nbounds}
	\begin{align}
	\frac{\xi^{\star,\eps}_c(\sigma^{(m_n+1)}_{XB})}{n} \ge \frac{\log M_n}{n} \ge \xi(\rho_0,\rho_1) -\delta_1.
	\end{align}
	By Theorem~\ref{theo-asympDil} we can pick for all $\delta_2 >0$, an $N_2\in\N$ such that for all $n\ge\max\{N_1,N_2\}$ we have $$\frac{\xi^{\star,\eps}_c(\sigma^{(m_n+1)}_{XB})}{m_n+1}\le \xi(\sigma_0,\sigma_1) + \delta_2$$
	and therefore
	\begin{align}
	\frac{m_n+1}{n}\left(\xi(\sigma_0,\sigma_1) + \delta_2\right)\ge  \frac{\log M_n}{n} \ge \xi(\rho_0,\rho_1) -\delta_1.
	\end{align}
	As $\delta_1,\delta_2>0$ are arbitrary, we can choose them small enough and moreover possibly $n$ even larger such that
	\begin{align*}
	\frac{\xi(\rho_0,\rho_1) -\delta_1}{\xi(\sigma_0,\sigma_1) + \delta_2} - \frac{1}{n}\ge 	\frac{\xi(\rho_0,\rho_1)}{\xi(\sigma_0,\sigma_1)} -\delta.
	\end{align*}
	In summary, we have proven that for all $\eps,\delta>0$ and $n$ large enough, there exists an $(n,m_n,\eps)$ ${\rm{CDS}}$-transformation protocol with
	\begin{align}
	\frac{m_n}{n} \ge \frac{\xi(\rho_0,\rho_1)}{\xi(\sigma_0,\sigma_1)} -\delta,
	\end{align}
	which shows that $R\coloneqq  \frac{\xi(\rho_0,\rho_1)}{\xi(\sigma_0,\sigma_1)}$ is an achievable rate and hence proves \eqref{eq:AchCDS}.
\end{proof}

\subsubsection{Strong converse}

We now prove the strong converse part of Theorem~\ref{thm:transformationrate}. The proof  follows a similar idea as the one for proving the strong converse in the resource theory of asymmetric distinguishability \cite{Wang2019states}. There, the key insight was the pseudo continuity bound for either the sandwiched R\'enyi relative entropy \cite[Lemma 1]{Wang2019states} or Petz-R\'enyi relative entropy \cite[Lemma 3]{Wang2019states}. However, here we instead use a pseudo continuity bound for the operational quantity $-\log(p_{\operatorname{err}}(\cdot))$, involving the scaled trace distance $D^\prime$, and holding for two arbitrary cq-states $\rho_{XA}$ and $\sigma_{XA}$:
\begin{align}
\label{eq:pseud}
-\log(p_{\operatorname{err}}(\rho_{XA})) + \log(p_{\operatorname{err}}(\sigma_{XA})) \ge -\log\!\left(1+D^\prime(\rho_{XA}, \sigma_{XA})\right).
\end{align}
Note that \eqref{eq:pseud} directly follows from Lemma~\ref{lem:PerrDivBound}. With that we prove the following lemma:
\begin{lemma}
		Let $\rho_{XA}$ and $\sigma_{XB}$ be the c-q states defined through \eqref{r1} and \eqref{o1} with $p,q\in(0,1)$. Then we have
	\begin{align}
	\label{eq:ConvCDS}
	\widetilde \RC(\rho_{XA} \mapsto \sigma_{XB}) &\le \frac{\xi(\rho_0,\rho_1)}{\xi(\sigma_0,\sigma_1)},
	\end{align}
	and
	\begin{align}
	\label{eq:ConvCPTPA}
	\widetilde \RA(\rho_{XA} \mapsto \sigma_{XB}) &\le \frac{\xi(\rho_0,\rho_1)}{\xi(\sigma_0,\sigma_1)}.
	\end{align}	
	Here, we interpreted  $\frac{\infty}{\infty}$ and $\frac{0}{0}$ as $\infty$.
	
	Furthermore, in the case in which $\xi(\sigma_0,\sigma_1)=0$ we get 
	\begin{align}
	  \widetilde \RA(\rho_{XA} \mapsto \sigma_{XB})=\widetilde \RC(\rho_{XA} \mapsto \sigma_{XB})=\infty.  
	\end{align}
\end{lemma}
\begin{proof}
As $\widetilde \RA(\rho_{XA} \mapsto \sigma_{XB})\le\widetilde \RC(\rho_{XA} \mapsto \sigma_{XB})$ we only need to show \eqref{eq:ConvCDS} as \eqref{eq:ConvCPTPA} directly follows. Moreover, consider $\xi(\sigma_0,\sigma_1)>0$ as otherwise \eqref{eq:ConvCDS} is trivially true.
Let $n,m\in\N$ and $\eps>0$ be such that there exists an $(n,m,\eps)$ CDS-transformation protocol. Hence, there exists a CDS map $\cN$ such that
\begin{align}
D^\prime\!\left(\cN(\rho_{XA}^{(n)}),\sigma_{XB}^{(m)}\right) \le \eps.
\end{align}
Using first monotonicity of the minimum error probability under ${\rm{CDS}}$ maps and then \eqref{eq:pseud} we get
\begin{align}
\nonumber
-\log\!\left(p_{\operatorname{err}}(\rho_{XA}^{(n)})\right) &\ge -\log\!\left(p_{\operatorname{err}}(\cN(\rho_{XA}^{(n)}))\right) \ge -\log\!\left(p_{\operatorname{err}}(\sigma^{(m)}_{XB})\right) - \log(1+\eps)  \\&\ge m\xi(\sigma_0,\sigma_1)  - \log(1+\eps),
\end{align} 
where we have used \cite[Theorem 1]{ACMBMAV07} for the last inequality.
This gives
\begin{align}
\label{eq:AlmostStrong}
\frac{-\log\!\left(p_{\operatorname{err}}(\rho_{XA}^{(n)})\right)}{\xi(\sigma_0,\sigma_1)}\ge m - \frac{\log(1+\eps)}{\xi(\sigma_0,\sigma_1)}.
\end{align}
Let now $\delta_1>0$ and consider $n\in\N$ large enough such that 
\begin{align}
\frac{-\log\!\left(p_{\operatorname{err}}(\rho_{XA}^{(n)})\right)}{n} \le \xi(\rho_0,\rho_1) +\delta_1,
\end{align}
which is possible by the results of \cite{nussbaum2009chernoff}.
Substituting this into \eqref{eq:AlmostStrong} we get
\begin{align}
\frac{\xi(\rho_0,\rho_1)}{\xi(\sigma_0,\sigma_1)} \ge \frac{m}{n} - \frac{\log(1+\eps)}{n\xi(\sigma_0,\sigma_1)} - \frac{\delta_1}{\xi(\sigma_0,\sigma_1)}.
\end{align}
Hence, by picking $\delta_1>0$ small enough, we see that for all $\eps,\delta>0$ and $n$ large enough a possible $(n,m,\eps)$ ${\rm{CDS}}$-transformation protocol satisfies
 \begin{align}
 \frac{m}{n}< \frac{\xi(\rho_0,\rho_1)}{\xi(\sigma_0,\sigma_1)} + \delta,
 \end{align}
 which shows that $R\coloneqq  \frac{\xi(\rho_0,\rho_1)}{\xi(\sigma_0,\sigma_1)}$ is a strong converse rate and finishes the proof.
 
 It remains to discuss the case $\xi(\sigma_0,\sigma_1)=0$, in which case $\sigma_0=\sigma_1\equiv\sigma$, and to show that then $\widetilde \RA(\rho_{XA} \mapsto \sigma_{XB})=\widetilde \RC(\rho_{XA} \mapsto \sigma_{XB})=\infty.$ Consider for all $n,m\in\N$ the CPTP map $\cE^{(m)} = \Tr(\cdot)\sigma^{\otimes m}$  and note that 
 \begin{align*}
 \left({\rm{id}}\otimes\cE^{(m)}\right)(\rho^{(n)}_{XA}) = p \kb{0}\otimes\sigma^{\otimes m}+ (1-p)\kb{1}\otimes\sigma^{\otimes m}.
 \end{align*}
Hence,
 \begin{align*}
 D^\prime\!\left(\left({\rm{id}}\otimes\cE^{(m)}\right)(\rho^{(n)}_{XA}),\sigma^{(m)}_{XB}\right) = \frac{\left\|\left({\rm{id}}\otimes\cE^{(m)}\right)(\rho^{(n)}_{XA})-\sigma^{(m)}_{XB}\right\|_1}{2p_{\operatorname{err}}(\sigma^{(m)}_{XB})} = \frac{|p-q|}{\min\{q,1-q\}} < \infty.
 \end{align*}
 Therefore, for $\eps\coloneqq \frac{|p-q|}{\min\{q,1-q\}}$ fixed and all $n,m\in\N$ we have found a $(n,m,\eps)$ transformation protocol (under both ${\rm{CPTP}}_A$ and ${\rm{CDS}}$) which implies that \begin{equation}
     \widetilde \RA(\rho_{XA} \mapsto \sigma_{XB})=\widetilde \RC(\rho_{XA} \mapsto \sigma_{XB})=\infty,
 \end{equation}
 concluding the proof.
\end{proof}

\section{Summary and open questions}
\label{sec:conclusion}

In summary, we have introduced the resource theory of symmetric distinguishability (RTSD) and have answered many of the fundamental questions associated with it. In particular, we have developed an axiomatic approach to the RTSD, which led to the conclusion that CDS maps are the natural choice for free operations, with CPTP$_A$ maps being a special case. We then introduced the golden units of the RTSD and argued why a particular scaled trace distance is a more appropriate figure of merit for approximate transformations, instead of the standard 
trace distance.
%how an error metric for approximate transformations, other than the trace distance, is needed. 
We finally defined and studied the tasks of dilution, distillation, and transformation, in the exact and approximate cases, both in the one-shot and asymptotic scenarios. We proved that the rate at which asymptotic transformations are possible is equal to the ratio of quantum Chernoff divergences of the elementary information sources, and we thus concluded that the RTSD is asymptotically reversible.

Going forward from here, it would be interesting to generalize the RTSD that we developed in this paper for elementary information sources to more general
information sources, i.e., to c-q states for which the classical alphabet has a size greater than two. We note here that many of the concepts considered in our paper, such as the basic axioms for the RTSD, CDS maps, and the scaled trace distance $D'(\cdot, \cdot)$ already apply to this more general setting. In light of the seminal result in \cite{Li16}, it is a tantalizing possibility that the optimal conversion rate between quantum information sources would be equal to a ratio of multiple-state Chernoff divergences, as a generalization of Theorem~\ref{thm:transformationrate}, but it remains open to determine if it is the case. It is also interesting to determine expressions for the one-shot distillable-SD and SD-cost, as generalizations of $\xi_{\min}$ and  $\xi_{\max}$. As an additional open direction, it is worth exploring whether there is an operational interpretation of the scaled trace distance $D'(\cdot, \cdot)$ that we introduced in Section~\ref{sec:conv-dist}. Finally, it is an open question to determine if the one-shot approximate SD-cost can be evaluated by a semi-definite program. We prove in Appendix~\ref{app:opts-approx-cost-not-SDPs} that the two variants of approximate SD-cost (based on CPTP$_A$ and CDS maps) can be evaluated by means of bilinear programs, so that the methods of \cite{HKT20} can be used to evaluate these quantities. However, it is not clear to us if these bilinear programs can be simplified further to semi-definite programs.

\bigskip

\textbf{Acknowledgements}---We thank Sumeet Khatri for insightful discussions and contributions at the early stage of this project.  MMW
 acknowledges support from the National Science
Foundation under Grant No.~1907615, as well as Stanford QFARM and AFOSR (FA9550-19-1-0369). GG acknowledges support from the Natural Sciences and Engineering Research Council of Canada (NSERC). RS acknowledges support from the Cambridge Commonwealth, European and International Trust.

\appendix

\section{General properties of the axiomatic framework of the RTSD}

\label{sec:AppendAxiom}
In this appendix, we discuss a number of consequences of Axioms I-V, and the preorder of SD, introduced in Section~\ref{sec:axioms}. We also define general resource measures to quantify symmetric distinguishability, and we provide several examples of such measures.

\begin{lemma}
The preorder of SD satisfies the following properties:
\begin{enumerate}
\item \textbf{Minimal Elements}. For all $\rho_{XA}\in\mathfrak{D}(XA)$ and $\pi_{X'}\otimes\omega_{A'}\in\md(X'A')$ we have
\be\label{210}
\pi_{X'}\otimes\omega_{A'}\prec\rho_{XA}\;.
\ee

\item \textbf{Maximal Elements}. For all $\rho_{XA}\in\mathfrak{D}(XA)$, $\{p_x\}_{x=1}^{|X|}$ a probability distribution, and $A'$ a quantum system with $|A'|\geq |X|$, we have
\be
\rho_{XA}\prec |1\rangle\!\langle 1|_X\sim\sum_{x=1}^{|X|}p_x|x\rangle\!\langle x|_X\otimes |x\rangle\!\langle x|_{A'}\;.
\ee
\item \textbf{Reduction to Majorisation}. 
For $\rho_{X},\sigma_{X}\in\cD(X)$ classical states of the same dimension, the preorder of SD is equivalent to the majorisation preorder. 
\end{enumerate}
\end{lemma}
\begin{proof}
The first property  follows immediately from Axiom~III and Axiom~V. For the second property, first observe that Axiom~III implies that $|1\rangle\!\langle 1|_X\sim|1\rangle\!\langle 1|_X\otimes\omega_A$ for some $\omega_{A}$. From the form of CDS maps we get that $|1\rangle\!\langle 1|_X\otimes\omega_A$ can be converted to any cq-state in $\md(XA')$. Hence, since $\mathfrak{F}(XA\to XA')=\cds(XA\to XA')$ (see Lemma~\ref{FO-CDS}) the assertion follows. For the third property, note that by Axioms I, III, and V,
\begin{align}
\label{eq:SDMajo}
\rho_{X}\prec\sigma_X
\end{align} 
with respect to the preorder of SD if and only if
\begin{align}
\label{eq:DoubleWhopper}
\rho_{X} = \tr_{X'}\circ\mP_{XX'} \left(\sigma_{X}\otimes\pi_{X'}\right),
\end{align}
for some classical system $X'$ and permutation channel $\mP_{XX'}$ on the joint classical system $XX'$.
The above channels are known as noisy operations~\cite{HHO03}.
First, observe that noisy operations are doubly stochastic, so that if such a permutation channel $\mP_{XX'}$ exists, then $\sigma_{X}$ majorizes $\rho_{X}$. 

Conversely,
suppose $\rho=\sum_{y=1}^{m}t_yU_y\sigma U_y$ where $\{t_y\}_{y=1}^m$ are the components of a probability distribution, $U_y$ are permutation matrices on system $X$, and for convenience of the exposition here we removed the subscript $X$. It is well known that such $\{t_y\}$ and $\{U_y\}$ exist iff $\rho\prec\sigma$. Suppose that $t_y=\frac{n_y}{n}$ are rational components where $n=\sum_{y=1}^mn_y$ is the common denominator and each $n_y\in\mathbb{N}$. Let $X'$ be a classical system of dimension $|X'|=n$. Define a permutation matrix $P_{XX'}$ by its action on the basis elements $$P_{XX'}\left(|x\ra_X\otimes|x'\ra_{X'}\right)=\left(U_{y_{x'}}|x\ra_X\right)\otimes|x'\ra_{X'}\quad\quad\forall\;x\in[|X|]\;,\;x'\in[n],$$
where $y_{x'}\in[m]$ is the index satisfying
\be
\sum_{z=1}^{y_{x'}-1}n_z \leq x'\leq\sum_{z=1}^{y_{x'}}n_z ,
\ee
and we used the convention that the left-hand side of the above inequality is zero for $y_{x'}=1$. With this definition we have
\begin{align*}
\tr_{X'}\circ\mP_{XX'} \left(\sigma_{X}\otimes\pi_{X'}\right)=\frac1m\sum_{x'=1}^nU_{y_{x'}}\sigma_X U_{y_{x'}}^\dag
=\frac1m\sum_{y=1}^mn_yU_y\sigma U_y^\dag=\sum_{y=1}^{m}t_yU_y\sigma_X U_y=\rho_X\;.
\end{align*}
Therefore, noisy operations can approximate any mixture of unitaries arbitrarily well. 
\end{proof}

Note that the lemma above indicates that $|1\rangle\!\langle 1|_X$ is the maximal resource in the fixed dimension of $X$. If, for example, $X'$ is another system with a higher dimension $|X'|>|X|$ then 
\be
\label{eq:XiaFactorNoooo}
|1\rangle\!\langle 1|_X\prec|1\rangle\!\langle 1|_{X'}\;.
\ee
That is, the `embedding' of $X$ into $X'$ by adding zero components to matrices/vectors is not allowed in this resource theory since it can increase the value of the resource. To get the intuition behind it, consider Xiao possessing either one of the two classical states $|1\rangle\!\langle 1|_{X\tilde{X}}$ or $|1\rangle\!\langle 1|_X\otimes\pi_{\tilde{X}}$. In the first case, Xiao has complete information of the state in her possession since she knows the values of \emph{both} the random variables ${X}$ and $\tilde{X}$. On the other hand, in the second case Xiao has no information about $\tilde{X}$, since it is in a uniform state. Therefore, the first state is more distinguishable than the second one and we get 
\be
\label{eq:Xiafactor}
|1\rangle\!\langle 1|_X\sim|1\rangle\!\langle 1|_X\otimes\pi_{\tilde X}\prec|1\rangle\!\langle 1|_{X\tilde X}.
\ee
Hence, for the case $X' = X\tilde X$, \eqref{eq:Xiafactor} reduces \eqref{eq:XiaFactorNoooo}.
More generally, the SD of a cq-state $\rho_{XA}$ represents the ability of Alice to distinguish the elements in Xiao's system. Therefore, the greater $|X|$ is, the more elements there are to distinguish, and consequently, the maximal resource has greater SD. This, in particular, applies to the minimum error probability $p_{\operatorname{err}}(\rho_{XA})$. That is, suppose for example that two states $\rho_{XA},\sigma_{X'A'}\in\mathfrak{D}_{cq}$, with $|X|<|X'|$ satisfy
\be
p_{\operatorname{err}}(\rho_{XA})=p_{\operatorname{err}}(\sigma_{X'A'})\;.
\ee
Then, we can expect that $\sigma_{X'A'}$ has more SD since Alice is able to distinguish among $|X'|>|X|$ elements with the same error as she would have if she held $\rho_{XA}$. This means in particular that if we consider inter-conversions among c-q states with different classical dimensions then the minimum error probability is not a good measure of SD. In the following subsection we show how the minimum error probability needs to be re-scaled with the classical dimension so that it becomes a proper measure of SD.

\subsection{Quantification of SD}

SD is quantified with functions that preserve the preorder of SD.

\begin{definition}
A function $f:\mathfrak{D}_{cq}\to\mathbb{R}$ is called a \emph{measure of SD} if:
\begin{enumerate}
\item For any $\rho_{XA},\sigma_{X'A'}\in\mathfrak{D}_{cq}$ we have
\be
\rho_{XA}\prec\sigma_{X'A'}\quad\Rightarrow\quad f\left(\rho_{XA}\right)\leq f\left(\sigma_{X'A'}\right)\;.
\ee
\item For the trivial state $1$ (i.e. $|X|=|A|=1$) we have
\be
f(1)=0\;.
\ee
\end{enumerate}
\end{definition}

Note that from Axiom~III and the second condition above, all measures of SD vanish on free states. Combining this with the property in~\eqref{210} we conclude that measures of SD are non-negative.
\begin{example}
Let $$\mathbb{D}\;:\;\bigcup_A\cD(A)\times\cD(A)\to\mathbb{R}\;:\;(\rho,\sigma)\mapsto\mathbb{D}(\rho\|\sigma)$$ be a relative entropy; i.e. it satisfies the DPI, additivity (under tensor products), and the normalization $\mathbb{D}(1\|1)=0$. Then, the function
\be
\min_{\omega\in\cD(A)}\mathbb{D}\left(\rho_{XA}\|\pi_X\otimes\omega_A\right)\eqdef\log|X|-\mathbb{H}(X|A)_\rho
\ee
is a measure of SD, since in any quantum resource theory, a function of the form $\min_{\omega\in\mathfrak{F}(A)}\mathbb{D}(\rho_{A}\|\omega_A)$ is a measure of a resource~\cite{CG18}.
\end{example}

\begin{example}[Normalized guessing probability]
A special example of the above family of measures of SD is obtained when setting $\mathbb{D}$ to be the max-relative entropy $D_{\max}$.
Specifically, in \cite{KRS08} it was shown that the guessing probability can be written as
\begin{align*}
     p_{\rm{guess}}(X|A)_\rho = 2^{\min_{\omega\in\cD(A)} D_{\max} (\rho_{XA}\| I_X \otimes \omega_A)}\;.
\end{align*}
Therefore, replacing $I_X$ with the maximally mixed state $\pi_X$ we get that the function 
\be\label{235}
\log\left(|X|p_{\rm{guess}}(X|A)_\rho\right)=\min_{\omega\in\cD(A)} D_{\max} (\rho_{XA}\| \pi_X \otimes \omega_A)
\ee
is a measure of SD. Note that the dimension of the classical system is included on the left-side so that the expression remains invariant under replacement of $\rho_{XA}$ with $\rho_{XA}\otimes\pi_{X'}$.
\end{example}}

In this paper we have focused on the RTSD for the particular case in which the dimension of the classical system $X$ is fixed to $|X|=2$. In this case, as mentioned in the main text, it suffices to
consider measures of SD that behave monotonically under CDS but not necessarily under conditional noisy operations (i.e., under free operations that change the dimension of $X$). The measure of SD that we have chosen in the paper is given by Definition~\ref{def:SD}. As shown in Theorem~\ref{theo-distilCDS} it has the particularly pleasing feature of having an operational meaning in the context of SD distillation.

\section{Proof of Eq.~\eqref{eq:err-prob-TD-to-inf-resource} --- Minimum trace distance to infinite-resource states under free operations}

\label{app:proof-of-err-ratio-interpretation}

Let $\left(  p\rho^{0},\left(  1-p\right)  \rho^{1}\right)  $ be a pair of
subnormalized states, with $p\in(0,1)$, and $\rho^{0}$ and $\rho^{1}$ states.
Then this pair is in one-to-one correspondence with the following
classical--quantum state:%
\begin{equation}
\rho_{XB} \coloneqq p|0\rangle\!\langle0|_{X}\otimes\rho_{B}^{0}+\left(  1-p\right)
|1\rangle\!\langle1|_{X}\otimes\rho_{B}^{1}. \label{eq:basic-cq-state-app}%
\end{equation}
Let $(q\sigma^{0},\left(  1-q\right)  \sigma^{1})$ be another pair of
subnormalized states, with $q\in(0,1)$, and $\sigma^{0}$ and $\sigma^{1}$
states. Then this pair is in one-to-one correspondence with the following
classical--quantum state:
\begin{equation}
\sigma_{XB^{\prime}} \coloneqq q|0\rangle\!\langle0|_{X}\otimes\sigma_{B^{\prime}}%
^{0}+\left(  1-q\right)  |1\rangle\!\langle1|_{X}\otimes\sigma_{B^{\prime}}^{1}.
\end{equation}
The trace-distance conversion error of $\rho_{XB}$ to $\sigma_{XB^{\prime}}$\ is
defined as follows:%
\begin{equation}
\min_{\mathcal{N}_{XB\rightarrow XB^{\prime}}\in\text{$\operatorname{CDS}$}%
}\frac{1}{2}\left\Vert \mathcal{N}_{XB\rightarrow XB^{\prime}}(\rho
_{XB})-\sigma_{XB^{\prime}}\right\Vert _{1}, \label{eq:conv-distance}%
\end{equation}
where CDS is the set of conditional doubly stochastic (CDS)\ maps.

We first show that the trace-distance conversion error can be computed by means of a
semi-definite program.

\begin{proposition}
\label{prop:CDS-conv-dist-SDPs}The trace-distance conversion error
\begin{equation}
\min_{\mathcal{N}_{XB\rightarrow XB^{\prime}}\in\operatorname{CDS}}\frac{1}%
{2}\left\Vert \mathcal{N}_{XB\rightarrow X^{\prime}B^{\prime}}(\rho
_{XB})-\sigma_{XB^{\prime}}\right\Vert _{1}%
\end{equation}
between the initial pair $(p\rho_{B}^{0},\left(  1-p\right)  \rho_{B}^{1})$
and the target pair $(q\sigma_{B^{\prime}}^{0},\left(  1-q\right)
\sigma_{B^{\prime}}^{1})$ can be calculated by means of the following
semi-definite program:%
\begin{equation}
\min_{Y_{B^{\prime}}^{0},Y_{B^{\prime}}^{1},\Gamma_{BB^{\prime}}%
^{\mathcal{N}^{0}},\Gamma_{BB^{\prime}}^{\mathcal{N}^{1}}\geq0}%
\operatorname{Tr}[Y_{B^{\prime}}^{0}+Y_{B^{\prime}}^{1}]
\end{equation}
subject to%
\begin{equation}
\operatorname{Tr}_{B^{\prime}}[\Gamma_{BB^{\prime}}^{\mathcal{N}^{0}}%
+\Gamma_{BB^{\prime}}^{\mathcal{N}^{1}}]=I_{B},
\end{equation}%
\begin{align}
Y_{B^{\prime}}^{0} &  \geq\operatorname{Tr}_{B}[(p\rho_{B}^{0})^{T}%
\Gamma_{BB^{\prime}}^{\mathcal{N}^{0}}+(\left(  1-p\right)  \rho_{B}^{1}%
)^{T}\Gamma_{BB^{\prime}}^{\mathcal{N}^{1}}]-q\sigma_{B^{\prime}}^{0},\\
Y_{B^{\prime}}^{1} &  \geq\operatorname{Tr}_{B}[(p\rho_{B}^{0})^{T}%
\Gamma_{BB^{\prime}}^{\mathcal{N}^{1}}+(\left(  1-p\right)  \rho_{B}^{1}%
)^{T}\Gamma_{BB^{\prime}}^{\mathcal{N}^{0}}]-\left(  1-q\right)
\sigma_{B^{\prime}}^{1}.
\end{align}
The dual program is given by%
\begin{equation}
\max_{Y_{B}\in\operatorname{Herm},W_{B^{\prime}},Z_{B^{\prime}}\geq
0}\operatorname{Tr}[Y_{B}]-q\operatorname{Tr}[W_{B^{\prime}}\sigma_{B^{\prime
}}^{0}]-\left(  1-q\right)  \operatorname{Tr}[Z_{B^{\prime}}\sigma_{B^{\prime
}}^{1}],
\end{equation}
subject to%
\begin{align}
W_{B^{\prime}},Z_{B^{\prime}} &  \leq I_{B^{\prime}},\\
Y_{B}\otimes I_{B^{\prime}} &  \leq p\rho_{B}^{0}\otimes W_{B^{\prime}%
}+\left(  1-p\right)  \rho_{B}^{1}\otimes Z_{B^{\prime}},\\
Y_{B}\otimes I_{B^{\prime}} &  \leq\left(  1-p\right)  \rho_{B}^{1}\otimes
W_{B^{\prime}}+p\rho_{B}^{0}\otimes Z_{B^{\prime}}.
\end{align}

\end{proposition}

\begin{proof}
Recall that an arbitrary CDS channel has the following form:%
\begin{equation}
\operatorname{id}_{X}\otimes\mathcal{N}_{B\rightarrow B^{\prime}}%
^{0}+\mathcal{P}_{X}\otimes\mathcal{N}_{B\rightarrow B^{\prime}}^{1},
\label{eq:app:CDS-SDP-1}
\end{equation}
where $\mathcal{N}_{B\rightarrow B^{\prime}}^{0}$ and $\mathcal{N}%
_{B\rightarrow B^{\prime}}^{1}$ are completely positive maps such that
$\mathcal{N}_{B\rightarrow B^{\prime}}^{0}+\mathcal{N}_{B\rightarrow
B^{\prime}}^{1}$ is trace preserving, and $\mathcal{P}_{X}$ is a unitary
channel that flips $|0\rangle\!\langle0|$ and $|1\rangle\!\langle1|$. This means
that its action on an input%
\begin{equation}
p|0\rangle\!\langle0|_{X}\otimes\rho_{B}^{0}+\left(  1-p\right)  |1\rangle\!\langle1|_{X}\otimes\rho_{B}^{1}%
\end{equation}
is as follows:%
\begin{align}
&  \left(  \operatorname{id}_{X}\otimes\mathcal{N}_{B\rightarrow B^{\prime}%
}^{0}+\mathcal{P}_{X}\otimes\mathcal{N}_{B\rightarrow B^{\prime}}^{1}\right)
\left(  p|0\rangle\!\langle0|_{X}\otimes\rho_{B}^{0}+\left(  1-p\right)
|1\rangle\!\langle1|_{X}\otimes\rho_{B}^{1}\right)  \nonumber\\
&  =\left(  \operatorname{id}_{X}\otimes\mathcal{N}_{B\rightarrow B^{\prime}%
}^{0}\right)  \left(  p|0\rangle\!\langle0|_{X}\otimes\rho_{B}^{0}+\left(
1-p\right)  |1\rangle\!\langle1|_{X}\otimes\rho_{B}^{1}\right)  \nonumber\\
&  \qquad+\left(  \mathcal{P}_{X}\otimes\mathcal{N}_{B\rightarrow B^{\prime}%
}^{1}\right)  \left(  p|0\rangle\!\langle0|_{X}\otimes\rho_{B}^{0}+\left(
1-p\right)  |1\rangle\!\langle1|_{X}\otimes\rho_{B}^{1}\right)  \\
&  =p|0\rangle\!\langle0|_{X}\otimes\mathcal{N}_{B\rightarrow B^{\prime}}%
^{0}(\rho_{B}^{0})+\left(  1-p\right)  |1\rangle\!\langle1|_{X}\otimes
\mathcal{N}_{B\rightarrow B^{\prime}}^{0}(\rho_{B}^{1})\nonumber\\
&  \qquad+p\mathcal{P}_{X}(|0\rangle\!\langle0|_{X})\otimes\mathcal{N}%
_{B\rightarrow B^{\prime}}^{1}(\rho_{B}^{0})\nonumber\\
&  \qquad+\left(  1-p\right)  \mathcal{P}_{X}(|1\rangle\!\langle1|_{X}%
)\otimes\mathcal{N}_{B\rightarrow B^{\prime}}^{1}(\rho_{B}^{1})\\
&  =p|0\rangle\!\langle0|_{X}\otimes\mathcal{N}_{B\rightarrow B^{\prime}}%
^{0}(\rho_{B}^{0})+\left(  1-p\right)  |1\rangle\!\langle1|_{X}\otimes
\mathcal{N}_{B\rightarrow B^{\prime}}^{0}(\rho_{B}^{1})\nonumber\\
&  \qquad+p|1\rangle\!\langle1|_{X}\otimes\mathcal{N}_{B\rightarrow B^{\prime}%
}^{1}(\rho_{B}^{0})+\left(  1-p\right)  |0\rangle\!\langle0|_{X}\otimes
\mathcal{N}_{B\rightarrow B^{\prime}}^{1}(\rho_{B}^{1})\\
&  =|0\rangle\!\langle0|_{X}\otimes\left[  \mathcal{N}_{B\rightarrow B^{\prime}%
}^{0}(p\rho_{B}^{0})+\mathcal{N}_{B\rightarrow B^{\prime}}^{1}(\left(
1-p\right)  \rho_{B}^{1})\right]  \nonumber\\
&  \qquad+|1\rangle\!\langle1|_{X}\otimes\left[  \mathcal{N}_{B\rightarrow
B^{\prime}}^{1}(p\rho_{B}^{0})+\mathcal{N}_{B\rightarrow B^{\prime}}%
^{0}(\left(  1-p\right)  \rho_{B}^{1})\right]  .
\end{align}
The semi-definite specifications for the completely positive maps
$\mathcal{N}_{B\rightarrow B^{\prime}}^{0}$ and $\mathcal{N}_{B\rightarrow
B^{\prime}}^{1}$ are as follows:%
\begin{align}
\Gamma_{BB^{\prime}}^{\mathcal{N}^{0}},\Gamma_{BB^{\prime}}^{\mathcal{N}^{1}}
&  \geq0,\\
\operatorname{Tr}_{B^{\prime}}[\Gamma_{BB^{\prime}}^{\mathcal{N}^{0}}%
+\Gamma_{BB^{\prime}}^{\mathcal{N}^{1}}] &  =I_{B}.
\end{align}
Furthermore, the output state is given by%
\begin{multline}
|0\rangle\!\langle0|_{X}\otimes\left[  \operatorname{Tr}_{B}[(p\rho_{B}^{0}%
)^{T}\Gamma_{BB^{\prime}}^{\mathcal{N}^{0}}]+\operatorname{Tr}_{B}[(\left(
1-p\right)  \rho_{B}^{1})^{T}\Gamma_{BB^{\prime}}^{\mathcal{N}^{1}}]\right]
\\
+|1\rangle\!\langle1|_{X}\otimes\left[  \operatorname{Tr}_{B}[(p\rho_{B}%
^{0})^{T}\Gamma_{BB^{\prime}}^{\mathcal{N}^{1}}]+\operatorname{Tr}%
_{B}[(\left(  1-p\right)  \rho_{B}^{1})^{T}\Gamma_{BB^{\prime}}^{\mathcal{N}%
^{0}}]\right]  .
\label{eq:app:CDS-SDP-last}
\end{multline}
Recall that the dual semi-definite program for computing the normalized trace
distance of two quantum states $\rho$ and $\sigma$ is as follows (see, e.g., \cite{Wang2019states}):%
\begin{equation}
\frac{1}{2}\left\Vert \rho-\sigma\right\Vert _{1}=\min_{Y\geq0}\left\{
\operatorname{Tr}[Y]:Y\geq\rho-\sigma\right\}  .
\end{equation}
So in this case, it follows that%
\begin{multline}
\min_{\mathcal{N}_{XB\rightarrow XB^{\prime}}\in\text{$\operatorname{CDS}$}%
}\frac{1}{2}\left\Vert \mathcal{N}_{XB\rightarrow XB^{\prime}}(\rho
_{XB})-\sigma_{XB^{\prime}}\right\Vert _{1}\\
=\min_{Y_{XB^{\prime}}\geq0}\left\{  \operatorname{Tr}[Y_{XB^{\prime}%
}]:Y_{XB^{\prime}}\geq\mathcal{N}_{XB\rightarrow XB^{\prime}}(\rho
_{XB})-\sigma_{XB^{\prime}}\right\}  ,
\end{multline}
where we have now called the output system $X$ for clarity. Then the SDP\ for the
trace-distance conversion error is given by%
\begin{equation}
\min_{Y_{XB^{\prime}}\geq0}\operatorname{Tr}[Y_{XB^{\prime}}]
\end{equation}
subject to%
\begin{align}
\Gamma_{BB^{\prime}}^{\mathcal{N}^{0}},\Gamma_{BB^{\prime}}^{\mathcal{N}^{1}}
&  \geq0,\\
\operatorname{Tr}_{B^{\prime}}[\Gamma_{BB^{\prime}}^{\mathcal{N}^{0}}%
+\Gamma_{BB^{\prime}}^{\mathcal{N}^{1}}] &  =I_{B},
\end{align}%
\begin{multline}
Y_{XB^{\prime}}\geq|0\rangle\!\langle0|_{X}\otimes\left[  \operatorname{Tr}%
_{B}[(p\rho_{B}^{0})^{T}\Gamma_{BB^{\prime}}^{\mathcal{N}^{0}}%
]+\operatorname{Tr}_{B}[(\left(  1-p\right)  \rho_{B}^{1})^{T}\Gamma
_{BB^{\prime}}^{\mathcal{N}^{1}}]-q\sigma_{B^{\prime}}^{0}\right]  \\
+|1\rangle\!\langle1|_{X}\otimes\left[  \operatorname{Tr}_{B}[(p\rho_{B}%
^{0})^{T}\Gamma_{BB^{\prime}}^{\mathcal{N}^{1}}]+\operatorname{Tr}%
_{B}[(\left(  1-p\right)  \rho_{B}^{1})^{T}\Gamma_{BB^{\prime}}^{\mathcal{N}%
^{0}}]-\left(  1-q\right)  \sigma_{B^{\prime}}^{1}\right]  .
\end{multline}
It is clear that the optimal $Y_{XB^{\prime}}$ respects the classical--quantum
structure. So this means that the final SDP\ can be written as follows:%
\begin{equation}
\min_{Y_{B^{\prime}}^{0},Y_{B^{\prime}}^{1},\Gamma_{BB^{\prime}}%
^{\mathcal{N}^{0}},\Gamma_{BB^{\prime}}^{\mathcal{N}^{1}}\geq0}%
\operatorname{Tr}[Y_{B^{\prime}}^{0}+Y_{B^{\prime}}^{1}]
\end{equation}
subject to%
\begin{equation}
\operatorname{Tr}_{B^{\prime}}[\Gamma_{BB^{\prime}}^{\mathcal{N}^{0}}%
+\Gamma_{BB^{\prime}}^{\mathcal{N}^{1}}]=I_{B},
\end{equation}%
\begin{align}
Y_{B^{\prime}}^{0} &  \geq\operatorname{Tr}_{B}[(p\rho_{B}^{0})^{T}%
\Gamma_{BB^{\prime}}^{\mathcal{N}^{0}}]+\operatorname{Tr}_{B}[(\left(
1-p\right)  \rho_{B}^{1})^{T}\Gamma_{BB^{\prime}}^{\mathcal{N}^{1}}%
]-q\sigma_{B^{\prime}}^{0},\\
Y_{B^{\prime}}^{1} &  \geq\operatorname{Tr}_{B}[(p\rho_{B}^{0})^{T}%
\Gamma_{BB^{\prime}}^{\mathcal{N}^{1}}]+\operatorname{Tr}_{B}[(\left(
1-p\right)  \rho_{B}^{1})^{T}\Gamma_{BB^{\prime}}^{\mathcal{N}^{0}}]-\left(
1-q\right)  \sigma_{B^{\prime}}^{1}.
\end{align}

Now we compute the dual of the semi-definite program above. Recall the
standard form of primal and dual\ SDPs:%
\begin{align}
&  \max_{X\geq0}\left\{  \operatorname{Tr}[AX]:\Phi(X)\leq B\right\}  ,\\
&  \min_{Y\geq0}\left\{  \operatorname{Tr}[BY]:\Phi^{\dag}(Y)\geq A\right\}  .
\end{align}
From inspecting the above, we see that%
\begin{align}
Y  &  =%
\begin{bmatrix}
Y_{B^{\prime}}^{0} & 0 & 0 & 0\\
0 & Y_{B^{\prime}}^{1} & 0 & 0\\
0 & 0 & \Gamma_{BB^{\prime}}^{\mathcal{N}^{0}} & 0\\
0 & 0 & 0 & \Gamma_{BB^{\prime}}^{\mathcal{N}^{1}}%
\end{bmatrix}
,\qquad B=%
\begin{bmatrix}
I & 0 & 0 & 0\\
0 & I & 0 & 0\\
0 & 0 & 0 & 0\\
0 & 0 & 0 & 0
\end{bmatrix}
,\\
\Phi^{\dag}(Y)  &  =%
\begin{bmatrix}
\operatorname{Tr}_{B^{\prime}}[\Gamma_{BB^{\prime}}^{\mathcal{N}^{0}}%
+\Gamma_{BB^{\prime}}^{\mathcal{N}^{1}}] & 0 & 0 & 0\\
0 & -\operatorname{Tr}_{B^{\prime}}[\Gamma_{BB^{\prime}}^{\mathcal{N}^{0}%
}+\Gamma_{BB^{\prime}}^{\mathcal{N}^{1}}] & 0 & 0\\
0 & 0 & Y_{B^{\prime}}^{0}-Z_{B^{\prime}}^{0} & 0\\
0 & 0 & 0 & Y_{B^{\prime}}^{1}-Z_{B^{\prime}}^{1}%
\end{bmatrix}
,\\
A  &  =%
\begin{bmatrix}
I_{B} & 0 & 0 & 0\\
0 & -I_{B} & 0 & 0\\
0 & 0 & -q\sigma_{B^{\prime}}^{0} & 0\\
0 & 0 & 0 & -\left(  1-q\right)  \sigma_{B^{\prime}}^{1}%
\end{bmatrix}
.
\end{align}
where%
\begin{align}
Z_{B^{\prime}}^{0}  &   \coloneqq \operatorname{Tr}_{B}[(p\rho_{B}^{0})^{T}%
\Gamma_{BB^{\prime}}^{\mathcal{N}^{0}}+(\left(  1-p\right)  \rho_{B}^{1}%
)^{T}\Gamma_{BB^{\prime}}^{\mathcal{N}^{1}}],\\
Z_{B^{\prime}}^{1}  &   \coloneqq \operatorname{Tr}_{B}[(p\rho_{B}^{0})^{T}%
\Gamma_{BB^{\prime}}^{\mathcal{N}^{1}}+(\left(  1-p\right)  \rho_{B}^{1}%
)^{T}\Gamma_{BB^{\prime}}^{\mathcal{N}^{0}}].
\end{align}
Now we need to find the adjoint map of $\Phi^{\dag}$. Consider that%
\begin{align}
&  \operatorname{Tr}[\Phi^{\dag}(Y)X]\nonumber\\
&  =\operatorname{Tr}[\operatorname{Tr}_{B^{\prime}}[\Gamma_{BB^{\prime}%
}^{\mathcal{N}^{0}}+\Gamma_{BB^{\prime}}^{\mathcal{N}^{1}}](X_{B}^{1}%
-X_{B}^{2})]\nonumber\\
&  \qquad+\operatorname{Tr}[(Y_{B^{\prime}}^{0}-\operatorname{Tr}_{B}%
[(p\rho_{B}^{0})^{T}\Gamma_{BB^{\prime}}^{\mathcal{N}^{0}}+(\left(
1-p\right)  \rho_{B}^{1})^{T}\Gamma_{BB^{\prime}}^{\mathcal{N}^{1}%
}])X_{B^{\prime}}^{3}]\nonumber\\
&  \qquad+\operatorname{Tr}[(Y_{B^{\prime}}^{1}-\operatorname{Tr}_{B}%
[(p\rho_{B}^{0})^{T}\Gamma_{BB^{\prime}}^{\mathcal{N}^{1}}+(\left(
1-p\right)  \rho_{B}^{1})^{T}\Gamma_{BB^{\prime}}^{\mathcal{N}^{0}%
}])X_{B^{\prime}}^{4}]\\
&  =\operatorname{Tr}[(\Gamma_{BB^{\prime}}^{\mathcal{N}^{0}}+\Gamma
_{BB^{\prime}}^{\mathcal{N}^{1}})((X_{B}^{1}-X_{B}^{2})\otimes I_{B^{\prime}%
})]\nonumber\\
&  \qquad+\operatorname{Tr}[Y_{B^{\prime}}^{0}X_{B^{\prime}}^{3}%
]-\operatorname{Tr}[\Gamma_{BB^{\prime}}^{\mathcal{N}^{0}}((p\rho_{B}^{0}%
)^{T}\otimes X_{B^{\prime}}^{3})]\nonumber\\
&  \qquad-\operatorname{Tr}[(\Gamma_{BB^{\prime}}^{\mathcal{N}^{1}})(\left(
1-p\right)  \rho_{B}^{1})^{T}\otimes X_{B^{\prime}}^{3})]\nonumber\\
&  \qquad+\operatorname{Tr}[Y_{B^{\prime}}^{1}X_{B^{\prime}}^{4}%
]-\operatorname{Tr}[\Gamma_{BB^{\prime}}^{\mathcal{N}^{1}}((p\rho_{B}^{0}%
)^{T}\otimes X_{B^{\prime}}^{4})]\nonumber\\
&  \qquad-\operatorname{Tr}[\Gamma_{BB^{\prime}}^{\mathcal{N}^{0}}((\left(
1-p\right)  \rho_{B}^{1})^{T}\otimes X_{B^{\prime}}^{4})]\\
&  =\operatorname{Tr}[Y_{B^{\prime}}^{0}X_{B^{\prime}}^{3}]+\operatorname{Tr}%
[Y_{B^{\prime}}^{1}X_{B^{\prime}}^{4}]\nonumber\\
&  \qquad+\operatorname{Tr}[\Gamma_{BB^{\prime}}^{\mathcal{N}^{0}}((X_{B}%
^{1}-X_{B}^{2})\otimes I_{B^{\prime}}-p(\rho_{B}^{0})^{T}\otimes X_{B^{\prime
}}^{3}-\left(  1-p\right)  (\rho_{B}^{1})^{T}\otimes X_{B^{\prime}}%
^{4})]\nonumber\\
&  \qquad+\operatorname{Tr}[\Gamma_{BB^{\prime}}^{\mathcal{N}^{1}}((X_{B}%
^{1}-X_{B}^{2})\otimes I_{B^{\prime}}-\left(  1-p\right)  (\rho_{B}^{1}%
)^{T}\otimes X_{B^{\prime}}^{3}-p(\rho_{B}^{0})^{T}\otimes X_{B^{\prime}}%
^{4})].
\end{align}
So this means that%
\begin{align}
\Phi(X)  &  =%
\begin{bmatrix}
X_{B^{\prime}}^{3} & 0 & 0 & 0\\
0 & X_{B^{\prime}}^{4} & 0 & 0\\
0 & 0 & W_{B}^{0} & 0\\
0 & 0 & 0 & W_{B}^{1}%
\end{bmatrix}
,\\
W_{B}^{0}  &   \coloneqq (X_{B}^{1}-X_{B}^{2})\otimes I_{B^{\prime}}-p(\rho_{B}%
^{0})^{T}\otimes X_{B^{\prime}}^{3}-\left(  1-p\right)  (\rho_{B}^{1}%
)^{T}\otimes X_{B^{\prime}}^{4},\\
W_{B}^{1}  &   \coloneqq (X_{B}^{1}-X_{B}^{2})\otimes I_{B^{\prime}}-\left(
1-p\right)  (\rho_{B}^{1})^{T}\otimes X_{B^{\prime}}^{3}-p(\rho_{B}^{0}%
)^{T}\otimes X_{B^{\prime}}^{4}.
\end{align}
Then the dual program is given by%
\begin{equation}
\max_{X_{B}^{1},X_{B}^{2},X_{B^{\prime}}^{3},X_{B^{\prime}}^{4}\geq
0}\operatorname{Tr}\left[
\begin{bmatrix}
I_{B} & 0 & 0 & 0\\
0 & -I_{B} & 0 & 0\\
0 & 0 & -q\sigma_{B^{\prime}}^{0} & 0\\
0 & 0 & 0 & -\left(  1-q\right)  \sigma_{B^{\prime}}^{1}%
\end{bmatrix}%
\begin{bmatrix}
X_{B}^{1} & 0 & 0 & 0\\
0 & X_{B}^{2} & 0 & 0\\
0 & 0 & X_{B^{\prime}}^{3} & 0\\
0 & 0 & 0 & X_{B^{\prime}}^{4}%
\end{bmatrix}
\right]
\end{equation}
subject to%
\begin{equation}
X_{B^{\prime}}^{3}\leq I_{B^{\prime}},\quad X_{B^{\prime}}^{4}\leq
I_{B^{\prime}},
\end{equation}%
\begin{align}
(X_{B}^{1}-X_{B}^{2})\otimes I_{B^{\prime}}-p(\rho_{B}^{0})^{T}\otimes
X_{B^{\prime}}^{3}-\left(  1-p\right)  (\rho_{B}^{1})^{T}\otimes X_{B^{\prime
}}^{4}  &  \leq0,\\
(X_{B}^{1}-X_{B}^{2})\otimes I_{B^{\prime}}-\left(  1-p\right)  (\rho_{B}%
^{1})^{T}\otimes X_{B^{\prime}}^{3}-p(\rho_{B}^{0})^{T}\otimes X_{B^{\prime}%
}^{4}  &  \leq0.
\end{align}
This can be simplified to the following:%
\begin{equation}
\max_{Y_{B}\in\operatorname{Herm},X_{B^{\prime}}^{3},X_{B^{\prime}}^{4}\geq
0}\operatorname{Tr}[Y_{B}]-q\operatorname{Tr}[X_{B^{\prime}}^{3}%
\sigma_{B^{\prime}}^{0}]-\left(  1-q\right)  \operatorname{Tr}[X_{B^{\prime}%
}^{4}\sigma_{B^{\prime}}^{1}],
\end{equation}
subject to%
\begin{align}
X_{B^{\prime}}^{3}  &  \leq I_{B^{\prime}},\quad X_{B^{\prime}}^{4}\leq
I_{B^{\prime}},\\
Y_{B}\otimes I_{B^{\prime}}  &  \leq p(\rho_{B}^{0})^{T}\otimes X_{B^{\prime}%
}^{3}+\left(  1-p\right)  (\rho_{B}^{1})^{T}\otimes X_{B^{\prime}}^{4},\\
Y_{B}\otimes I_{B^{\prime}}  &  \leq\left(  1-p\right)  (\rho_{B}^{1}%
)^{T}\otimes X_{B^{\prime}}^{3}+p(\rho_{B}^{0})^{T}\otimes X_{B^{\prime}}^{4}.
\end{align}
Then we can set $X_{B^{\prime}}^{3}=W_{B^{\prime}}$ and $X_{B^{\prime}}%
^{4}=Z_{B^{\prime}}$ to get%
\begin{equation}
\max_{Y_{B}\in\operatorname{Herm},W_{B^{\prime}},Z_{B^{\prime}}\geq
0}\operatorname{Tr}[Y_{B}]-q\operatorname{Tr}[W_{B^{\prime}}\sigma_{B^{\prime
}}^{0}]-\left(  1-q\right)  \operatorname{Tr}[Z_{B^{\prime}}\sigma_{B^{\prime
}}^{1}],
\end{equation}
subject to%
\begin{align}
W_{B^{\prime}},Z_{B^{\prime}}  &  \leq I_{B^{\prime}},\\
Y_{B}\otimes I_{B^{\prime}}  &  \leq p(\rho_{B}^{0})^{T}\otimes W_{B^{\prime}%
}+\left(  1-p\right)  (\rho_{B}^{1})^{T}\otimes Z_{B^{\prime}},\\
Y_{B}\otimes I_{B^{\prime}}  &  \leq\left(  1-p\right)  (\rho_{B}^{1}%
)^{T}\otimes W_{B^{\prime}}+p(\rho_{B}^{0})^{T}\otimes Z_{B^{\prime}}.
\end{align}
We can finally make the substitution $Y_{B}\rightarrow Y_{B}^{T}$ and the
optimal value is unchanged. Since the operators on the right-hand side of the
inequalities just above are separable, the partial transpose has no effect and
can be removed. This concludes the proof.
\end{proof}

\subsection{Minimum error probability and minimum conversion error (in terms of trace distance) to
infinite-resource states}

\comment{The guessing probability of the state $\rho_{XB}$ is defined as follows:%
\begin{equation}
p_{g}(\rho_{XB}) \coloneqq \max_{\left\{  \Lambda_{B}^{x}\right\}  _{x\in\left\{
0,1\right\}  }}\sum_{x\in\left\{  0,1\right\}  }p(x)\operatorname{Tr}%
[\Lambda_{B}^{x}\rho_{B}^{x}],
\end{equation}
where $p(0)=p$ and $p(1)=1-p$, and the optimization is over all POVMs. The
minimum error probability is then given by%
\begin{equation}
p_{\operatorname{err}}(\rho_{XB}) \coloneqq 1-p_{g}(\rho_{XB}).
\end{equation}
It is a well known result that the guessing probability is equal to%
\begin{equation}
p_{g}(\rho_{XB})=\frac{1}{2}\left(  1+\left\Vert p\rho_{B}^{0}-\left(
1-p\right)  \rho_{B}^{1}\right\Vert _{1}\right)  ,
\end{equation}
which implies that}
As mentioned in the main text, the minimum error probability is given by
\begin{equation}
p_{\operatorname{err}}(\rho_{XB})=\frac{1}{2}\left(  1-\left\Vert p\rho
_{B}^{0}-\left(  1-p\right)  \rho_{B}^{1}\right\Vert _{1}\right)  .
\end{equation}
An alternative expression for it is given by
\begin{equation}
p_{\operatorname{err}}(\rho_{XB})=\max_{Y_{B}\in\operatorname{Herm}}\left\{
\operatorname{Tr}[Y_{B}]:Y_{B}\leq p\rho_{B}^{0},Y_{B}\leq\left(  1-p\right)
\rho_{B}^{1}\right\}, \label{eq:GLB-err-prob}%
\end{equation}
where $\operatorname{Herm}$ denotes the set of Hermitian operators acting on the system $B$.
Note that the maximising operator $Y_B$ on the right hand side of the above equation
is called the \textquotedblleft greatest lower bound (GLB) operator\textquotedblright\ of
the operators $p\rho_{B}^{0}$ and $ ( 1-p)  \rho_{B}^{1}$.
\comment{It seems less well known that the minimum error probability is equal to the trace of the
\textquotedblleft greatest lower bound (GLB) operator\textquotedblright\ of
$p\rho_{B}^{0}$ and $\left(  1-p\right)  \rho_{B}^{1}$. This is the following
statement:}%
The GLB operator is defined in Eq.~(84)\ of \cite{AM14},
and the above result was established as Lemma~A.7 of the same paper.

Consider the infinite-resource state
\begin{equation}
\gamma^{(\infty,q)}_{XQ}  \coloneqq q|0\rangle\!\langle0|_{X}\otimes|0\rangle\!\langle0|_{Y}+\left(
1-q\right)  |1\rangle\!\langle1|_{X}\otimes|1\rangle\!\langle1|_{Y},
\end{equation}
which is the $(M,q)$-golden unit (Definition~\ref{golden-unit}) with $M=\infty$.
It follows from Lemma~\ref{lem:InfToInf} that for all $q,q^{\prime}\in\left[  0,1\right]$, it is possible to perform the
transformation $\mathcal{N}_{XQ}(\gamma^{(\infty,q)}_{XQ})=\gamma^{(\infty,q')}_{XQ}$,
where $\mathcal{N}_{XQ}$ is a CDS map.

\begin{lemma}
\label{lem:max-resource-conv}Let $\rho_{XB}\equiv(p, \rho^0, \rho^1)$ be a c-q state. Then the following equality holds for
all $q,q^{\prime}\in\left[  0,1\right]  $:%
\begin{multline}
\min_{\mathcal{N}_{XB\rightarrow XQ}\in\operatorname{CDS}}\frac{1}%
{2}\left\Vert \mathcal{N}_{XB\rightarrow XQ}(\rho_{XB})-\gamma^{(\infty,q)}_{XQ}\right\Vert _{1}\\
=\min_{\mathcal{N}_{XB\rightarrow XQ}\in\operatorname{CDS}}\frac{1}%
{2}\left\Vert \mathcal{N}_{XB\rightarrow XQ}(\rho_{XB})-\gamma^{(\infty,q')}_{XQ}\right\Vert _{1}.
\end{multline}

\end{lemma}

\begin{proof}
We first establish the inequality $\geq$. Let $\mathcal{N}_{XQ}^{\prime}$ be the
CDS\ channel that converts $\gamma^{(\infty,q)}_{XQ}$ to $\gamma^{(\infty,q')}_{XQ}$. Then
consider that%
\begin{align}
&  \min_{\mathcal{N}_{XB\rightarrow XQ}\in\text{$\operatorname{CDS}$}}\frac
{1}{2}\left\Vert \mathcal{N}_{XB\rightarrow XQ}(\rho_{XB})-\gamma^{(\infty,q')}_{XQ}\right\Vert _{1}\nonumber\\
&  =\min_{\mathcal{N}_{XB\rightarrow XQ}\in\text{$\operatorname{CDS}$}}%
\frac{1}{2}\left\Vert \mathcal{N}_{XB\rightarrow XQ}(\rho_{XB})-\mathcal{N}%
_{XQ}^{\prime}(\gamma^{(\infty,q)}_{XQ})\right\Vert \\
&  \leq\min_{\mathcal{N}_{XB\rightarrow XQ}\in\text{$\operatorname{CDS}$}%
}\frac{1}{2}\left\Vert (\mathcal{N}_{XQ}^{\prime}\circ\mathcal{N}%
_{XB\rightarrow XQ})(\rho_{XB})-\mathcal{N}_{XQ}^{\prime}(\gamma^{(\infty,q)}_{XQ})\right\Vert _{1}\\
&  \leq\min_{\mathcal{N}_{XB\rightarrow XQ}\in\text{$\operatorname{CDS}$}%
}\frac{1}{2}\left\Vert \mathcal{N}_{XB\rightarrow XQ}(\rho_{XB})-\gamma^{(\infty,q)}_{XQ}\right\Vert _{1}.
\end{align}
The first inequality follows because $\mathcal{N}_{XQ}^{\prime}\circ
\mathcal{N}_{XB\rightarrow XQ}$ is a member of the set of CDS channels. The
second inequality follows from the DPI under the trace distance. We can
then apply the same argument to arrive at the opposite inequality.
\end{proof}

\begin{lemma}
\label{thm:CDS-dist-max-res-err-prob}
Let $\rho_{XB}\equiv(p, \rho^0, \rho^1)$ be a c-q state. Then the following equality holds%
\begin{equation}
\min_{\substack{q\in\left[  0,1\right]  ,\\\mathcal{N}_{XB\rightarrow XQ}%
\in\text{$\operatorname{CDS}$}}}\frac{1}{2}\left\Vert \mathcal{N}%
_{XB\rightarrow XQ}(\rho_{XB})-\gamma^{(\infty,q)}_{XQ}\right\Vert _{1}%
=p_{\operatorname{err}}(\rho_{XB}). \label{eq:conv-dist-to-err-prob}%
\end{equation}

\end{lemma}

\begin{proof}
We first establish the inequality%
\begin{equation}
\min_{\substack{q\in\left[  0,1\right]  ,\\\mathcal{N}_{XB\rightarrow XQ}%
\in\text{$\operatorname{CDS}$}}}\frac{1}{2}\left\Vert \mathcal{N}%
_{XB\rightarrow XQ}(\rho_{XB})-\gamma^{(\infty,q)}_{XQ}\right\Vert _{1}\leq
p_{\operatorname{err}}(\rho_{XB}),
\end{equation}
by demonstrating the existence of a value of $q\in\left[  0,1\right]  $ and a
CDS\ channel for which the left-hand side is equal to $p_{\operatorname{err}%
}(\rho_{XB})$. From Lemma~\ref{lem:max-resource-conv}, it follows that the
left-hand side of \eqref{eq:conv-dist-to-err-prob} is independent of
$q\in\left[  0,1\right]  $. So we can pick $q=p$, and the value is unchanged.
Now consider that the channel used in state discrimination is a simple local
channel of the following form:%
\begin{equation}
\mathcal{M}_{B\rightarrow Q}(\sigma_{B}) \coloneqq \operatorname{Tr}[\Lambda_{B}%
\sigma_{B}]|0\rangle\!\langle0|_{Q}+\operatorname{Tr}[(I_{B}-\Lambda_{B}%
)\sigma_{B}]|1\rangle\!\langle1|_{Q}, \label{eq:state-disc-ch-1}%
\end{equation}
and so $\operatorname{id}_{X}\otimes\mathcal{M}_{B\rightarrow Q}$ is a CDS
channel. Acting with it on $\rho_{XB}$ leads to the following state:
\begin{align}
&  \mathcal{M}_{B\rightarrow Q}(\rho_{XB})\nonumber\\
&  =p|0\rangle\!\langle0|_{X}\otimes\left(  \operatorname{Tr}[\Lambda_{B}%
\rho_{B}^{0}]|0\rangle\!\langle0|_{Q}+\operatorname{Tr}[(I_{B}-\Lambda_{B}%
)\rho_{B}^{0}]|1\rangle\!\langle1|_{Q}\right) \nonumber\\
&  \quad+\left(  1-p\right)  |1\rangle\!\langle1|_{X}\otimes\left(
\operatorname{Tr}[\Lambda_{B}\rho_{B}^{1}]|0\rangle\!\langle0|_{Q}%
+\operatorname{Tr}[(I_{B}-\Lambda_{B})\rho_{B}^{1}]|1\rangle\!\langle
1|_{Q}\right) \\
&  =p\operatorname{Tr}[\Lambda_{B}\rho_{B}^{0}]|0\rangle\!\langle0|_{X}%
\otimes|0\rangle\!\langle0|_{Q}+p\operatorname{Tr}[(I_{B}-\Lambda_{B})\rho
_{B}^{0}]|0\rangle\!\langle0|_{X}\otimes|1\rangle\!\langle1|_{Q}\nonumber\\
&  \quad+\left(  1-p\right)  \operatorname{Tr}[\Lambda_{B}\rho_{B}%
^{1}]|1\rangle\!\langle1|_{X}\otimes|0\rangle\!\langle0|_{Q}+\left(  1-p\right)
\operatorname{Tr}[(I_{B}-\Lambda_{B})\rho_{B}^{1}]|1\rangle\!\langle
1|_{X}\otimes|1\rangle\!\langle1|_{Q}. \label{eq:state-disc-ch-2}%
\end{align}
Since this is a particular choice, it follows that%
\begin{equation}
\min_{\substack{q\in\left[  0,1\right]  ,\\\mathcal{N}_{XB\rightarrow XQ}%
\in\text{$\operatorname{CDS}$}}}\frac{1}{2}\left\Vert \mathcal{N}%
_{XB\rightarrow XQ}(\rho_{XB})-\gamma^{(\infty,q)}_{XQ}\right\Vert _{1}
\leq\frac{1}{2}\left\Vert \mathcal{M}_{B\rightarrow Q}(\rho_{XB})-\gamma^{(\infty,p)}_{XQ}\right\Vert _{1}.
\end{equation}
Now let us compute the trace distance between $\mathcal{M}_{B\rightarrow
Q}(\rho_{XB})$ and the simple state $\gamma^{(\infty,p)}_{XQ}$:%
\begin{align}
&  \left\Vert \mathcal{M}_{B\rightarrow Q}(\rho_{XB})-\gamma^{(\infty,p)}_{XQ}\right\Vert _{1}\nonumber\\
&  =\left\Vert
\begin{array}
[c]{c}%
p\operatorname{Tr}[\Lambda_{B}\rho_{B}^{0}]|0\rangle\!\langle0|_{X}%
\otimes|0\rangle\!\langle0|_{Q}+p\operatorname{Tr}[(I_{B}-\Lambda_{B})\rho
_{B}^{0}]|0\rangle\!\langle0|_{X}\otimes|1\rangle\!\langle1|_{Q}\\
+\left(  1-p\right)  \operatorname{Tr}[\Lambda_{B}\rho_{B}^{1}]|1\rangle
\!\langle1|_{X}\otimes|0\rangle\!\langle0|_{Q}+\left(  1-p\right)
\operatorname{Tr}[(I_{B}-\Lambda_{B})\rho_{B}^{1}]|1\rangle\!\langle
1|_{X}\otimes|1\rangle\!\langle1|_{Q}\\
-p|0\rangle\!\langle0|_{X}\otimes|0\rangle\!\langle0|_{Q}-\left(  1-p\right)
|1\rangle\!\langle1|_{X}\otimes|1\rangle\!\langle1|_{Q}%
\end{array}
\right\Vert _{1}\\
&  =\left\Vert
\begin{array}
[c]{c}%
\left(  p\operatorname{Tr}[\Lambda_{B}\rho_{B}^{0}]-p\right)  |0\rangle\!
\langle0|_{X}\otimes|0\rangle\!\langle0|_{Q}+p\operatorname{Tr}[(I_{B}%
-\Lambda_{B})\rho_{B}^{0}]|0\rangle\!\langle0|_{X}\otimes|1\rangle\!\langle
1|_{Q}\\
+\left(  1-p\right)  \operatorname{Tr}[\Lambda_{B}\rho_{B}^{1}]|1\rangle\!
\langle1|_{X}\otimes|0\rangle\!\langle0|_{Q}\\
+\left[  \left(  1-p\right)  \operatorname{Tr}[(I_{B}-\Lambda_{B})\rho_{B}%
^{1}]-\left(  1-p\right)  \right]  |1\rangle\!\langle1|_{X}\otimes
|1\rangle\!\langle1|_{Q}%
\end{array}
\right\Vert _{1}\\
&  =\left\vert p\operatorname{Tr}[\Lambda_{B}\rho_{B}^{0}]-p\right\vert
+p\operatorname{Tr}[(I_{B}-\Lambda_{B})\rho_{B}^{0}]\nonumber\\
&  \qquad+\left(  1-p\right)  \operatorname{Tr}[\Lambda_{B}\rho_{B}%
^{1}]+\left\vert \left(  1-p\right)  \operatorname{Tr}[(I_{B}-\Lambda_{B}%
)\rho_{B}^{1}]-\left(  1-p\right)  \right\vert \\
&  =\left\vert p\operatorname{Tr}[(I_{B}-\Lambda_{B})\rho_{B}^{0}]\right\vert
+p\operatorname{Tr}[(I_{B}-\Lambda_{B})\rho_{B}^{0}]\nonumber\\
&  \qquad+\left(  1-p\right)  \operatorname{Tr}[\Lambda_{B}\rho_{B}%
^{1}]+\left\vert \left(  1-p\right)  \operatorname{Tr}[\Lambda_{B}\rho_{B}%
^{1}]\right\vert \\
&  =2\left(  p\operatorname{Tr}[(I_{B}-\Lambda_{B})\rho_{B}^{0}]+\left(
1-p\right)  \operatorname{Tr}[\Lambda_{B}\rho_{B}^{1}]\right) \\
&  =2p_{\operatorname{err}}(\rho_{XB}).
\end{align}

We now establish the opposite inequality. Since the value of $q$ does not
matter, let us set it to $1/2$, so that
\begin{equation}
\omega_{XY}(q)=\frac{1}{2}|0\rangle\!\langle0|_{X}\otimes|0\rangle\!\langle
0|_{Q}+\frac{1}{2}|1\rangle\!\langle1|_{X}\otimes|1\rangle\!\langle1|_{Q}.
\end{equation}
By applying Proposition~\ref{prop:CDS-conv-dist-SDPs}\ and weak duality of
semi-definite programming, we conclude that the trace-distance conversion error
\begin{equation}
\min_{\substack{q\in\left[  0,1\right]  ,\\\mathcal{N}_{XB\rightarrow XQ}%
\in\text{$\operatorname{CDS}$}}}\frac{1}{2}\left\Vert \mathcal{N}%
_{XB\rightarrow XQ}(\rho_{XB})-\omega_{XQ}(q)\right\Vert _{1}%
\end{equation}
is not smaller than the optimal value of the following SDP:%
\begin{equation}
\max_{Y_{B}\in\operatorname{Herm},W_{B^{\prime}},Z_{B^{\prime}}\geq
0}\operatorname{Tr}[Y_{B}]-\frac{1}{2}\left(  \operatorname{Tr}[W_{B^{\prime}%
}|0\rangle\!\langle0|]+\operatorname{Tr}[Z_{B^{\prime}}|1\rangle\!\langle
1|]\right)  ,
\end{equation}
subject to%
\begin{align}
W_{B^{\prime}},Z_{B^{\prime}}  &  \leq I_{B^{\prime}},\\
Y_{B}\otimes I_{B^{\prime}}  &  \leq p\rho_{B}^{0}\otimes W_{B^{\prime}%
}+\left(  1-p\right)  \rho_{B}^{1}\otimes Z_{B^{\prime}},\\
Y_{B}\otimes I_{B^{\prime}}  &  \leq\left(  1-p\right)  \rho_{B}^{1}\otimes
W_{B^{\prime}}+p\rho_{B}^{0}\otimes Z_{B^{\prime}}.
\end{align}
Let us then pick $W_{B^{\prime}}=|1\rangle\!\langle1|_{B^{\prime}}$ and
$Z_{B^{\prime}}=|0\rangle\!\langle0|_{B^{\prime}}$. Then the SDP\ simplifies as
follows:%
\begin{equation}
\max_{Y_{B}\in\operatorname{Herm}}\operatorname{Tr}[Y_{B}],
\end{equation}
subject to%
\begin{align}
Y_{B}\otimes I_{B^{\prime}}  &  \leq p\rho_{B}^{0}\otimes|1\rangle\!
\langle1|_{B^{\prime}}+\left(  1-p\right)  \rho_{B}^{1}\otimes|0\rangle\!
\langle0|_{B^{\prime}},\\
Y_{B}\otimes I_{B^{\prime}}  &  \leq\left(  1-p\right)  \rho_{B}^{1}%
\otimes|1\rangle\!\langle1|_{B^{\prime}}+p\rho_{B}^{0}\otimes|0\rangle\!
\langle0|_{B^{\prime}}.
\end{align}
Since we can write%
\begin{equation}
Y_{B}\otimes I_{B^{\prime}}=Y_{B}\otimes|0\rangle\!\langle0|_{B^{\prime}}%
+Y_{B}\otimes|1\rangle\!\langle1|_{B^{\prime}},
\end{equation}
the above constraints are equivalent to the following:%
\begin{equation}
Y_{B}\leq p\rho_{B}^{0},\quad Y_{B}\leq\left(  1-p\right)  \rho_{B}^{1}.
\end{equation}
This SDP\ is thus equal to the following one:%
\begin{equation}
\max_{Y_{B}\in\operatorname{Herm}}\left\{  \operatorname{Tr}[Y_{B}]:Y_{B}\leq
p\rho_{B}^{0},\quad Y_{B}\leq\left(  1-p\right)  \rho_{B}^{1}\right\}  .
\end{equation}
This quantity is precisely the trace of the greatest lower bound operator, and
so we conclude by applying \eqref{eq:GLB-err-prob}.
\end{proof}

\begin{lemma}
\label{thm:CPTP-A-dist-max-res-err-prob}
Let $\rho_{XB}\equiv\left(  p,\rho^{0}_B,
\rho^{1}_B\right)$ be a c-q state. Then the following equality holds%
\begin{equation}
\min_{\mathcal{N}_{B\rightarrow Q}%
\in\text{$\operatorname{CPTP}_A$}}\frac{1}{2}\left\Vert \mathcal{N}%
_{B\rightarrow Q}(\rho_{XB})-\gamma^{(\infty,p)}_{XQ}\right\Vert _{1}%
=p_{\operatorname{err}}(\rho_{XB}). \label{eq:conv-dist-to-err-prob-CPTP}%
\end{equation}

\end{lemma}

\begin{proof}
The inequality $\leq$ follows by the same reasoning given at the beginning of the proof of the previous theorem. The opposite inequality follows because we can apply a completely dephasing channel to the $Q$ system and the state $\gamma^{(\infty,p)}_{XQ}$ remains invariant, while the channel $\mathcal{N}$ is transformed to a measurement channel. The trace distance does not increase under such a channel and evaluating it leads to an expression for the error probability under a particular measurement.
\end{proof}

\section{Derivation of the SDPs for scaled trace distance $D'$ and minimum conversion error in Propositions~\ref{prop:SDP-div-err} and~\ref{prop:SDP-conv-err}}

\label{app:dual-derivation}
\begin{proof}[Proof of Proposition~\ref{prop:SDP-div-err}]
We begin by rewriting the scaled trace distance $D^{\prime}(\rho_{X A}, \sigma_{X A})$ as follows:
\begin{align}
D^{\prime}\left(\rho_{X A}, \sigma_{X A}\right) & = \frac{ \frac{1}{2}\|\rho_{XA}-\sigma_{XA}\|_1}{\frac{1}{2}\left(1-\left\|q \sigma_{0}-(1-q) \sigma_{1}\right\|_{1}\right)}\\
& = \frac{ \|\rho_{XA}-\sigma_{XA}\|_1}{1-\left\|q \sigma_{0}-(1-q) \sigma_{1}\right\|_{1}}\\
& = \frac{\max_{-I_{XA}\le L_{XA}\le I_{XA}}\tr L_{XA}(\rho_{XA}-\sigma_{XA})}{1-\max_{-I_A\le P_A\le I_A}\tr P_A(q\sigma_{0}-(1-q)\sigma_{1})} \\
& = \max_{-I_{XA} \le L_{XA}\leq I_{XA},-I_A \leq P_A \le I_A}\frac{\tr L_{XA}(\rho_{XA}-\sigma_{XA})}{1-\tr P_A(q\sigma_{0}-(1-q)\sigma_{1})}.
\end{align}
Furthermore, let us introduce $t=\frac{\tr L_{XA}(\rho_{XA}-\sigma_{XA})}{1-\tr P_A(q\sigma_{0}-(1-q)\sigma_{1})}$ and obtain
\begin{multline}
D^{\prime}\left(\rho_{X A},  \sigma_{X A}\right)  =
\{\max t : t-t\tr P_A(q\sigma_{0}-(1-q)\sigma_{1})=\tr L_{XA}(\rho_{XA}-\sigma_{XA}),\\
-I_{XA}\le L_{XA} \leq I_{XA},-I_A\leq P_A\le I_A\}.
\end{multline}
The constraints in the optimization above still have bilinear conditions. However, we can absorb $tP_A$ into a single variable and obtain the simplified SDP in \eqref{eq:SDP-primal-div-err}.

We now continue with the derivation of the dual SDP stated in \eqref{eq:SDP-dual-div-err}.
The standard form of primal and dual SDPs is as follows \cite{Wat18}:
\begin{align}
& \sup_{X\geq0}\left\{  \operatorname{Tr}[AX]:\Phi(X)\leq B\right\}  ,\\
& \inf_{Y\geq0}\left\{  \operatorname{Tr}[BY]:\Phi^{\dag}(Y)\geq A\right\}  .
\end{align}
The SDP for the scaled trace distance can be written as%
\begin{equation}
\sup_{t,L_{XA}^{0},L_{XA}^{1},P_{A}^{0},P_{A}^{1}\geq0}\left\{
\begin{array}
[c]{c}%
t:-I_{XA}\leq L_{XA}^{1}-L_{XA}^{0}\leq I_{XA},\\
-tI_{A}\leq P_{A}^{1}-P_{A}^{0}\leq tI_{A},\\
t-\operatorname{Tr}[\left(  P_{A}^{1}-P_{A}^{0}\right)  \left(  q\sigma
_{0}-\left(  1-q\right)  \sigma_{1}\right)  ]=\operatorname{Tr}[\left(
L_{XA}^{1}-L_{XA}^{0}\right)  \left(  \rho_{XA}-\sigma_{XA}\right)  ]
\end{array}
\right\}  .
\end{equation}
In standard form, this SDP\ is as follows:%
\begin{align}
X  & =\text{diag}(t,L_{XA}^{0},L_{XA}^{1},P_{A}^{0},P_{A}^{1}),\\
A  & =\text{diag}(1,0,0,0,0),\\
\Phi(X)  & =\text{diag}(L_{XA}^{1}-L_{XA}^{0},L_{XA}^{0}-L_{XA}^{1},P_{A}%
^{1}-P_{A}^{0}-tI_{A},P_{A}^{0}-P_{A}^{1}-tI_{A},\notag \\
& \qquad\operatorname{Tr}[\left(  L_{XA}^{1}-L_{XA}^{0}\right)  \left(
\rho_{XA}-\sigma_{XA}\right)  ]-\left(  t-\operatorname{Tr}[\left(  P_{A}%
^{1}-P_{A}^{0}\right)  \left(  q\sigma_{0}-\left(  1-q\right)  \sigma
_{1}\right)  ]\right)  ,\notag \\
& \qquad t-\operatorname{Tr}[\left(  P_{A}^{1}-P_{A}^{0}\right)  \left(
q\sigma_{0}-\left(  1-q\right)  \sigma_{1}\right)  ]-\operatorname{Tr}[\left(
L_{XA}^{1}-L_{XA}^{0}\right)  \left(  \rho_{XA}-\sigma_{XA}\right)  ]),\\
B  & =\text{diag}(I_{XA},I_{XA},0,0,0,0).
\end{align}
So we need to derive the adjoint map $\Phi^{\dag}$, satisfying
$\operatorname{Tr}[Y\Phi(X)]=\operatorname{Tr}[\Phi^{\dag}(Y)X]$. Consider
that the dual variables are given as%
\begin{equation}
Y=\text{diag}(B_{XA},C_{XA},D_{A},E_{A},s_{1},s_{2}).
\end{equation}
Then we find that%
\begin{align}
& \operatorname{Tr}[Y\Phi(X)]\nonumber\\
& =\operatorname{Tr}[B_{XA}\left(  L_{XA}^{1}-L_{XA}^{0}\right)
]+\operatorname{Tr}[C_{XA}\left(  L_{XA}^{0}-L_{XA}^{1}\right)  ]\nonumber\\
& \qquad+\operatorname{Tr}[D_{A}\left(  P_{A}^{1}-P_{A}^{0}-tI_{A}\right)
]+\operatorname{Tr}[E_{A}\left(  P_{A}^{0}-P_{A}^{1}-tI_{A}\right)
]\nonumber\\
& \qquad+\left(  s_{1}-s_{2}\right)  \left(  \operatorname{Tr}[\left(
L_{XA}^{1}-L_{XA}^{0}\right)  \left(  \rho_{XA}-\sigma_{XA}\right)  ]-\left(
t-\operatorname{Tr}[\left(  P_{A}^{1}-P_{A}^{0}\right)  \left(  q\sigma
_{0}-\left(  1-q\right)  \sigma_{1}\right)  ]\right)  \right)  \\
& =\operatorname{Tr}[\left(  B_{XA}-C_{XA}+\left(  s_{1}-s_{2}\right)  \left(
\rho_{XA}-\sigma_{XA}\right)  \right)  L_{XA}^{1}]\nonumber\\
& \qquad-\operatorname{Tr}[\left(  B_{XA}-C_{XA}+\left(  s_{1}-s_{2}\right)
\left(  \rho_{XA}-\sigma_{XA}\right)  \right)  L_{XA}^{0}]\nonumber\\
& \qquad+\operatorname{Tr}[\left(  D_{A}-E_{A}+\left(  s_{1}-s_{2}\right)
\left(  q\sigma_{0}-\left(  1-q\right)  \sigma_{1}\right)  \right)  P_{A}%
^{1}]\nonumber\\
& \qquad-\operatorname{Tr}[\left(  D_{A}-E_{A}+\left(  s_{1}-s_{2}\right)
\left(  q\sigma_{0}-\left(  1-q\right)  \sigma_{1}\right)  \right)  P_{A}%
^{0}]\nonumber\\
& \qquad-\left(  \left(  s_{1}-s_{2}\right)  +\operatorname{Tr}[D_{A}%
+E_{A}]\right)  t.
\end{align}
So we conclude that%
\begin{align}
\Phi^{\dag}(Y)  & =\text{diag}(-\left(  \left(  s_{1}-s_{2}\right)
+\operatorname{Tr}[D_{A}+E_{A}]\right)  ,-\left(  B_{XA}-C_{XA}+\left(
s_{1}-s_{2}\right)  \left(  \rho_{XA}-\sigma_{XA}\right)  \right)
,\nonumber\\
& \qquad\left(  B_{XA}-C_{XA}+\left(  s_{1}-s_{2}\right)  \left(  \rho
_{XA}-\sigma_{XA}\right)  \right)  ,\nonumber\\
& \qquad-\left(  D_{A}-E_{A}+\left(  s_{1}-s_{2}\right)  \left(  q\sigma
_{0}-\left(  1-q\right)  \sigma_{1}\right)  \right)  ,\nonumber\\
& \qquad\left(  D_{A}-E_{A}+\left(  s_{1}-s_{2}\right)  \left(  q\sigma
_{0}-\left(  1-q\right)  \sigma_{1}\right)  \right)  ).
\end{align}
So then $\Phi^{\dag}(Y)\geq A$ is equivalent to the following constraints:%
\begin{align}
-\left(  \left(  s_{1}-s_{2}\right)  +\operatorname{Tr}[D_{A}+E_{A}]\right)
& \geq1,\\
-\left(  B_{XA}-C_{XA}+\left(  s_{1}-s_{2}\right)  \left(  \rho_{XA}%
-\sigma_{XA}\right)  \right)    & \geq0,\\
\left(  B_{XA}-C_{XA}+\left(  s_{1}-s_{2}\right)  \left(  \rho_{XA}%
-\sigma_{XA}\right)  \right)    & \geq0,\\
-\left(  D_{A}-E_{A}+\left(  s_{1}-s_{2}\right)  \left(  q\sigma_{0}-\left(
1-q\right)  \sigma_{1}\right)  \right)    & \geq0,\\
\left(  D_{A}-E_{A}+\left(  s_{1}-s_{2}\right)  \left(  q\sigma_{0}-\left(
1-q\right)  \sigma_{1}\right)  \right)    & \geq0,
\end{align}
which is the same as%
\begin{align}
-\left(  s_{1}-s_{2}\right)  -1  & \geq\operatorname{Tr}[D_{A}+E_{A}],\\
B_{XA}-C_{XA}  & =-\left(  s_{1}-s_{2}\right)  \left(  \rho_{XA}-\sigma
_{XA}\right)  ,\\
D_{A}-E_{A}  & =-\left(  s_{1}-s_{2}\right)  \left(  q\sigma_{0}-\left(
1-q\right)  \sigma_{1}\right)  .
\end{align}
Then the dual SDP\ is as follows:%
\begin{equation}
\inf_{B_{XA},C_{XA},D_{A},E_{A},s_{1},s_{2}\geq0}\left\{
\begin{array}
[c]{c}%
\operatorname{Tr}[B_{XA}+C_{XA}]:\\
-\left(  s_{1}-s_{2}\right)  -1\geq\operatorname{Tr}[D_{A}+E_{A}],\\
B_{XA}-C_{XA}=-\left(  s_{1}-s_{2}\right)  \left(  \rho_{XA}-\sigma
_{XA}\right)  ,\\
D_{A}-E_{A}=-\left(  s_{1}-s_{2}\right)  \left(  q\sigma_{0}-\left(
1-q\right)  \sigma_{1}\right)
\end{array}
\right\}  .
\end{equation}
Now setting $s=-\left(  s_{1}-s_{2}\right)  $, this can be rewritten as%
\begin{equation}
\inf_{\substack{B_{XA},C_{XA},D_{A},E_{A}\geq0,\\s\in\mathbb{R}}}\left\{
\begin{array}
[c]{c}%
\operatorname{Tr}[B_{XA}+C_{XA}]:\\
s-1\geq\operatorname{Tr}[D_{A}+E_{A}],\\
B_{XA}-C_{XA}=s\left(  \rho_{XA}-\sigma_{XA}\right)  ,\\
D_{A}-E_{A}=s\left(  q\sigma_{0}-\left(  1-q\right)  \sigma_{1}\right)
\end{array}
\right\}  .
\end{equation}

We now prove that strong duality holds under the following conditions:
\begin{align}
\left\Vert \rho_{XA}-\sigma_{XA}\right\Vert _{1}  &
>0,\label{eq:conditions-for-strong-duality-SDp-err-div}\\
\left\Vert q\sigma_{0}-\left(  1-q\right)  \sigma_{1}\right\Vert _{1}  &
<1.\label{eq:conditions-for-strong-duality-SDp-err-div-2}%
\end{align}
Then strong duality holds by picking the primal variables as%
\begin{align}
t  & =\frac{1}{2}\left\Vert \rho_{XA}-\sigma_{XA}\right\Vert _{1},\\
L_{XA}^{1}  & =\frac{1}{2}\Pi_{\rho_{XA}\geq\sigma_{XA}}+\frac{1}{2}I_{XA},\\
L_{XA}^{0}  & =\frac{1}{2}\Pi_{\rho_{XA}<\sigma_{XA}}+\frac{1}{2}I_{XA},\\
P_{A}^{1}  & =P_{A}^{0}=I_{A},
\end{align}
where $\Pi_{\rho_{XA}\geq\sigma_{XA}}$ is the projection onto the non-negative eigenspace of $\rho_{XA}\geq\sigma_{XA}$ and $\Pi_{\rho_{XA}<\sigma_{XA}}$ is the projection onto the strictly negative eigenspace.
Then all of the following constraints are satisfied with strict inequality
(except for the final equality):
\begin{align}
t,L_{XA}^{0},L_{XA}^{1},P_{A}^{0},P_{A}^{1}  & \geq0,\\
-I_{XA}  & \leq L_{XA}^{1}-L_{XA}^{0}\leq I_{XA},\\
-tI_{A}  & \leq P_{A}^{1}-P_{A}^{0}\leq tI_{A},\\
t-\operatorname{Tr}[\left(  P_{A}^{1}-P_{A}^{0}\right)  \left(  q\sigma
_{0}-\left(  1-q\right)  \sigma_{1}\right)  ]  & =\operatorname{Tr}[\left(
L_{XA}^{1}-L_{XA}^{0}\right)  \left(  \rho_{XA}-\sigma_{XA}\right)  ].
\end{align}

For the dual program, we pick%
\begin{align}
s  & =\frac{1}{1-\left\Vert q\sigma_{0}-\left(  1-q\right)  \sigma
_{1}\right\Vert _{1}},\\
B_{XA}  & =I_{XA}+s\rho_{XA},\\
C_{XA}  & =I_{XA}+s\sigma_{XA},\\
D_{A}  & =sP,\\
E_{A}  & =sN,
\end{align}
where $P$ is the positive part of $q\sigma_{0}-\left(  1-q\right)  \sigma_{1}$
and $N$ is the negative part of $q\sigma_{0}-\left(  1-q\right)  \sigma_{1}$.
Under these choices, we find that the constraints from the dual program are
met, i.e., as follows:%
\begin{align}
B_{XA},C_{XA},D_{A},E_{A}  & \geq0,s\in\mathbb{R},\\
s-1  & =\operatorname{Tr}[D_{A}+E_{A}],\\
B_{XA}-C_{XA}  & =s\left(  \rho_{XA}-\sigma_{XA}\right)  ,\\
D_{A}-E_{A}  & =s\left(  q\sigma_{0}-\left(  1-q\right)  \sigma_{1}\right)  .
\end{align}
Thus, strong duality holds under the conditions given in \eqref{eq:conditions-for-strong-duality-SDp-err-div}--\eqref{eq:conditions-for-strong-duality-SDp-err-div-2}.

Finally, suppose now that $\left\Vert \rho_{XA} - \sigma_{XA}\right\Vert_1 = 0 $ and $\left\Vert q\sigma_{0}-\left(  1-q\right)  \sigma_{1}\right\Vert _{1}  
<1$. Then the choices $t=0$, $L_{XA} = 0$, and $P_A = 0$ are feasible for the primal and lead to a value of zero for the objective function. Also, setting $B_{XA}$ to be the positive part of  $s( \rho_{XA} - \sigma_{XA})$ and $C_{XA}$ to be the negative part of $s( \rho_{XA} - \sigma_{XA})$, with the same choices for $s$, $D_A$, and $E_A$ as given above, leads to feasible choices for the dual, for which the objective function also evaluates to zero. So strong duality holds in this case also.
\end{proof}

\medskip

\begin{proof} [Proof of Proposition~\ref{prop:SDP-conv-err}]
Using \eqref{eq:SDP-dual-div-err}, the scaled trace distance $D'(\cdot, \cdot)$ for states $\tau_{XA} \coloneqq t|0\rangle\!\langle0|\otimes
\tau_{0}+\left(  1-t\right)  |1\rangle\!\langle1|\otimes\tau_{1}$ and
$\omega_{XA} \coloneqq w|0\rangle\!\langle0|\otimes\omega_{0}+\left(  1-w\right)
|1\rangle\!\langle1|\otimes\omega_{1}$ can be written as the following
semi-definite program%
\begin{equation}
D^{\prime}(\tau_{XA},\omega_{XA})=\min_{B_{XA},C_{XA},D_{A},E_{A}%
\geq0,s\in\mathbb{R}}\left\{
\begin{array}
[c]{c}%
\operatorname{Tr}[B_{XA}+C_{XA}]:\\
B_{XA}-C_{XA}=s\left(  \tau_{XA}-\omega_{XA}\right)  ,\\
D_{A}-E_{A}=s(w\omega_{0}-\left(  1-w\right)  \omega_{1}),\\
\operatorname{Tr}[D_{A}+E_{A}]\leq s-1
\end{array}
\right\}  .
\label{eq-app:SDP-D-prime}
\end{equation}
Following the development in \eqref{eq:app:CDS-SDP-1}--\eqref{eq:app:CDS-SDP-last} and combining with \eqref{eq-app:SDP-D-prime}, we conclude the following form for the optimization task%
\begin{equation}
\min_{\substack{B_{XA^{\prime}},C_{XA^{\prime}},D_{A^{\prime}},E_{A^{\prime}%
}\geq0,\\\Gamma_{AA^{\prime}}^{\mathcal{N}^{0}},\Gamma_{AA^{\prime}%
}^{\mathcal{N}^{1}}\geq0,s\in\mathbb{R}}}\left\{
\begin{array}
[c]{c}%
\operatorname{Tr}[B_{XA^{\prime}}+C_{XA^{\prime}}]:\\
B_{XA^{\prime}}-C_{XA^{\prime}}=s\left(  \tau_{XA^{\prime}}-\sigma
_{XA^{\prime}}\right)  ,\\
D_{A^{\prime}}-E_{A^{\prime}}=s(q\sigma_{0}-\left(  1-q\right)  \sigma_{1}),\\
\operatorname{Tr}[D_{A^{\prime}}+E_{A^{\prime}}]\leq s-1,\\
\operatorname{Tr}_{A^{\prime}}[\Gamma_{AA^{\prime}}^{\mathcal{N}^{0}}%
+\Gamma_{AA^{\prime}}^{\mathcal{N}^{1}}]=I_{A},\\
\langle0|_{X}\tau_{XA^{\prime}}|0\rangle_{X}=\operatorname{Tr}_{A}[(p\rho
_{A}^{0})^{T}\Gamma_{AA^{\prime}}^{\mathcal{N}^{0}}]+\operatorname{Tr}%
_{A}[(\left(  1-p\right)  \rho_{A}^{1})^{T}\Gamma_{AA^{\prime}}^{\mathcal{N}%
^{1}}],\\
\langle1|_{X}\tau_{XA^{\prime}}|1\rangle_{X}=\operatorname{Tr}_{A}[(p\rho
_{A}^{0})^{T}\Gamma_{AA^{\prime}}^{\mathcal{N}^{1}}]+\operatorname{Tr}%
_{A}[(\left(  1-p\right)  \rho_{A}^{1})^{T}\Gamma_{AA^{\prime}}^{\mathcal{N}%
^{0}}],\\
\langle1|_{X}\tau_{XA^{\prime}}|0\rangle_{X}=\langle0|_{X}\tau_{XA^{\prime}%
}|1\rangle_{X}=0
\end{array}
\right\}  .
\end{equation}
As written, this is not a semi-definite program, due to the bilinear term
$s\tau_{XA^{\prime}}$ in the second line above, given that $s$ is an
optimization variable and $\tau_{XA^{\prime}}$ includes the optimization
variables $\Gamma_{AA^{\prime}}^{\mathcal{N}^{0}}$ and $\Gamma_{AA^{\prime}%
}^{\mathcal{N}^{1}}$. However, we observe that $s\geq1$, due to the
constraints $\operatorname{Tr}[D_{A^{\prime}}+E_{A^{\prime}}]\leq s-1$ and
$D_{A^{\prime}},E_{A^{\prime}}\geq0$. We can then make the reassignments
$s\Gamma_{AA^{\prime}}^{\mathcal{N}^{0}}\rightarrow\Omega_{AA^{\prime}}^{0}$
and $s\Gamma_{AA^{\prime}}^{\mathcal{N}^{1}}\rightarrow\Omega_{AA^{\prime}%
}^{1}$ to rewrite the above optimization as follows:%
\begin{equation}
\min_{\substack{B_{XA^{\prime}},C_{XA^{\prime}},D_{A^{\prime}},E_{A^{\prime}%
}\geq0,\\\Omega_{AA^{\prime}}^{0},\Omega_{AA^{\prime}}^{1}\geq0,s\geq
1}}\left\{
\begin{array}
[c]{c}%
\operatorname{Tr}[B_{XA^{\prime}}+C_{XA^{\prime}}]:\\
B_{XA^{\prime}}-C_{XA^{\prime}}=\tau_{XA^{\prime}}-s\sigma_{XA^{\prime}},\\
D_{A^{\prime}}-E_{A^{\prime}}=s(q\sigma_{0}-\left(  1-q\right)  \sigma_{1}),\\
\operatorname{Tr}[D_{A^{\prime}}+E_{A^{\prime}}]\leq s-1,\\
\operatorname{Tr}_{A^{\prime}}[\Omega_{AA^{\prime}}^{0}+\Omega_{AA^{\prime}%
}^{1}]=sI_{A},\\
\langle0|_{X}\tau_{XA^{\prime}}|0\rangle_{X}=\operatorname{Tr}_{A}[(p\rho
_{A}^{0})^{T}\Omega_{AA^{\prime}}^{0}]+\operatorname{Tr}_{A}[(\left(
1-p\right)  \rho_{A}^{1})^{T}\Omega_{AA^{\prime}}^{1}],\\
\langle1|_{X}\tau_{XA^{\prime}}|1\rangle_{X}=\operatorname{Tr}_{A}[(p\rho
_{A}^{0})^{T}\Omega_{AA^{\prime}}^{1}]+\operatorname{Tr}_{A}[(\left(
1-p\right)  \rho_{A}^{1})^{T}\Omega_{AA^{\prime}}^{0}],\\
\langle1|_{X}\tau_{XA^{\prime}}|0\rangle_{X}=\langle0|_{X}\tau_{XA^{\prime}%
}|1\rangle_{X}=0
\end{array}
\right\}  .
\end{equation}
This concludes the proof.
\end{proof}

\section{Proof of Proposition~\ref{prop:alt-SDP-one-shot-distillable-SD} --- SDP for approximate one-shot distillable-SD under CPTP$_A$ maps}

\label{app:proof-one-shot-distillable-SD-SDP}

Let us make the substitution $M\rightarrow1/r$, setting%
\begin{align}
\pi_{1/r} &   \coloneqq \left(  1-\frac{r}{2}\right)  |0\rangle\!\langle0|+\frac{r}%
{2}|1\rangle\!\langle1|,\label{eq:approx-distill-SDP-begin-1}\\
\gamma_{XQ}^{(1/r,q)} &   \coloneqq q|0\rangle\!\langle0|_{X}\otimes\pi_{r}+\left(
1-q\right)  |1\rangle\!\langle1|_{X}\otimes\sigma^{(1)}\pi_{r}\sigma^{(1)},\\
q &  =p,\label{eq:approx-distill-SDP-begin-3}%
\end{align}
and employing Proposition~\ref{prop:SDP-conv-err}, we find that the optimization in
\eqref{eq:start-point-SDP-distill-CPTP-A} is equal to the negative logarithm
of the following:%
\begin{multline}
\inf_{\mathcal{A}\in\text{CPTP}_{A}}\left\{  r\ |\ D^{\prime}(\mathcal{A}%
(\rho_{XA}),\gamma_{XQ}^{(1/r,p)})\leq\varepsilon\right\}
\label{eq:approx-distill-SDP-begin-other}\\
=\inf_{\substack{B_{XQ},C_{XQ},D_{Q},\\E_{Q},\Gamma_{AQ}^{\mathcal{N}}%
\geq0,s\in\mathbb{R},r\in\left[  0,1\right]  }}\left\{
\begin{array}
[c]{c}%
r:\\
\operatorname{Tr}[B_{XQ}+C_{XQ}]\leq\varepsilon,\\
B_{XQ}-C_{XQ}=s\left(  \operatorname{Tr}_{A}[\rho_{XA}^{T_{A}}\Gamma
_{AQ}^{\mathcal{N}}]-\gamma_{XQ}^{(1/r,p)}\right)  ,\\
D_{Q}-E_{Q}=s(p\pi_{1/r}-\left(  1-p\right)  \sigma^{(1)}\pi_{1/r}\sigma^{(1)}),\\
\operatorname{Tr}[D_{Q}+E_{Q}]\leq s-1,\\
\operatorname{Tr}_{Q}[\Gamma_{AQ}^{\mathcal{N}}]=I_{A}%
\end{array}
\right\}  .
\end{multline}
As written, this is not an SDP. However, through the substitutions
$B_{XQ}\rightarrow sB_{XQ}$, $C_{XQ}\rightarrow sC_{XQ}$, $D_{Q}\rightarrow
sD_{Q}$, and $E_{Q}\rightarrow sE_{Q}$ and observing that $s\geq1$, we arrive
at the following SDP:%
\begin{equation}
\inf_{\substack{B_{XQ},C_{XQ},D_{Q},\\E_{Q},\Gamma_{AQ}^{\mathcal{N}}%
\geq0,s\geq1,r\in\left[  0,1\right]  }}\left\{
\begin{array}
[c]{c}%
r:\\
\operatorname{Tr}[B_{XQ}+C_{XQ}]\leq\frac{\varepsilon}{s},\\
B_{XQ}-C_{XQ}=\operatorname{Tr}_{A}[\rho_{XA}^{T_{A}}\Gamma_{AQ}^{\mathcal{N}%
}]-\gamma_{XQ}^{(1/r,p)},\\
D_{Q}-E_{Q}=p\pi_{1/r}-\left(  1-p\right)  \sigma^{(1)}\pi_{1/r}\sigma^{(1)},\\
\operatorname{Tr}[D_{Q}+E_{Q}]\leq1-\frac{1}{s},\\
\operatorname{Tr}_{Q}[\Gamma_{AQ}^{\mathcal{N}}]=I_{A}%
\end{array}
\right\}  .
\end{equation}
We then make a final substitution of $s\rightarrow\frac{1}{t}$ to arrive at%
\begin{equation}
\inf_{\substack{B_{XQ},C_{XQ},D_{Q},\\E_{Q},\Gamma_{AQ}^{\mathcal{N}}%
\geq0,t\geq0,r\in\left[  0,1\right]  }}\left\{
\begin{array}
[c]{c}%
r:\\
\operatorname{Tr}[B_{XQ}+C_{XQ}]\leq\varepsilon t,\\
B_{XQ}-C_{XQ}=\operatorname{Tr}_{A}[\rho_{XA}^{T_{A}}\Gamma_{AQ}^{\mathcal{N}%
}]-\gamma_{XQ}^{(1/r,p)},\\
D_{Q}-E_{Q}=p\pi_{1/r}-\left(  1-p\right)  \sigma^{(1)}\pi_{1/r}\sigma^{(1)},\\
\operatorname{Tr}[D_{Q}+E_{Q}]\leq1-t,\\
t\leq1,\\
\operatorname{Tr}_{Q}[\Gamma_{AQ}^{\mathcal{N}}]=I_{A}%
\end{array}
\right\}  .
\end{equation}
We can actually drop the constraint $t\leq1$ and the variable $t$ itself
because it is redundant, leaving us with%
\begin{equation}
\inf_{\substack{B_{XQ},C_{XQ},D_{Q},\\E_{Q},\Gamma_{AQ}^{\mathcal{N}}%
\geq0,r\in\left[  0,1\right]  }}\left\{
\begin{array}
[c]{c}%
r:\\
\operatorname{Tr}[B_{XQ}+C_{XQ}]\leq\varepsilon\left(  1-\operatorname{Tr}%
[D_{A}+E_{A}]\right)  ,\\
B_{XQ}-C_{XQ}=\operatorname{Tr}_{A}[\rho_{XA}^{T_{A}}\Gamma_{AQ}^{\mathcal{N}%
}]-\gamma_{XQ}^{(1/r,p)},\\
D_{Q}-E_{Q}=p\pi_{1/r}-\left(  1-p\right)  \sigma^{(1)}\pi_{1/r}\sigma^{(1)},\\
\operatorname{Tr}_{Q}[\Gamma_{AQ}^{\mathcal{N}}]=I_{A}%
\end{array}
\right\}  .
\end{equation}
It suffices to take the channel $\mathcal{N}$ to be a measurement channel of
the following form:%
\begin{equation}
\mathcal{N}(\omega)=\operatorname{Tr}[\Lambda\omega]|0\rangle\!\langle
0|+\operatorname{Tr}[\left(  I-\Lambda\right)  \omega]|1\rangle\!\langle1|.
\end{equation}
Then consider that%
\begin{align}
\operatorname{Tr}_{A}[\rho_{XA}^{T_{A}}\Gamma_{AQ}^{\mathcal{N}}] &
=\operatorname{Tr}_{A}[\Lambda_{A}\rho_{XA}]|0\rangle\!\langle
0|+\operatorname{Tr}_{A}[\left(  I_{A}-\Lambda_{A}\right)  \rho_{XA}%
]|1\rangle\!\langle1|\\
&  =|0\rangle\!\langle0|_{X}\otimes p\operatorname{Tr}[\Lambda\rho_{0}%
]|0\rangle\!\langle0|_{Q}+|1\rangle\!\langle1|_{X}\otimes\left(  1-p\right)
\operatorname{Tr}[\Lambda\rho_{1}]|0\rangle\!\langle0|_{Q}\nonumber\\
&  \quad+|0\rangle\!\langle0|_{X}\otimes p\operatorname{Tr}[\left(
I-\Lambda\right)  \rho_{0}]|1\rangle\!\langle1|_{Q}+|1\rangle\!\langle
1|_{X}\otimes\left(  1-p\right)  \operatorname{Tr}[\left(  I-\Lambda\right)
\rho_{1}]|1\rangle\!\langle1|_{Q}.
\end{align}
Then given that%
\begin{multline}
\gamma_{XQ}^{(1/r,p)}=p\left(  1-\frac{r}{2}\right)  |0\rangle\!\langle
0|_{X}\otimes|0\rangle\!\langle0|_{Q}+p\frac{r}{2}|0\rangle\!\langle0|_{X}%
\otimes|1\rangle\!\langle1|_{Q}\\
+\left(  1-p\right)  \frac{r}{2}|1\rangle\!\langle1|_{X}\otimes|0\rangle
\langle0|_{Q}+\left(  1-p\right)  \left(  1-\frac{r}{2}\right)  |1\rangle
\langle1|_{X}\otimes|1\rangle\!\langle1|_{Q}.
\end{multline}
we find that%
\begin{align}
  \operatorname{Tr}_{A}[\rho_{XA}^{T_{A}}\Gamma_{AQ}^{\mathcal{N}}%
]-\gamma_{XQ}^{(1/r,p)}
&  =p\left(  \operatorname{Tr}[\Lambda\rho_{0}]-\left(  1-\frac{r}{2}\right)
\right)  |0\rangle\!\langle0|_{X}\otimes|0\rangle\!\langle0|_{Q}\nonumber\\
&  \qquad+p\left(  \operatorname{Tr}[\left(  I-\Lambda\right)  \rho_{0}%
]-\frac{r}{2}\right)  |0\rangle\!\langle0|_{X}\otimes|1\rangle\!\langle
1|_{Q}\nonumber\\
&  \qquad+\left(  1-p\right)  \left(  \operatorname{Tr}[\Lambda\rho_{1}%
]-\frac{r}{2}\right)  |1\rangle\!\langle1|_{X}\otimes|0\rangle\!\langle
0|_{Q}\nonumber\\
&  \qquad+\left(  1-p\right)  \left(  \operatorname{Tr}[\left(  I-\Lambda
\right)  \rho_{1}]-\left(  1-\frac{r}{2}\right)  \right)  |1\rangle
\langle1|_{X}\otimes|1\rangle\!\langle1|_{Q}\\
&  =p\left(  \operatorname{Tr}[\Lambda\rho_{0}]-\left(  1-\frac{r}{2}\right)
\right)  |0\rangle\!\langle0|_{X}\otimes|0\rangle\!\langle0|_{Q}\nonumber\\
&  \qquad-p\left(  \operatorname{Tr}[\Lambda\rho_{0}]-\left(  1-\frac{r}%
{2}\right)  \right)  |0\rangle\!\langle0|_{X}\otimes|1\rangle\!\langle
1|_{Q}\nonumber\\
&  \qquad+\left(  1-p\right)  \left(  \operatorname{Tr}[\Lambda\rho_{1}%
]-\frac{r}{2}\right)  |1\rangle\!\langle1|_{X}\otimes|0\rangle\!\langle
0|_{Q}\nonumber\\
&  \qquad-\left(  1-p\right)  \left(  \operatorname{Tr}[\Lambda\rho_{1}%
]-\frac{r}{2}\right)  |1\rangle\!\langle1|_{X}\otimes|1\rangle\!\langle1|_{Q}.
\end{align}
It suffices to take $B_{XQ}$ and $C_{XQ}$ to have the following form:%
\begin{equation}
B_{XQ}=\sum_{i,j\in\left\{  0,1\right\}  }b^{i,j}|i\rangle\!\langle
i|_{X}\otimes|j\rangle\!\langle j|_{Q},\qquad C_{XQ}=\sum_{i,j\in\left\{
0,1\right\}  }c^{i,j}|i\rangle\!\langle i|_{X}\otimes|j\rangle\!\langle j|_{Q},
\end{equation}
with $b^{i,j},c^{i,j}\geq0$. Also, it suffices to take%
\begin{equation}
D_{Q}=\sum_{i\in\left\{  0,1\right\}  }d^{i}|i\rangle\!\langle i|_{Q},\qquad
E_{Q}=\sum_{i\in\left\{  0,1\right\}  }e^{i}|i\rangle\!\langle i|_{Q},
\end{equation}
with $d^{i},e^{i}\geq0$. We also have that%
\begin{align}
&  p\pi_{1/r}-\left(  1-p\right)  \sigma^{(1)}\pi_{1/r}\sigma^{(1)}\nonumber\\
&  =p\left(  1-\frac{r}{2}\right)  |0\rangle\!\langle0|+p\frac{r}{2}%
|1\rangle\!\langle1|-\left(  1-p\right)  \left(  \left(  1-\frac{r}{2}\right)
|1\rangle\!\langle1|+\frac{r}{2}|0\rangle\!\langle0|\right)  \\
&  =\left(  p-\frac{r}{2}\right)  |0\rangle\!\langle0|+\left(  \frac{r}%
{2}-\left(  1-p\right)  \right)  |1\rangle\!\langle1|.
\end{align}
Then the SDP\ simplifies as follows:%
\begin{equation}
\inf_{\substack{b^{i,j},c^{i,j}\geq0,\\d^{i},e^{i}\geq0,\Lambda\geq
0,r\in\left[  0,1\right]  }}\left\{
\begin{array}
[c]{c}%
r:\\
\sum_{i,j\in\left\{  0,1\right\}  }b^{i,j}+c^{i,j}\leq\varepsilon\left(
1-\left(  d^{0}+d^{1}+e^{0}+e^{1}\right)  \right)  ,\\
b^{0,0}-c^{0,0}=p\left(  \operatorname{Tr}[\Lambda\rho_{0}]-\left(  1-\frac
{r}{2}\right)  \right)  ,\\
b^{0,1}-c^{0,1}=-p\left(  \operatorname{Tr}[\Lambda\rho_{0}]-\left(
1-\frac{r}{2}\right)  \right)  ,\\
b^{1,0}-c^{1,0}=\left(  1-p\right)  \left(  \operatorname{Tr}[\Lambda\rho
_{1}]-\frac{r}{2}\right)  ,\\
b^{1,1}-c^{1,1}=-\left(  1-p\right)  \left(  \operatorname{Tr}[\Lambda\rho
_{1}]-\frac{r}{2}\right)  ,\\
d^{0}-e^{0}=p-\frac{r}{2},\\
d^{1}-e^{1}=\frac{r}{2}-\left(  1-p\right)  ,\\
\Lambda\leq I
\end{array}
\right\}
\end{equation}
We can eliminate the $b^{i,j}$ variables and the $d^{i}$ variables, and this
simplifies as follows:%
\begin{equation}
\inf_{\substack{c^{i,j}\geq0,\\e^{i}\geq0,\Lambda\geq0,r\in\left[  0,1\right]
}}\left\{
\begin{array}
[c]{c}%
r:\\
2\sum_{i,j\in\left\{  0,1\right\}  }c^{i,j}\leq\varepsilon\left(  1-p-\left(
e^{0}+e^{1}\right)  \right)  ,\\
c^{0,0}\geq-p\left(  \operatorname{Tr}[\Lambda\rho_{0}]-\left(  1-\frac{r}%
{2}\right)  \right)  ,\\
c^{0,1}\geq p\left(  \operatorname{Tr}[\Lambda\rho_{0}]-\left(  1-\frac{r}%
{2}\right)  \right)  ,\\
c^{1,0}\geq-\left(  1-p\right)  \left(  \operatorname{Tr}[\Lambda\rho
_{1}]-\frac{r}{2}\right)  ,\\
c^{1,1}\geq\left(  1-p\right)  \left(  \operatorname{Tr}[\Lambda\rho
_{1}]-\frac{r}{2}\right)  ,\\
e^{0}\geq\frac{r}{2}-p,\\
e^{1}\geq\left(  1-p\right)  -\frac{r}{2},\\
\Lambda\leq I
\end{array}
\right\}  \label{eq-approx-distill-SDP-correct}%
\end{equation}
So then the one-shot approximate distillable-SD\ is equal to%
\begin{equation}
\log \left(  \frac{1}{r^{\ast}}\right)  ,
\end{equation}
where $r^{\ast}$ is an optimal solution for \eqref{eq-approx-distill-SDP-correct}.

\section{Proof of Proposition~\ref{prop:one-shot-distillable-SD-CDS-SDP} --- SDP for approximate one-shot distillable-SD under CDS maps}

\label{app:approx-one-shot-distill-CDS}

We begin with the following SDP, which follows by applying similar reasoning
given around
\eqref{eq:approx-distill-SDP-begin-1}--\eqref{eq:approx-distill-SDP-begin-3}
and \eqref{eq:approx-distill-SDP-begin-other}:%
\begin{equation}
\inf_{\substack{B_{XQ},C_{XQ},D_{Q},\\E_{Q},\Gamma_{AQ}^{\mathcal{N}^{0}%
},\Gamma_{AQ}^{\mathcal{N}^{1}}\geq0}}\left\{
\begin{array}
[c]{c}%
r:\\
\operatorname{Tr}[B_{XQ}+C_{XQ}]\leq\varepsilon\left(  1-\operatorname{Tr}%
[D_{Q}+E_{Q}]\right)  ,\\
B_{XQ}-C_{XQ}=\tau_{XQ}-\gamma_{XQ}^{(1/r,1/2)},\\
D_{Q}-E_{Q}=\frac{1}{2}(\pi_{1/r}-\sigma^{(1)}\pi_{1/r}\sigma^{(1)}),\\
\operatorname{Tr}_{Q}[\Gamma_{AQ}^{\mathcal{N}^{0}}+\Gamma_{AQ}^{\mathcal{N}%
^{1}}]=I_{A},\\
\tau_{XQ}=|0\rangle\!\langle0|_{X}\otimes\left[  \operatorname{Tr}_{A}%
[(p\rho_{A}^{0})^{T}\Gamma_{AQ}^{\mathcal{N}^{0}}]+\operatorname{Tr}%
_{A}[(\left(  1-p\right)  \rho_{A}^{1})^{T}\Gamma_{AQ}^{\mathcal{N}^{1}%
}]\right]  \\
+|1\rangle\!\langle1|_{X}\otimes\left[  \operatorname{Tr}_{A}[(p\rho_{A}%
^{0})^{T}\Gamma_{AQ}^{\mathcal{N}^{1}}]+\operatorname{Tr}_{A}[(\left(
1-p\right)  \rho_{A}^{1})^{T}\Gamma_{AQ}^{\mathcal{N}^{0}}]\right]
\end{array}
\right\}  .
\end{equation}
Then by applying a completely dephasing channel to the $Q$ system, the state
$\gamma_{XQ}^{(1/r,1/2)}$ does not change, whereas the CDS\ channel becomes a
measurement channel of the following form:%
\begin{equation}
\mathcal{N}_{XA\rightarrow XQ} \coloneqq \operatorname{id}_{X}\otimes\mathcal{N}%
_{A\rightarrow Q}^{0}+\mathcal{F}_{X}\otimes\mathcal{N}_{A\rightarrow Q}^{1},
\end{equation}
where%
\begin{align}
\mathcal{N}_{A\rightarrow Q}^{0}(\omega) &   \coloneqq \operatorname{Tr}[\Lambda
_{A}^{0,0}\omega]|0\rangle\!\langle0|_{Q}+\operatorname{Tr}[\Lambda_{A}%
^{0,1}\omega]|1\rangle\!\langle1|_{Q},\\
\mathcal{N}_{A\rightarrow Q}^{1}(\omega) &   \coloneqq \operatorname{Tr}[\Lambda
_{A}^{1,0}\omega]|0\rangle\!\langle0|_{Q}+\operatorname{Tr}[\Lambda_{A}%
^{1,1}\omega]|1\rangle\!\langle1|_{Q},\\
I_{A} &  =\Lambda_{A}^{0,0}+\Lambda_{A}^{0,1}+\Lambda_{A}^{1,0}+\Lambda
_{A}^{1,1}.
\end{align}
The effect of the CDS channel $\mathcal{N}_{XA\rightarrow XQ}$ on the c-q
state $\rho_{XA}$ is as follows:%
\begin{align}
&  \mathcal{N}_{XA\rightarrow XQ}(\rho_{XA})\nonumber\\
&  =\mathcal{N}_{XA\rightarrow XQ}(|0\rangle\!\langle0|_{X}\otimes p\rho
_{0}+|1\rangle\!\langle1|_{X}\otimes\left(  1-p\right)  \rho_{1})\\
&  =|0\rangle\!\langle0|_{X}\otimes p\operatorname{Tr}[\Lambda_{A}^{0,0}\rho
_{0}]|0\rangle\!\langle0|_{Q}+|0\rangle\!\langle0|_{X}\otimes p\operatorname{Tr}%
[\Lambda_{A}^{0,1}\rho_{0}]|1\rangle\!\langle1|_{Q}\nonumber\\
&  \quad+|1\rangle\!\langle1|_{X}\otimes p\operatorname{Tr}[\Lambda_{A}%
^{1,0}\rho_{0}]|0\rangle\!\langle0|_{Q}+|1\rangle\!\langle1|_{X}\otimes
p\operatorname{Tr}[\Lambda_{A}^{1,1}\rho_{0}]|1\rangle\!\langle1|_{Q}\nonumber\\
&  \quad+|1\rangle\!\langle1|_{X}\otimes\left(  1-p\right)  \operatorname{Tr}%
[\Lambda_{A}^{0,0}\rho_{1}]|0\rangle\!\langle0|_{Q}+|1\rangle\!\langle
1|_{X}\otimes\left(  1-p\right)  \operatorname{Tr}[\Lambda_{A}^{0,1}\rho
_{1}]|1\rangle\!\langle1|_{Q}\nonumber\\
&  \quad+|0\rangle\!\langle0|_{X}\otimes\left(  1-p\right)  \operatorname{Tr}%
[\Lambda_{A}^{1,0}\rho_{1}]|0\rangle\!\langle0|_{Q}+|0\rangle\!\langle
0|_{X}\otimes\left(  1-p\right)  \operatorname{Tr}[\Lambda_{A}^{1,1}\rho
_{1}]|1\rangle\!\langle1|_{Q}\\
&  =\left(  p\operatorname{Tr}[\Lambda_{A}^{0,0}\rho_{0}]+\left(  1-p\right)
\operatorname{Tr}[\Lambda_{A}^{1,0}\rho_{1}]\right)  |0\rangle\!\langle
0|_{X}\otimes|0\rangle\!\langle0|_{Q}\nonumber\\
&  \quad+\left(  p\operatorname{Tr}[\Lambda_{A}^{0,1}\rho_{0}]+\left(
1-p\right)  \operatorname{Tr}[\Lambda_{A}^{1,1}\rho_{1}]\right)
|0\rangle\!\langle0|_{X}\otimes|1\rangle\!\langle1|_{Q}\nonumber\\
&  \quad+\left(  p\operatorname{Tr}[\Lambda_{A}^{1,0}\rho_{0}]+\left(
1-p\right)  \operatorname{Tr}[\Lambda_{A}^{0,0}\rho_{1}]\right)
|1\rangle\!\langle1|_{X}\otimes|0\rangle\!\langle0|_{Q}\nonumber\\
&  \quad+\left(  p\operatorname{Tr}[\Lambda_{A}^{1,1}\rho_{0}]+\left(
1-p\right)  \operatorname{Tr}[\Lambda_{A}^{0,1}\rho_{1}]\right)
|1\rangle\!\langle1|_{X}\otimes|1\rangle\!\langle1|_{Q}.
\end{align}
Thus, the state $\tau_{XQ}$ has the form%
\begin{align}
\tau_{XQ} &  =\left(  p\operatorname{Tr}[\Lambda_{A}^{0,0}\rho_{0}]+\left(
1-p\right)  \operatorname{Tr}[\Lambda_{A}^{1,0}\rho_{1}]\right)
|0\rangle\!\langle0|_{X}\otimes|0\rangle\!\langle0|_{Q}\nonumber\\
&  +\left(  p\operatorname{Tr}[\Lambda_{A}^{0,1}\rho_{0}]+\left(  1-p\right)
\operatorname{Tr}[\Lambda_{A}^{1,1}\rho_{1}]\right)  |0\rangle\!\langle
0|_{X}\otimes|1\rangle\!\langle1|_{Q}\nonumber\\
&  +\left(  p\operatorname{Tr}[\Lambda_{A}^{1,0}\rho_{0}]+\left(  1-p\right)
\operatorname{Tr}[\Lambda_{A}^{0,0}\rho_{1}]\right)  |1\rangle\!\langle
1|_{X}\otimes|0\rangle\!\langle0|_{Q}\nonumber\\
&  +\left(  p\operatorname{Tr}[\Lambda_{A}^{1,1}\rho_{0}]+\left(  1-p\right)
\operatorname{Tr}[\Lambda_{A}^{0,1}\rho_{1}]\right)  |1\rangle\!\langle
1|_{X}\otimes|1\rangle\!\langle1|_{Q},
\end{align}
where $I_{A}=\Lambda_{A}^{0,0}+\Lambda_{A}^{0,1}+\Lambda_{A}^{1,0}+\Lambda
_{A}^{1,1}$ and each $\Lambda_{A}^{i,j}\geq0$. Consider that%
\begin{multline}
\gamma_{XQ}^{(1/r,1/2)}=\frac{1}{2}\left(  1-\frac{r}{2}\right)
|0\rangle\!\langle0|_{X}\otimes|0\rangle\!\langle0|_{Q}+\frac{r}{4}|0\rangle
\langle0|_{X}\otimes|1\rangle\!\langle1|_{Q}\\
+\frac{r}{4}|1\rangle\!\langle1|_{X}\otimes|0\rangle\!\langle0|_{Q}+\frac{1}%
{2}\left(  1-\frac{r}{2}\right)  |1\rangle\!\langle1|_{X}\otimes|1\rangle
\langle1|_{Q}.
\end{multline}
So we find that%
\begin{multline}
\tau_{XQ}-\gamma_{XQ}^{(1/r,1/2)}=\left[  p\operatorname{Tr}[\Lambda_{A}%
^{0,0}\rho_{0}]+\left(  1-p\right)  \operatorname{Tr}[\Lambda_{A}^{1,0}%
\rho_{1}]-\frac{1}{2}\left(  1-\frac{r}{2}\right)  \right]  |0\rangle
\langle0|_{X}\otimes|0\rangle\!\langle0|_{Q}\\
+\left[  p\operatorname{Tr}[\Lambda_{A}^{0,1}\rho_{0}]+\left(  1-p\right)
\operatorname{Tr}[\Lambda_{A}^{1,1}\rho_{1}]-\frac{r}{4}\right]
|0\rangle\!\langle0|_{X}\otimes|1\rangle\!\langle1|_{Q}\\
+\left[  p\operatorname{Tr}[\Lambda_{A}^{1,0}\rho_{0}]+\left(  1-p\right)
\operatorname{Tr}[\Lambda_{A}^{0,0}\rho_{1}]-\frac{r}{4}\right]
|1\rangle\!\langle1|_{X}\otimes|0\rangle\!\langle0|_{Q}\\
+\left[  p\operatorname{Tr}[\Lambda_{A}^{1,1}\rho_{0}]+\left(  1-p\right)
\operatorname{Tr}[\Lambda_{A}^{0,1}\rho_{1}]-\frac{1}{2}\left(  1-\frac{r}%
{2}\right)  \right]  |1\rangle\!\langle1|_{X}\otimes|1\rangle\!\langle1|_{Q}.
\end{multline}
Also, we have that%
\begin{align}
\pi_{1/r}-\sigma^{(1)}\pi_{1/r}\sigma^{(1)} &  =\left(  1-\frac{r}{2}\right)
|0\rangle\!\langle0|+\frac{r}{2}|1\rangle\!\langle1|-\left(  \left(  1-\frac{r}%
{2}\right)  |1\rangle\!\langle1|+\frac{r}{2}|0\rangle\!\langle0|\right)  \\
&  =\left(  1-r\right)  \left(  |0\rangle\!\langle0|-|1\rangle\!\langle1|\right)
.
\end{align}
It thus suffices to take%
\begin{equation}
B_{XQ}=\sum_{i,j\in\left\{  0,1\right\}  }b^{i,j}|i\rangle\!\langle
i|_{X}\otimes|j\rangle\!\langle j|_{Q},\qquad C_{XQ}=\sum_{i,j\in\left\{
0,1\right\}  }c^{i,j}|i\rangle\!\langle i|_{X}\otimes|j\rangle\!\langle j|_{Q},
\end{equation}
with $b^{i,j},c^{i,j}\geq0$. Also, it suffices to take%
\begin{equation}
D_{Q}=\sum_{i\in\left\{  0,1\right\}  }d^{i}|i\rangle\!\langle i|_{Q},\qquad
E_{Q}=\sum_{i\in\left\{  0,1\right\}  }e^{i}|i\rangle\!\langle i|_{Q},
\end{equation}
with $d^{i},e^{i}\geq0$. Then the SDP\ simplifies to the following:%
\begin{equation}
\inf_{\substack{b^{i,j},c^{i,j},\Lambda_{A}^{i,j}\geq0,\\d^{i},e^{i}\geq
0,r\in\left[  0,1\right]  }}\left\{
\begin{array}
[c]{c}%
r:\\
\sum_{i,j\in\left\{  0,1\right\}  }\left(  b^{i,j}+c^{i,j}\right)
\leq\varepsilon\left(  1-\left(  d^{0}+e^{0}+d^{1}+e^{1}\right)  \right)  ,\\
b^{0,0}-c^{0,0}=p\operatorname{Tr}[\Lambda_{A}^{0,0}\rho_{0}]+\left(
1-p\right)  \operatorname{Tr}[\Lambda_{A}^{1,0}\rho_{1}]-\frac{1}{2}\left(
1-\frac{r}{2}\right)  ,\\
b^{0,1}-c^{0,1}=p\operatorname{Tr}[\Lambda_{A}^{0,1}\rho_{0}]+\left(
1-p\right)  \operatorname{Tr}[\Lambda_{A}^{1,1}\rho_{1}]-\frac{r}{4},\\
b^{1,0}-c^{1,0}=p\operatorname{Tr}[\Lambda_{A}^{1,0}\rho_{0}]+\left(
1-p\right)  \operatorname{Tr}[\Lambda_{A}^{0,0}\rho_{1}]-\frac{r}{4},\\
b^{1,1}-c^{1,1}=p\operatorname{Tr}[\Lambda_{A}^{1,1}\rho_{0}]+\left(
1-p\right)  \operatorname{Tr}[\Lambda_{A}^{0,1}\rho_{1}]-\frac{1}{2}\left(
1-\frac{r}{2}\right)  ,\\
d^{0}-e^{0}=\frac{1}{2}\left(  1-r\right)  ,\\
d^{1}-e^{1}=-\frac{1}{2}\left(  1-r\right)  ,\\
I_{A}=\Lambda_{A}^{0,0}+\Lambda_{A}^{0,1}+\Lambda_{A}^{1,0}+\Lambda_{A}^{1,1}%
\end{array}
\right\}
\end{equation}
We can eliminate the $b^{i,j}$ and $d^{i}$ variables to arrive at%
\begin{equation}
\inf_{\substack{c^{i,j},\Lambda_{A}^{i,j}\geq0,\\e^{i}\geq0,r\in\left[
0,1\right]  }}\left\{
\begin{array}
[c]{c}%
r:\\
2\sum_{i,j\in\left\{  0,1\right\}  }c^{i,j}\leq\varepsilon\left(  1-2\left(
e^{0}+e^{1}\right)  \right)  ,\\
p\operatorname{Tr}[\Lambda_{A}^{0,0}\rho_{0}]+\left(  1-p\right)
\operatorname{Tr}[\Lambda_{A}^{1,0}\rho_{1}]-\frac{1}{2}\left(  1-\frac{r}%
{2}\right)  +c^{0,0}\geq0,\\
p\operatorname{Tr}[\Lambda_{A}^{0,1}\rho_{0}]+\left(  1-p\right)
\operatorname{Tr}[\Lambda_{A}^{1,1}\rho_{1}]-\frac{r}{4}+c^{0,1}\geq0,\\
p\operatorname{Tr}[\Lambda_{A}^{1,0}\rho_{0}]+\left(  1-p\right)
\operatorname{Tr}[\Lambda_{A}^{0,0}\rho_{1}]-\frac{r}{4}+c^{1,0}\geq0,\\
p\operatorname{Tr}[\Lambda_{A}^{1,1}\rho_{0}]+\left(  1-p\right)
\operatorname{Tr}[\Lambda_{A}^{0,1}\rho_{1}]-\frac{1}{2}\left(  1-\frac{r}%
{2}\right)  +c^{1,1}\geq0,\\
\frac{1}{2}\left(  1-r\right)  +e^{0}\geq0,\\
-\frac{1}{2}\left(  1-r\right)  +e^{1}\geq0,\\
I_{A}=\Lambda_{A}^{0,0}+\Lambda_{A}^{0,1}+\Lambda_{A}^{1,0}+\Lambda_{A}^{1,1}%
\end{array}
\right\}  .
\end{equation}
Now using the fact that $r\in\left[  0,1\right]  $, this simplifies to%
\begin{equation}
\inf_{\substack{c^{i,j},\Lambda_{A}^{i,j}\geq0,\\e^{i}\geq0,r\in\left[
0,1\right]  }}\left\{
\begin{array}
[c]{c}%
r:\\
2\sum_{i,j\in\left\{  0,1\right\}  }c^{i,j}\leq\varepsilon r,\\
p\operatorname{Tr}[\Lambda_{A}^{0,0}\rho_{0}]+\left(  1-p\right)
\operatorname{Tr}[\Lambda_{A}^{1,0}\rho_{1}]-\frac{1}{2}\left(  1-\frac{r}%
{2}\right)  +c^{0,0}\geq0,\\
p\operatorname{Tr}[\Lambda_{A}^{0,1}\rho_{0}]+\left(  1-p\right)
\operatorname{Tr}[\Lambda_{A}^{1,1}\rho_{1}]-\frac{r}{4}+c^{0,1}\geq0,\\
p\operatorname{Tr}[\Lambda_{A}^{1,0}\rho_{0}]+\left(  1-p\right)
\operatorname{Tr}[\Lambda_{A}^{0,0}\rho_{1}]-\frac{r}{4}+c^{1,0}\geq0,\\
p\operatorname{Tr}[\Lambda_{A}^{1,1}\rho_{0}]+\left(  1-p\right)
\operatorname{Tr}[\Lambda_{A}^{0,1}\rho_{1}]-\frac{1}{2}\left(  1-\frac{r}%
{2}\right)  +c^{1,1}\geq0,\\
I_{A}=\Lambda_{A}^{0,0}+\Lambda_{A}^{0,1}+\Lambda_{A}^{1,0}+\Lambda_{A}^{1,1}%
\end{array}
\right\}  .
\end{equation}
This concludes the proof.

\section{Optimizations for approximate one-shot SD-cost}

\label{app:opts-approx-cost-not-SDPs}

In this appendix, we detail the optimization problem for approximate one-shot SD-cost under CPTP$_A$ and CDS maps. We prove that both of these operational quantities can be calculated by means of bilinear programs.

Consider from Lemma~\ref{lem:ApproxDilEquality} that the approximate one-shot SD-cost is equal to%
\begin{equation}
\xi_{c}^{\text{FO},q,\varepsilon}(\rho_{XA})=\inf_{\widetilde{\rho}_{XA}\in B_{\varepsilon}^{\prime
}(\rho_{XA})}\xi_{c}^{\text{FO},q}(\widetilde{\rho}_{XA}),
\end{equation}
where $\xi_{c}^{\text{FO},q}(\widetilde{\rho}_{XA})$ is the exact cost. Under
CPTP$_{A}$ maps, we know from Theorem~\ref{theo-diluteCPTP} that%
\begin{equation}
\xi_{c}^{\text{CPTP}_{A},p}(\widetilde{\rho}_{XA})=\log\inf_{M\geq1}\left\{
M:\widetilde{\rho}_{0}\leq\left(  2M-1\right)  \widetilde{\rho}_{1}%
,\ \widetilde{\rho}_{1}\leq\left(  2M-1\right)  \widetilde{\rho}_{0}\right\}
.
\end{equation}
Combining with the SDP\ for the minimum conversion error (from Proposition~\ref{prop:SDP-conv-err}), we conclude that%
\begin{align}
\xi_{c}^{\varepsilon}(\rho_{XA})  & =\xi_{c}^{\text{CPTP}_{A},p,\varepsilon
}(\rho_{XA})\label{eq:cost-CPTP-attempt-at-SDP}\\
& =\log\inf_{\substack{M\geq1,B_{XA},C_{XA},\\D_{A},E_{A}\geq0,s\in
\mathbb{R},\\\tilde{p}\in\left[  0,1\right]  ,\widetilde{\rho}_{0}%
,\widetilde{\rho}_{1}\geq0}}\left\{
\begin{array}
[c]{c}%
M:\\
\widetilde{\rho}_{0}\leq\left(  2M-1\right)  \widetilde{\rho}_{1},\\
\widetilde{\rho}_{1}\leq\left(  2M-1\right)  \widetilde{\rho}_{0},\\
\operatorname{Tr}[B_{XA}+C_{XA}]\leq\varepsilon,\\
s-1\geq\operatorname{Tr}[D_{A}+E_{A}],\\
B_{XA}-C_{XA}=s\left(  \rho_{XA}-\widetilde{\rho}_{XA}\right)  ,\\
D_{A}-E_{A}=s\left(  \tilde{p}\widetilde{\rho}_{0}-\left(  1-\tilde{p}\right)
\widetilde{\rho}_{1}\right)  ,\\
\widetilde{\rho}_{XA}=|0\rangle\!\langle0|\otimes\tilde{p}\widetilde{\rho}%
_{0}+|1\rangle\!\langle1|\otimes\left(  1-\tilde{p}\right)  \widetilde{\rho}%
_{1},\\
\operatorname{Tr}[\widetilde{\rho}_{0}]=\operatorname{Tr}[\widetilde{\rho}%
_{1}]=1
\end{array}
\right\}  \\
& =\log\inf_{\substack{M\geq1,B_{XA},C_{XA},\\D_{A},E_{A}\geq0,s\geq 1,\\\tilde{p}\in\left[  0,1\right]  ,\widetilde{\rho}_{0}%
,\widetilde{\rho}_{1}\geq0}}\left\{
\begin{array}
[c]{c}%
M:\\
\widetilde{\rho}_{0}\leq\left(  2M-1\right)  \widetilde{\rho}_{1},\\
\widetilde{\rho}_{1}\leq\left(  2M-1\right)  \widetilde{\rho}_{0},\\
\operatorname{Tr}[B_{XA}+C_{XA}]\leq\varepsilon,\\
s-1\geq\operatorname{Tr}[D_{A}+E_{A}],\\
B_{XA}-C_{XA}=s  \rho_{XA}-\widetilde{\rho}_{XA}  ,\\
D_{A}-E_{A}=  \tilde{p}\widetilde{\rho}_{0}-\left(  1-\tilde{p}\right)
\widetilde{\rho}_{1}  ,\\
\widetilde{\rho}_{XA}=|0\rangle\!\langle0|\otimes\tilde{p}\widetilde{\rho}%
_{0}+|1\rangle\!\langle1|\otimes\left(  1-\tilde{p}\right)  \widetilde{\rho}%
_{1},\\
\operatorname{Tr}[\widetilde{\rho}_{0}]=\operatorname{Tr}[\widetilde{\rho}%
_{1}]=s
\end{array}
\right\}  .
\end{align}
This is a bilinear program, due to terms like $\left(  2M-1\right)  \widetilde{\rho}_{1}$, $\left(  2M-1\right)  \widetilde{\rho}_{0}$, $\tilde{p}\widetilde{\rho}
_{0}$, and $\left(  1-\tilde{p}\right)  \widetilde{\rho}
_{1}$ appearing in the optimization.

Under CDS maps, we know from Theorem~\ref{theo-diluteCDS} that%
\begin{equation}
\xi_{c}^{\text{CDS},\frac{1}{2}}(\widetilde{\rho}_{XA})=\log\inf_{M\geq
1}\left\{
\begin{array}
[c]{c}%
M:\tilde{p}\widetilde{\rho}_{0}\leq\left(  2M-1\right)  \left(  1-\tilde
{p}\right)  \widetilde{\rho}_{1},\\
\left(  1-\tilde{p}\right)  \widetilde{\rho}_{1}\leq\left(  2M-1\right)
\tilde{p}\widetilde{\rho}_{0}%
\end{array}
\right\}  .
\end{equation}
Combining with the SDP\ for the minimum conversion error (from Proposition~\ref{prop:SDP-conv-err}), we conclude that%
\begin{align}
\xi_{c}^{\star,\varepsilon}(\rho_{XA})  & =\xi_{c}^{\text{CDS},\frac{1}%
{2},\varepsilon}(\rho_{XA})\\
& =\log\inf_{\substack{M\geq1,B_{XA},C_{XA},\\D_{A},E_{A}\geq0,s\in
\mathbb{R},\\\tilde{p}\in\left[  0,1\right]  ,\widetilde{\rho}_{0}%
,\widetilde{\rho}_{1}\geq0}}\left\{
\begin{array}
[c]{c}%
M:\\
\tilde{p}\widetilde{\rho}_{0}\leq\left(  2M-1\right)  \left(  1-\tilde
{p}\right)  \widetilde{\rho}_{1},\\
\left(  1-\tilde{p}\right)  \widetilde{\rho}_{1}\leq\left(  2M-1\right)
\tilde{p}\widetilde{\rho}_{0},\\
\operatorname{Tr}[B_{XA}+C_{XA}]\leq\varepsilon,\\
s-1\geq\operatorname{Tr}[D_{A}+E_{A}],\\
B_{XA}-C_{XA}=s\left(  \rho_{XA}-\widetilde{\rho}_{XA}\right)  ,\\
D_{A}-E_{A}=s\left(  \tilde{p}\widetilde{\rho}_{0}-\left(  1-\tilde{p}\right)
\widetilde{\rho}_{1}\right)  ,\\
\widetilde{\rho}_{XA}=|0\rangle\!\langle0|\otimes\tilde{p}\widetilde{\rho}%
_{0}+|1\rangle\!\langle1|\otimes\left(  1-\tilde{p}\right)  \widetilde{\rho}%
_{1},\\
\operatorname{Tr}[\widetilde{\rho}_{0}]=\operatorname{Tr}[\widetilde{\rho}%
_{1}]=1
\end{array}
\right\}  \\
& =\log\inf_{\substack{M\geq1,B_{XA},C_{XA},\\D_{A},E_{A}\geq0,s\in
\mathbb{R},\\\widehat{\rho}_{0},\widehat{\rho}_{1}\geq0}}\left\{
\begin{array}
[c]{c}%
M:\\
\widehat{\rho}_{0}\leq\left(  2M-1\right)  \widehat{\rho}_{1},\\
\widehat{\rho}_{1}\leq\left(  2M-1\right)  \widehat{\rho}_{0},\\
\operatorname{Tr}[B_{XA}+C_{XA}]\leq\varepsilon,\\
s-1\geq\operatorname{Tr}[D_{A}+E_{A}],\\
B_{XA}-C_{XA}=s\left(  \rho_{XA}-\widetilde{\rho}_{XA}\right)  ,\\
D_{A}-E_{A}=s\left(  \widehat{\rho}_{0}-\widehat{\rho}_{1}\right)  ,\\
\widetilde{\rho}_{XA}=|0\rangle\!\langle0|\otimes\widehat{\rho}_{0}%
+|1\rangle\!\langle1|\otimes\widehat{\rho}_{1},\\
\operatorname{Tr}[\widehat{\rho}_{0}]+\operatorname{Tr}[\widehat{\rho}_{1}]=1
\end{array}
\right\}  \\
& =\log\inf_{\substack{M\geq1,B_{XA},C_{XA},\\D_{A},E_{A}\geq0,s\geq 0,\\\widehat{\rho}_{0},\widehat{\rho}_{1}\geq0}}\left\{
\begin{array}
[c]{c}%
M:\\
\widehat{\rho}_{0}\leq\left(  2M-1\right)  \widehat{\rho}_{1},\\
\widehat{\rho}_{1}\leq\left(  2M-1\right)  \widehat{\rho}_{0},\\
\operatorname{Tr}[B_{XA}+C_{XA}]\leq\varepsilon,\\
s-1\geq\operatorname{Tr}[D_{A}+E_{A}],\\
B_{XA}-C_{XA}= s \rho_{XA}-\widetilde{\rho}_{XA}  ,\\
D_{A}-E_{A}=  \widehat{\rho}_{0}-\widehat{\rho}_{1}  ,\\
\widetilde{\rho}_{XA}=|0\rangle\!\langle0|\otimes\widehat{\rho}_{0}%
+|1\rangle\!\langle1|\otimes\widehat{\rho}_{1},\\
\operatorname{Tr}[\widehat{\rho}_{0}]+\operatorname{Tr}[\widehat{\rho}_{1}]=s
\end{array}
\right\}  .\label{eq:cost-CDS-attempt-at-SDP}%
\end{align}
This is also a bilinear program, due to terms like
 $\left(  2M-1\right)  \widetilde{\rho}_{1}$ and
$\left(  2M-1\right)  \widehat{\rho}_{1}$ appearing in the optimization.

\bibliography{Ref}
\bibliographystyle{alpha}

\end{document}